\documentclass[11pt]{article}
\usepackage{authblk}
\author[1]{Yahya Sattar}
\author[2]{Yassir Jedra}
\author[3]{Robin Str\"asser}
\author[3]{Frank Allg\"ower}
\author[4,5]{\\Maryam Fazel}
\author[1]{Sarah Dean}

\affil[1]{Cornell University}
\affil[2]{Imperial College London}
\affil[3]{University of Stuttgart}
\affil[4]{University of Washington Seattle}
\affil[5]{Amazon Inc.}

\usepackage{ifthen}
\newcommand{\version}{arxiv} 
\usepackage[utf8]{inputenc} \usepackage[T1]{fontenc} 

\usepackage{fullpage}

\usepackage{url}            \usepackage{booktabs}       \usepackage{amsfonts}       \usepackage{nicefrac}       \usepackage{microtype}      \usepackage[dvipsnames]{xcolor}

\usepackage{natbib} 

\usepackage{enumitem}
\usepackage{amsthm}
\usepackage{amsmath}
\numberwithin{equation}{section}

\usepackage{thmtools}
\usepackage{thm-restate}
\usepackage{mdframed}

\usepackage{algorithm, algorithmicx,algpseudocode}
\usepackage{subcaption}

\usepackage{adjustbox}
\usepackage{tikz}
\usetikzlibrary{arrows.meta, positioning, fit, backgrounds, calc}
\usepackage{pgfplots}

\newcommand{\diag}{\textnormal{diag}}

\newcommand{\E}{\mathbb{E}}

\DeclareMathOperator*{\argmin}{argmin}

\DeclareMathOperator{\polylog}{polylog}

\DeclareMathOperator{\tr}{tr}

\DeclareMathOperator{\vecop}{\mathsf{vec}}
\DeclareMathOperator{\matop}{\mathsf{mtx}}

\newcommand{\T}{\top}

\newcommand{\EE}{\mathbb{E}}
\newcommand{\PP}{\mathbb{P}}
\newcommand{\RR}{\mathbb{R}}
\renewcommand{\P}{\mathbb{P}}
\newcommand{\util}{\tilde{u}}

\newcommand{\cN}{\mathcal{N}}
\newcommand{\cF}{\mathcal{F}}

\newcommand{\cE}{\mathcal{E}}

\newcommand{\cS}{\mathcal{S}}

\newcommand{\cU}{\mathcal{U}}
\newcommand{\cO}{\mathcal{O}}

\newcommand{\R}{\mathbb{R}}

\renewcommand{\Pr}{\mathbb{P}}

\newcommand{\norm}[1]{\lVert #1 \rVert}
\newcommand{\bignorm}[1]{\left\lVert #1 \right\rVert}

\newcommand{\Z}{\mathbb{Z}}

\newtheorem{definition}{Definition}[section] \newtheorem{theorem}{Theorem}[section] \newtheorem{proposition}{Proposition}[section]  \newtheorem{assumption}{Assumption}[section]  \newtheorem{remark}{Remark}[section]
\newtheorem{lemma}{Lemma}[section]
\newtheorem{example}{Example}[section]

\usepackage{tcolorbox}
\usepackage{xcolor}

\newcommand{\op}{\textup{op}}

\newcommand{\ubar}{\bar{u}}
\newcommand{\nn}{\nonumber}
\newcommand{\Gbar}{\bar{G}}

\newcommand{\Acal}{\mathcal{A}}
\newcommand{\Dcal}{\mathcal{D}}
\newcommand{\Scal}{\mathcal{S}}
\newcommand{\Ecal}{\mathcal{E}}
\newcommand{\Fcal}{\mathcal{F}}
\newcommand{\Mcal}{\mathcal{M}}
\newcommand{\Ucal}{\mathcal{U}}

\newcommand{\Ghat}{\widehat{G}}

\newcommand{\xtil}{\tilde{x}}

\newcommand{\tf}[1]{\left\|{#1}\right\|_{F}}
\newcommand{\tn}[1]{\left\|{#1}\right\|_{\ell_2}}

\newcommand{\opn}[1]{\left\|{#1}\right\|_{\op}}
\newcommand{\bgl}{{~\big |~}}
\newcommand{\distas}{\overset{\text{i.i.d.}}{\sim}}

\newcommand{\leqsym}[1]{\stackrel{\text{(#1)}}{\leq}}

\newcommand{\eqsym}[1]{\stackrel{\text{(#1)}}{=}}

\newcommand{\li}{\left<}
\newcommand{\ri}{\right>}
\usepackage{stackengine}[2013-10-15]
\usepackage{scalerel}
\newcommand\mysqrt[2][0pt]{\stretchrel{\sqrt{}}{\addstackgap [#1]{$\displaystyle\overline{#2}$}}}

\newcommand{\Ocal}{\mathcal{O}}
\newcommand{\vxhat}{\hat{\vx}}

\newcommand{\vu}{\bv{u}}

\newcommand{\vw}{\bv{w}}
\newcommand{\vx}{\bv{x}}
\newcommand{\vy}{\bv{y}}
\newcommand{\vz}{\bv{z}}

\newcommand{\Bcal}{\mathcal{B}}

\newcommand{\Gcal}{\mathcal{G}}

\newcommand{\Ical}{\mathcal{I}}

\newcommand{\Ncal}{\mathcal{N}}

\newcommand{\vA}{\bv{A}}
\newcommand{\vB}{\bv{B}}
\newcommand{\vC}{\bv{C}}

\newcommand{\vF}{\bv{F}}

\newcommand{\vK}{\bv{K}}
\newcommand{\vL}{\bv{L}}
\newcommand{\vM}{\bv{M}}
\newcommand{\vN}{\bv{N}}

\newcommand{\vP}{\bv{P}}
\newcommand{\vQ}{\bv{Q}}
\newcommand{\vR}{\bv{R}}

\newcommand{\Ab}{\bv{A}}

\newcommand{\vSigma}{\bvgrk{\Sigma}} 
\usepackage{amssymb}
\usepackage{mathrsfs}
\usepackage{hyperref}
\hypersetup{
pdftoolbar=true,        pdfmenubar=true,        pdffitwindow=false,     pdfstartview={FitH},    pdfnewwindow=true,      colorlinks=true,       linkcolor=blue,          citecolor=ForestGreen,        filecolor=magenta,      urlcolor=red,           breaklinks=false,
}
\usepackage{cleveref}

\date{}

\title{Learning and Control Beyond Linearity:\\ Towards a Non-asymptotic Theory for Bilinear Systems }
\begin{document}
\maketitle

\begin{abstract}
This tutorial provides a unified view of the emerging area of bilinear learning and control. 
Using linear systems as a benchmark, it explains what fundamentally changes in the bilinear settings, how recent theory addresses finite-sample learning and control, and how these ideas connect to broader themes in nonlinear control, representation
learning, and data-driven decision making. 
For learning, we emphasize tools that are particularly useful in the bilinear settings, such as one-sided Bernstein's inequality for dependent and heavy-tailed covariates, blocking arguments, and martingale concentration for input-dependent noise. 
We then apply these tools to obtain finite-sample learning guarantees for fully observed bilinear systems, partially observed bilinear systems, and linear systems with bilinear observations. 
For control, we discuss quadratic control from bilinear observations, where the classical separation principle fails, and review tractable approaches based on belief-space receding horizon control. 
We also cover stabilization of bilinear dynamics under state feedback using semi-definite programming, LMI relaxations, sum-of-squares methods, and Koopman-based lifting. 
We conclude by discussing connections to reinforcement learning and machine learning, and some open problems in combined learning and control of bilinear systems. 
\ifthenelse{\equal{\version}{cdc}}{
Due to limited space, proofs have
been omitted in this version and are available in an online version~\cite{sattar2026learning}: \url{https://sdean.website/bilinear_tutorial.html}
}{}
\end{abstract}
 \tableofcontents
\newpage

\section{Introduction}\label{sec:intro}
Learning and control of linear dynamical systems (time-invariant or time-varying) form the backbone of modern control theory. 
The recent surge of interest in data-driven
control and non-asymptotic analysis of linear dynamical systems~\citep{dean2019sample,simchowitz2018learning,oymak2021revisiting,sarkar2019near,tsiamis2019finite}, motivated by problems such as robot control in unknown environments, has also provided a key benchmark for reinforcement learning with continuous state and action spaces. This line of work has produced powerful tools for non-asymptotic system identification~\citep{dean2019sample,simchowitz2018learning,oymak2021revisiting,sun2022finite}, uncertainty quantification~\citep{mania2019certainty}, regret analysis~\citep{abbasi2011regret,lale2020explore,abeille2020efficient}, and controller synthesis~\citep{abbasi2019model,faradonbeh2020optimism,hazan2020nonstochastic}. However, much of what makes the linear setting analytically tractable depends on structural simplifications such as easy exploration, decoupling of estimation and control via the separation principle, and closed-form optimal controllers. 
These simplifications often fail in real-world systems. 
As a result, while the linear case remains foundational, it offers only a limited view of the challenges that arise in learning and control beyond linearity.

Bilinear dynamical systems provide a mathematically structured yet expressive framework that bridges classical linear systems and fully nonlinear dynamics. 
They arise naturally in a variety of domains from engineering, biology~\citep{bilinearbook}, quantum mechanical processes~\citep{pardalos2010optimization}, recommendation systems~\citep{koren2021advances} to sequence modeling~\citep{gu2023mamba}.
Moreover, bilinear systems approximate a much broader class of nonlinear systems via Carleman linearization~\citep{kowalski1991nonlinear} or Koopman canonical transform~\citep{surana2016koopman,goswami2017bilinear, bruder2021advantages,strasser2026overview} of control-affine nonlinear systems~\citep{svoronos1980bilinear,Lo1975bilinear}.
Bilinear systems exhibit several phenomena that do not appear in the linear setting. 
In particular, fundamental system properties such as stability, stabilizability, controllability, and observability can depend on the input sequence~\citep{sattar2022finite, sattar2025finite}. 

We will consider several canonical bilinear models, including fully observed bilinear state-space systems~\citep{sattar2022finite}, partially observed bilinear systems~\citep{sattar2025finite}, and linear systems with bilinear observation models~\citep{sattar2024learning,choi2025explore}. 
Across these formulations, we will provide a self-contained exposition of non-asymptotic learning~\citep{sattar2022finite,sattar2025finite,sattar2024learning,choi2025explore,chatzikiriakos2026endtoend} and control~\citep{chatzikiriakos2026endtoend,sattar2025sub,cao2026dual,strasser2023robust,strasser2023control,strasser2025koopman,strasser2026safedmd,strasser2025sos,strasser2025performance}. 
It is our ambition that this tutorial serves as a comprehensive reference for new researchers as to what the state-of-the-art techniques are in the area of non-asymptotic analysis of bilinear systems. 
In particular, we aim to provide a self-contained exposition reviewing the main results along with the most important proof ideas. An important goal is to make the presentation accessible to a wider audience, without requiring prior background on bilinear systems theory. 
Lastly, we strive to establish a sound theoretical understanding of learning and control problems beyond linear dynamical systems. 
We aim to achieve this by revisiting the least squares theory for bilinear systems, introducing new tools, such as one-sided concentration approaches to persistence of excitation for heavy-tailed covariates, and input-dependent Martingale bounds for controlling error terms. 
Besides technical goals, this tutorial aims to establish a common language between control theorists, machine learning theorists, and statisticians; to push the boundary of non-asymptotic learning and control beyond linear dynamics.

\subsection{Notation}

Maxima (resp.\ minima) of two numbers $a,b\in \R$ are denoted by $a\vee b =\max(a,b)$ ($a\wedge b = \min(a,b)$). For two sequences $\{a_t\}_{t\in \Z}$ and $\{b_t\}_{t\in \Z}$, we introduce the shorthand $a_t \lesssim b_t$ if there exists a universal constant $C>0$ and an integer $t_0$ such that $a_t \leq C b_t$ for every $t \geq t_0$. 
For an integer $N$, we define the shorthand $[N] := \{1,\dots,N\}$. 
Expectation (resp.\ probability) with respect to all the randomness of the underlying probability space is denoted by $\EE$ (resp.\ $\PP$).

The Euclidean norm on $\mathbb{R}^{d}$ is denoted $\|\cdot\|_{\ell_2}$,
and the unit sphere in $\R^d$ is denoted $\mathcal{S}^{d-1}$. 
The standard inner product on $\R^{d}$ is denoted $\langle\cdot,\cdot\rangle$. $\rho(X)$, $\|X\|_\op$, and $\tf{X}$ denote the spectral radius, spectral norm, and Frobenius norm of a matrix $X {\in} \R^{d \times d}$, respectively. $(v)_i$ denotes the $i$-th element of a vector $v {\in  \R^d}$. 
For a positive definite matrix $M \in \R^{d \times d}$, the Mahalanobis norm of a vector $v \in \R^d$ is given by $\norm{v}_{M} = \sqrt{v^\T M v}$. 
We use $\vecop(\cdot): \R^{d_1 \times d_2} \mapsto \R^{d_1 d_2}$ to denote the vectorization operator, which transforms a matrix into a vector, whereas $\matop(\cdot): \R^{d_1 d_2} \mapsto \R^{d_1 \times d_2}$ denotes its inverse operator, which transforms a vector into a matrix. The transpose of a matrix $M$ is denoted by $M^\T$ and $\tr M $ denotes its trace. 
For a matrix $M \in \R^{d_1 \times d_2}$, we order its singular values $\sigma_{1}(M),\dots,\sigma_{d_1 \wedge d_2}(M)$ in descending order by magnitude. 
We also write $\opn{M}$ for its largest singular value $\opn{M} \triangleq \sigma_1(M)$. 
To not carry dimensional notation, we will also use $\sigma_{\min}(M)$ for the smallest nonzero singular value.
For square matrices $M\in \R^{d\times d}$ with real eigenvalues, we similarly order the eigenvalues of $M$ in descending order as $\lambda_{1}(M),\dots,\lambda_{d}(M)$. 
In this case, $\lambda_{\min}(M)$ will also be used to denote the minimum (possibly zero) eigenvalue of $M$.
For two symmetric matrices $M, N$, we write $M \succ N$ ($M\succeq N)$ if $M-N$ is positive (semi-)definite. 
Given a vector $u \in \RR^{d_u}$ and matrices $M_0, M_1, \dots, M_{d_u} \in \RR^{d_1 \times d_2}$, we write $u \circ M := M_0 + \sum_{k=1}^{d_u} u_k M_k$.
We use $\tilde{\Ocal}(\cdot)$ and $\tilde{\Omega}(\cdot)$ to hide constants and logarithmic factors.  
Finally, $\otimes$ denote the Kronecker product.

\subsection{Setting}\label{subsec:setting}

In this tutorial, we consider the learning and control of dynamical systems that are governed by the state-space equations of the form
\begin{equation}\label{eqn:general_dynamics}
    \begin{aligned} 
        x_{t+1} &= f_{\theta}(x_t, u_t) + w_t, \\
        y_t &= g_{\theta}(x_t, u_t) + v_t,
    \end{aligned}
\end{equation}
where $x_t\in \R^{d_x}$, $u_t \in \R^{d_u}$, $y_t \in \R^{d_y}$, $w_t \in \R^{d_x}$, and $v_t \in \R^{d_y}$ denote the state, input, output, process noise, and measurement noise at time $t \geq 0$, respectively. The functions $(f_\theta,g_\theta)$ are parameterized by $\theta \in \R^{d_{\theta}}$ which may be unknown. 
The goal of this tutorial is to learn and control the dynamical systems of the form~\eqref{eqn:general_dynamics} from a single input-output trajectory $\lbrace(u_t,y_t)\rbrace_{t=0}^T$. 
Without loss of generality, we consider $x_0 = 0$ throughout the manuscript. 
The noise processes $\lbrace w_t \rbrace_{t\ge0}$ and $\lbrace v_t \rbrace_{t\ge 0}$ are assumed to be sequences of independent, zero-mean, $\sigma^2$-sub-Gaussian random vectors taking values in $\R^{d_x}$ and $\R^{d_y}$, respectively, for some variance proxy parameter $\sigma > 0$ (i.e., $\E[w_t]= 0$, and $\E[\exp(\li z,w_t \ri)] \leq \exp(\sigma^2 \tn{z}^2/2)$ for all $z \in \R^{d_x}$).

Throughout this tutorial, we focus on dynamical systems of the form \eqref{eqn:general_dynamics}, where at least one of the state-space functions $(f_\theta,g_\theta)$ is bilinear in $(x_t,u_t)$. 
Bilinearity arises when there is inherent interaction, not merely linear superposition, between the state and input.
Bilinear dynamics often arise from physical conservation laws when the control input corresponds to a flow parameter.
\begin{example}[Variable conductance]
    Consider a simple RC circuit with a capacitor $C$ connected in series to a variable resistor $R(t)$ and a fixed voltage source $v_\mathrm{in}$. 
    Let the capacitor voltage be the state $x(t)$ and defined the input  $u(t)=1/R(t)$ as variable conductance. 
    Then, Kirchhoff's Current Law states that
$C\dot x(t)
=
u(t)(v_{\mathrm{in}}-x(t))$. 
The forward Euler discretizations with step size $h$ results in the bilinear dynamics
\begin{equation*}
x_{t+1}
=
 x_t -\frac{h}{C} u_t x_t + \frac{hv_{\mathrm{in}}}{C} u_t.
\end{equation*}
\end{example}
Bilinear state evolution also appears broadly beyond physical dynamical systems.
\begin{example}[HMM posterior]
    Posterior computations in hidden Markov models (HMMs) can be written
as bilinear recursions. 
Consider an HMM with $S$ discrete states and $O$ possible observations, with state transition matrix $T\in\mathbb R^{S\times S}$ and emission matrix $E\in\mathbb R^{S\times O}$.
Let $x_t\in\mathbb R^S$ denote the vector of unnormalized forward probabilities $(x_t)_i = p(s_t=i,o_{0:t})$, 
let $y_t\in\mathbb R^O$ denote the vector of unnormalized observation probabilities $(y_{t})_o = \Pr(o_{t+1}=o, o_{0:t})$,
and let $u_t=e_{o_t}\in\mathbb R^O$ be a one-hot encoding of the observation
at time $t$. 
The forward recursion is bilinear:
\begin{equation*}
    x_{t+1} = \sum_{k=1}^{O} (u_{t+1})_k
        T\operatorname{diag}(E_{:,k})x_t,
    \quad 
    y_{t}
    =
    E^\top T x_t\:.
\end{equation*}
\end{example}
These example both motivate the general form for a multi-input-multi-output bilinear dynamical system.
\begin{definition}[Bilinear dynamical systems] 
A dynamical system of particular interest to us which is a special case of \eqref{eqn:general_dynamics} is a partially-observed bilinear dynamical systems with linear state observation (BDS-PO). 
It has the state-space representation
\begin{equation} \label{eqn:BDS-PO}
    \begin{aligned}
         x_{t+1} &= A_0 x_t + \sum_{k=1}^{d_u} (u_t)_k A_k x_t + B u_t +  w_t, \\
         y_t &= C x_t + D u_t + v_t,
    \end{aligned}
    \tag{BDS-PO}
\end{equation}
where $A_0, A_1, \dots, A_{d_u}$ denote the $d_u+1$ state matrices taking values in $\R^{d_x \times d_x}$, $B \in \R^{d_x \times d_u}$ denotes the input matrix, $C \in \R^{d_y \times d_x}$, and $D \in \R^{d_y \times d_u}$ denote the observation/measurement matrices.
All of these matrices constitute the parameter $\theta=(A_0,...,A_{d_u},B,C,D)$.
We are interested both in learning the input-output behavior of unknown bilinear systems from a single trajectory of input-output samples $\lbrace(u_t,y_t)\rbrace_{t=0}^T$
and in feedback control design.
\end{definition}

We may also expect bilinearity to arise in the measurement equation.
In general, bilinear observations occur when the measurement involves an interaction between the state and the input. 
\begin{example}[Thermal observation]
    Consider thermal observation of a circuit whose state contains the current $I_t$, as it can arise from parasitic inductance, and suppose the applied
voltage $u_t$ is the input. 
The electrical power $I_t u_t$ is dissipated as heat, so a thermal measurement may be modeled as proportional to the dissipated power, i.e.,
\begin{equation*}
y_t =\alpha I_t u_t = C x_t u_t,
\end{equation*}
where $C$ both selects $I_t$ from the overall state $x_t$ and includes the constant $\alpha$.
This thermal observation is bilinear in the state and input.
\end{example}

Bilinear observations may also occur more broadly, from personalized recommendation systems to quantum mechanical systems.
\begin{example}[Preference dynamics]
    The simplest model of personalized recommendation posits that $y=u^\top x+$noise, where $y$ is an observed preference signal like a rating, e.g., 1-5 stars, $u$ is a latent factor for a recommended item, e.g., a movie, and $x$ is a latent factor for a user to whom the item was recommended.
Supposing that user preferences are actually dynamic gives rise to a dynamical system of the form  \eqref{eqn:general_dynamics} in which the measurement equation is bilinear.
\end{example}
 
\newpage
\part{Learning}\label{part:learning}
\section{The Task of Learning Bilinear Systems}

In this section, we consider the task of learning the input-output behavior of unknown bilinear systems.
We first cast the problem as a linear regression problem with nonlinear features, approximation bias, and intricate dependencies. 
We then present the least-squares solution with decompositions into relevant error terms and give an overview of the prototypical analysis strategy.
Finally, we conclude with
notes discussing existing ideas, alternative decompositions, and comparison with linear systems. 

\subsection{Problem Formulation} \label{subsec:bilinear_feature_map}

Our objective is to identify the parameter $\theta$, describing the state-space functions $(f_\theta,g_\theta)$ which can be in particular \emph{bilinear} in $(x_t, u_t)$. 
The approach we aim to present, and which naturally applies to various families of bilinear systems, consists of reducing the dynamics to a linear form with an approximation bias of the form
\begin{align}\label{eqn:feature_map}
    \forall~ t \ge 0, \qquad y_t = G\,\phi_t + z_t + \eta_t, 
\end{align}
where $\{y_t\}_{t \geq 0}$ is a sequence of outputs (when partially observed) or states (when fully observed) taking values in $\R^{d_y}$. 
$\{\phi_t\}_{t\geq 0}$ is a sequence of dependent, nonlinear features, constructed from inputs and/or states, taking values in $\R^{d_\phi}$. 
$\{z_t\}_{t \geq 0}$ is a sequence of approximation (or truncation) biases, that are intricately dependent on past states and inputs, taking values in $\R^{d_y}$. 
$\{\eta_t\}_{t \geq 0}$ is a sequence of dependent, correlated, and possibly heavy-tailed noise processes, taking values in $\R^{d_y}$.  
$G \in \RR^{d_y \times d_\phi}$ is the feature-to-output matrix and is a-priori unknown.

\begin{remark}
    The matrix $G$ encodes information about the unknown parameter $\theta$ which, as we will see, is possible to extract using subspace methods such as Ho-Kalman \citep{ho1966effective, oymak2021revisiting}. 
    Therefore, we will mainly focus on the task of identifying $G$ and the challenges that come with that for bilinear systems.
\end{remark}

\begin{remark}
    The linear form \eqref{eqn:feature_map} emerges from expressing $y_t$ using past states and inputs. This is achieved by unrolling the dynamics up to a truncation length $\tau>0$ as
    \ifthenelse{\equal{\version}{arxiv}}{
    \begin{align*}
        y_t &= g_\theta(f_\theta(f_\theta( \cdots (f_\theta( x_{t-\tau}, u_{t- \tau})+ w_{t-\tau},u_{t-\tau +1}) \cdots  ) + w_{t-2}, u_{t-1})+ w_{t-1}, u_t) + v_t. 
    \end{align*}
    }{
    \begin{align*}
        y_t &= g_\theta(f_\theta(f_\theta( \cdots (f_\theta( x_{t-\tau}, u_{t- \tau})+ w_{t-\tau},u_{t-\tau +1}) \cdots  ) \nn \\
        &\qquad\qquad\qquad\qquad\quad+ w_{t-2}, u_{t-1})+ w_{t-1}, u_t) + v_t.
    \end{align*}
    }
\end{remark}

\begin{remark} The reader may notice that the form \eqref{eqn:feature_map} is quite similar to that often appearing in the study of \emph{linear} dynamical systems~\citep{ljung1999system, simchowitz2018learning, faradonbeh2018finite, sarkar2019near, sarkar2021finite, tsiamis2019finite, oymak2018non, ziemann2023tutorial}. 
The key difference lies in the statistical dependencies that arise between the outputs $\{y_t\}_{t \geq 0}$, features $\{\phi_t\}_{t\geq 0}$, approximation biases $\{z_t\}_{t \geq 0}$, and the noise sequence $\{\eta_t\}_{t \geq 0}$. As we shall see, these dependencies pose non-trivial challenges for the analysis of identification procedures for bilinear systems.
\end{remark}

To illustrate the intricate dependencies arising when performing a reduction of bilinear systems to a linear form with approximation bias, we consider the example of bilinear dynamical systems with partial linear observations \eqref{eqn:BDS-PO}.

\begin{example}[Bilinear dynamical systems]
Fix a history length $\tau>0$. It can be shown that the state-space representation \eqref{eqn:BDS-PO} can be alternately represented as the non-linear input-output map
\ifthenelse{\equal{\version}{arxiv}}{
\begin{equation}
\begin{aligned}\label{eqn:feature_map_BDS_PO}
    y_t &=  G \,\begin{bmatrix} u_{t} \\
 u_{t-1} \\ 
 \util_{t-1} \otimes u_{t-2} \\ 
 \util_{t-1} \otimes \util_{t-2} \otimes u_{t-3} \\ \vdots \\ 
 \util_{t-1} \otimes \util_{t-2} \otimes \cdots \otimes u_{t-\tau+1} 
 \end{bmatrix} +  C  \left( \prod_{\ell = 1}^{\tau-1} (u_{t-\ell}\circ A )  \right) x_{t-\tau+1} +  F_t \, \omega_t, 
\end{aligned}
\end{equation}
}{
\begin{equation}
\begin{aligned}\label{eqn:feature_map_BDS_PO}
    y_t &=  G \,\begin{bmatrix} u_{t} \\
 u_{t-1} \\ 
 \util_{t-1} \otimes u_{t-2} \\ 
 \util_{t-1} \otimes \util_{t-2} \otimes u_{t-3} \\ \vdots \\ 
 \util_{t-1} \otimes \util_{t-2} \otimes \cdots \otimes u_{t-\tau+1} 
 \end{bmatrix}  \\
 &+  C  \left( \prod_{\ell = 1}^{\tau-1} (u_{t-\ell}\circ A )  \right) x_{t-\tau+1} +  F_t \, \omega_t, 
\end{aligned}
\end{equation}
}
where we define,
\begin{align}
    \util_t := \begin{bmatrix} 1 \\ u_t \end{bmatrix}, \quad \omega_t &:= \begin{bmatrix} v_t^\T & w_{t-1}^\T & \cdots & 
  w_{t-\tau+1}^\T \end{bmatrix}^\T. \label{eqn:util_omega_def}
\end{align}
The precise form of the feature-to-output matrix $G$ and the error matrix $F_t$ are deferred to \textsection\ref{subsec:sysid_partial_bds}. 
\end{example}

In the first part of this tutorial, we aim to provide the key technical tools that one needs to solve the following identification problem:  Given a single trajectory of the feature-output pairs $\lbrace \phi_t, y_t \rbrace_{t=1}^{T}$ obeying \eqref{eqn:feature_map}, construct an estimate $\widehat{G}$ of $G$ such that for a prescribed level of confidence $\delta \in (0,1)$, 
\begin{align}
    \PP\left(\opn{\widehat{G} - G} \le \varepsilon \right) \ge 1 - \delta,
\end{align}
where $\varepsilon$ corresponds to the estimation error rate and typically would depend on the sample size $T$, the confidence level $\delta$, the dimensions $d_y$ and $d_\phi$, and the properties of the underlying bilinear system. 
In general, under suitable noise assumptions on the underlying dynamical system, one expects the estimation error rate to satisfy
\begin{align}\label{eq:err_rate}
    \varepsilon \propto  C_{\rm sys} \times  \sqrt{\frac{\textup{problem dimension} + \log(1/\delta)}{\textup{sample size}}}, 
\end{align}
where $C_{\rm sys}$ denotes the system dependent constants. Furthermore, as we shall see next, an error rate of the above form \eqref{eq:err_rate} will only be possible provided the sample size is sufficiently large, i.e., verifying a condition of the form
\begin{align}\label{eq:sample_size}
    T \gtrsim C_{\rm sys} \times (\textup{problem dimension} + \log(1/\delta)).
\end{align}
This is similar to what one would expect for learning linear dynamical systems \citep{ziemann2023tutorial}, except that an understanding of the optimal dependencies on the underlying system properties and dimensions are far less understood for bilinear systems. Therefore, our focus in this section will be on how to obtain rates that are at least reasonable in dimension dependence and somewhat tight in terms of $\delta$ and the sample size $T$.

\subsection{Least-Squares Regression Revisited}\label{subsec:LSE}

Given the linear form \eqref{eqn:feature_map}, a natural choice for estimating $G$ using a single finite trajectory of feature-output pairs $\lbrace \phi_t, y_t \rbrace_{t=1}^{T}$ is that of the Least-Squares Estimator (LSE), which is defined as
\begin{align} \label{eqn:ERM_Ghat}
	\widehat{G}  \in \argmin_{G \in \R^{d_y \times d_{\phi}}} \sum_{t=1}^{T} \tn{y_t - G \phi_{t}}^2. 
\end{align}
The (minimum norm) LSE of $G$ admits the closed form  
\begin{align}
    \widehat{G} := \left( \sum_{t=1}^T   y_t \phi_{t}^\top \right) \left(\sum_{t=1}^{T} \phi_{t} \phi_{t}^\top \right)^{\dagger},  
\end{align}
where $(\,\cdot\,)^\dagger$ denotes the Moore–Penrose inverse.
An important consequence of this closed form, when $\sum_{t=1}^{T} \phi_{t} \phi_{t}^\top \succ 0$, is that we can express the estimation error as
\begin{align}\label{eq:est error}
    \widehat{G} - G &= \left( \sum_{t=1}^T \left(z_t  +  \eta_t\right) \phi_t ^\top \right) \left(\sum_{t=1}^{T} \phi_{t} \phi_{t}^\top \right)^{-1}.
\end{align}
To analyze the estimation error \eqref{eq:est error}, a natural approach is to split it first into two terms, one that depends on the approximation biases, and the one that depends on the noise process.
In particular,
\ifthenelse{\equal{\version}{arxiv}}{
\begin{equation}
\begin{aligned}\label{eqn:estimation_error}
    \Vert \widehat{G} - G \Vert_\op &\le  \underbrace{\left\Vert \left( \sum_{t=1}^T  \eta_t \phi_t ^\top \right) \left(\sum_{t=1}^{T} \phi_{t} \phi_{t}^\top \right)^{-1}\right\Vert_{\op}}_{E_1}  + \underbrace{\left\Vert \left( \sum_{t=1}^T z_t  \phi_t ^\top \right) \left(\sum_{t=1}^{T} \phi_{t} \phi_{t}^\top \right)^{-1}\right\Vert_{\op}}_{E_2}, 
\end{aligned}
\end{equation}
where $E_1$ denotes the error due to noise, and $E_2$ denotes the error due to truncation.
}{
\begin{equation}
\begin{aligned}\label{eqn:estimation_error}
    \Vert \widehat{G} - G \Vert_\op &\le  E_1  + E_2 ,
\end{aligned}
\end{equation}~where
\begin{align*}
    E_1 &\triangleq  \underbrace{\left\Vert \left( \sum_{t=1}^T  \eta_t \phi_t ^\top \right) \left(\sum_{t=1}^{T} \phi_{t} \phi_{t}^\top \right)^{-1}\right\Vert_{\op}}_{\textup{Error due to noise}}, \\
     E_2  & \triangleq  \underbrace{\left\Vert \left( \sum_{t=1}^T z_t  \phi_t ^\top \right) \left(\sum_{t=1}^{T} \phi_{t} \phi_{t}^\top \right)^{-1}\right\Vert_{\op}}_{\textup{Error due to truncation}}. 
\end{align*}
}
To proceed, we further decompose the noise error and the truncation error as
\ifthenelse{\equal{\version}{arxiv}}{
\begin{equation}
\begin{aligned}\label{eqn:estimation_error_decomposition} 
    E_1 & \le \opn{\left(\sum_{t=1}^{T} \phi_{t} \phi_{t}^\top \right)^{-{1}/{2}}} \opn{\left(\sum_{t=1}^T \eta_t  \phi_t ^\top \right)\left(\sum_{t=1}^{T} \phi_{t} \phi_{t}^\top \right)^{-{1}/{2}}}, \\
    E_2  & \le \opn{\left(\sum_{t=1}^{T} \phi_{t} \phi_{t}^\top \right)^{-1}} \opn{\sum_{t=1}^T z_t  \phi_t ^\top}.
\end{aligned}
\end{equation}
}{
\begin{equation}
\begin{aligned}\label{eqn:estimation_error_decomposition} 
    E_1 & \le \opn{\left(\sum_{t=1}^{T} \phi_{t} \phi_{t}^\top \right)^{\!\!\!\!-\frac{1}{2}}} \opn{\left(\sum_{t=1}^T \eta_t  \phi_t ^\top \right)\left(\sum_{t=1}^{T} \phi_{t} \phi_{t}^\top \right)^{\!\!\!\!-\frac{1}{2}}}, \\
    E_2  & \le \opn{\left(\sum_{t=1}^{T} \phi_{t} \phi_{t}^\top \right)^{\!\!\!-1}} \opn{\sum_{t=1}^T z_t  \phi_t ^\top}.
\end{aligned}
\end{equation}
}
In the decompositions \eqref{eqn:estimation_error_decomposition}, we see that there are three terms that one needs to analyze in order to obtain a finite-time guarantee on the least-squares estimator $\widehat{G}$:
\begin{itemize}
    \item [(1)] Establishing persistence of exciting by upper bounding $\opn{\left(\sum_{t=1}^{T} \phi_{t} \phi_{t}^\top \right)^{-1}}$, which equivalently amounts to lower bounding $\lambda_{\min}\left(\sum_{t=1}^T \phi_t \phi_t^\top\right)$.
    \item[(2)] Controlling the truncation bias term $\opn{  \sum_{t=1}^T z_t  \phi_t ^\top}$ by providing an upper bound.
    \item [(3)] Controlling the self-normalized martingale term $\opn{ \left( \sum_{t=1}^T  \eta_t \phi_t ^\top \right) \left(\sum_{t=1}^{T} \phi_{t} \phi_{t}^\top \right)^{-1/2}}$ by providing an upper bound.
\end{itemize} 

The purpose of the first part of this paper is precisely to provide the tools for the analysis of these three terms, which will be the subject of the subsequent sections. Again, one may recognize that the error decomposition in \eqref{eqn:estimation_error_decomposition} is very similar to that used for linear dynamical systems \citep{ziemann2022single}.
Here, we want to emphasize that the major distinction lies in the statistical dependencies that arise between  $\lbrace \phi_t \rbrace_{t\ge 1}$, $\lbrace \eta_t \rbrace_{t\ge 1}$, and $\lbrace z_t \rbrace_{t\ge 1}$. 
What is perhaps not usual is the decomposition of the error term due to truncation, which is typically, for partially observed linear dynamical systems, upper bounded in a similar fashion to the error term due to noise as
\ifthenelse{\equal{\version}{arxiv}}{
\begin{align}
     E_2 & \le \opn{\left(\sum_{t=1}^{T} \phi_{t} \phi_{t}^\top \right)^{-{1}/{2}}} \opn{\left(\sum_{t=1}^T z_t  \phi_t ^\top \right)\left(\sum_{t=1}^{T} \phi_{t} \phi_{t}^\top \right)^{-{1}/{2}}}, \nn \\
     &\le \opn{\left(\sum_{t=1}^{T} \phi_{t} \phi_{t}^\top \right)^{-{1}/{2}}} \cdot   \left(\max_{t \in [T]} \tn{z_t} \sqrt{T d_\phi} \right). 
\end{align}
}{
\begin{align}
     E_2 & \le \opn{\left(\sum_{t=1}^{T} \phi_{t} \phi_{t}^\top \right)^{\!\!\!\!-\frac{1}{2}}} \opn{\left(\sum_{t=1}^T z_t  \phi_t ^\top \right)\left(\sum_{t=1}^{T} \phi_{t} \phi_{t}^\top \right)^{\!\!\!\!-\frac{1}{2}}} \nn   \\
      & \le \opn{\left(\sum_{t=1}^{T} \phi_{t} \phi_{t}^\top \right)^{\!\!\!\!-\frac{1}{2}}} \cdot   \left(\max_{t \in [T]} \tn{z_t} \sqrt{T d_\phi} \right). \nn
\end{align}
}
This approach, however, does not always work for bilinear systems and may yield vacuous bounds, as we discuss in more detail in \textsection\ref{sec:sysid}.

\subsection{Overview}
The remainder of Part~\ref{part:learning} is organized as follows: In \textsection\ref{sec:prels}, we will provide some of the key probabilistic tools that we will rely on for the subsequent analysis. In \textsection\ref{sec:onesided}, we provide a a generic one-sided approach for establishing a persistence of excitation with heavy-tailed covariates. 
Finally, in \textsection\ref{sec:sysid}, we specialize the previously presented tools to the various classes of bilinear systems to obtain finite-sample bounds.

\subsection{Notes}

The approach of unrolling the dynamics to obtain a linear form with approximation bias has been extensively used in the study  linear dynamical systems, particularly under partial observability \citep{oymak2021revisiting, sarkar2021finite, he2026finite}. For bilinear systems, the same approach applies  with the caveat that the resulting features are nonlinear, and have complex dependencies with the bias and noise terms \citep{berk2012identification, sattar2022finite,  sattar2024learning, sattar2025finite, chatzikiriakos2026endtoend}.   

Least-squares estimation for system identification has a rich history \citep{ljung1999system}. The decompositions that we presented have been recently analyzed quite extensively \citep{simchowitz2018learning, faradonbeh2018finite, sarkar2019near, sarkar2021finite, tsiamis2019finite, oymak2018non, ziemann2023tutorial}. The decomposition of the statistical estimation error into a self-normalized term and a persistence of excitation term is due to \cite{lai1982least, lai1983asymptotic} where they established the asymptotic properties of least squares for linear system identification. \cite{sarkar2019near} revisited this decomposition using modern concentration tools to obtain finite sample size bounds. \cite{simchowitz2018learning} was first to establish the so-called learning without mixing guarantees, i.e., finite-sample guarantees that do not degrade with the stability radius, by replacing mixing-time arguments with a block-martingale small-ball condition that lower bounds the empirical Gram matrix. 
Most of these results focus on analyzing the estimation error in the operator norm. Recently, \cite{zheng2026near, zhou2026clt} establish sharper finite sample bounds for the Frobenius norm error. In particular, \cite{zhou2026clt} rely on a different decomposition than the ones typically used prior work and presented in this tutorial, and this is precisely what leads to their refined analysis.

 \section{Concentrations, Coverings, and Martingales}\label{sec:prels}
Before we present our main analysis tools for learning bilinear systems, we present a few preliminary concentration and covering tools that are frequently used in our analysis.

\subsection{Concentration Inequalities}

Concentration inequalities are powerful tools that allow us to derive quantitative bounds on the deviation of random quantities around their expected values \citep{boucheron2003concentration}. Therefore, they are extremely useful for deriving non-asymptotic bounds for system identification, notably for linear dynamical systems \citep{matni2019tutorial, ziemann2023tutorial}. There is a plethora of such inequalities with various strengths and weaknesses and our aim is not to cover all of these, nor to explain how to derive them. Instead, we only recall some of these that happen to be useful for our exposition. 

First, we present below a one-sided Bernstein's inequality (Proposition 2.14 in \cite{wainwright2019high}). This inequality characterizes the tail probability of the deviation of a sum of random variables bounded from above by their respective variances.
    
\begin{proposition}[One-sided Bernstein's inequality] \label{thrm:onesided_bernstein}
  Let $X_1, \dots, X_n$ be independent random variables taking values in $\RR$ and satisfying for all $i \in [n]$, $X_i \le b$ almost surely. We have for all $\varepsilon > 0$, 
  \ifthenelse{\equal{\version}{arxiv}}{
  \begin{align*}
      \PP \left( \sum_{i=1}^n \left( X_i - \E[X_i]\right) \geq  n\varepsilon\right) \leq \exp \left( \frac{-n \varepsilon^2}{\frac{2}{n}\sum_{i=1}^n \E[X_i^2] + \frac{2b\varepsilon}{3}}\right). 
  \end{align*}
  }{
  \begin{align*}
      \PP \left( \sum_{i=1}^n \left( X_i {-} \E[X_i]\right) {\geq}  n\varepsilon\right) {\leq} \exp \left( \frac{-n \varepsilon^2}{\frac{2}{n}\sum_{i=1}^n \E[X_i^2] {+} \frac{2b\varepsilon}{3}}\right). 
  \end{align*}
  }
\end{proposition}
\medskip
\noindent Using Proposition \ref{thrm:onesided_bernstein}, we can immediately deduce the following result:

\begin{proposition}\label{corr:onesided_bernstein}
   Let $X_1, \dots, X_n$ be i.i.d. copies of a random variable $X$ taking values in $\RR$ with finite fourth moment. We have for all $\varepsilon  \in  (0, 1)$, 
  \ifthenelse{\equal{\version}{arxiv}}{
  \begin{align}
      \PP \left( \sum_{i=1}^n  X_i^2  \leq n(1-\varepsilon)\EE[X^2]\right) \leq \exp \left( - \frac{n \varepsilon^2}{2} \cdot \frac{\EE[X^2]^2}{\E[X^4] }\right). \nn
  \end{align}
  }{
  \begin{align}
      \PP \left( \sum_{i=1}^n  X_i^2  {\leq} n(1-\varepsilon)\EE[X^2]\right) \leq \exp \left( - \frac{n \varepsilon^2}{2} \cdot \frac{\EE[X^2]^2}{\E[X^4] }\right). \nn
  \end{align}
  }
\end{proposition}
\medskip 
 The statement in Proposition \ref{corr:onesided_bernstein} only requires a fourth moment assumption to hold. Therefore, this result is particularly powerful and will prove useful for establishing persistence of excitation results under heavy tailed covariates as we shall see in \textsection\ref{sec:onesided}.

\subsection{Covering Arguments}

In the previous subsection, we have seen concentration inequalities that hold for sums of scalar random variables. 
However, most of the quantities we aim to analyze are either matrices or vectors. To analyze such quantities, we use the so-called $\epsilon$-net arguments, which consist of reducing the analysis of matrices or vectors to simpler scalar terms for which scalar concentration inequalities can immediately apply. To start with, let us define what we mean by nets in the case of the unit sphere in $\RR^d$.
\begin{definition}[Nets]
    Let $\epsilon > 0$. A finite set $\cN \subseteq \cS^{d-1}$ is said to be an $\epsilon$-net of $\cS^{d-1}$ with respect to the norm $\tn{\cdot}$ if for all $x \in \cS^{d-1}$ there exist $x' \in \cN$ such that $\tn{x - x'} \le \epsilon$.  
    
\end{definition}

For our purposes, we will need to use an $\epsilon$-net argument to control either the operator norm $\Vert W \Vert_{\op}$ of some random matrix $W$, or its minimum eigenvalue $\lambda_{\min}(W)$. 
Below, we present two lemmas that make this argument precise.
\begin{lemma}[$\epsilon$-net argument for the operator norm]\label{lem:two-sided-net-nonsym}
    Let $W$ be a $k \times d$ random matrix, and $\epsilon \in (0,1)$. Let $\cN$ be an $\epsilon$-net of $\cS^{d-1}$ with maximal cardinality. Then, for all $\nu > 0$, we have
     \ifthenelse{\equal{\version}{arxiv}}{
    \begin{align*}
    \PP\left( \Vert W \Vert_{\op} > \nu \right) &\le \left(1+\frac{2}{\epsilon}\right)^d  \max_{x \in \cN}\PP\left( \tn{W x}
    > (1 - \epsilon)\nu \right). 
    \end{align*}
    }{
    \begin{align*}
    \PP\left( \Vert W \Vert_{\op} {>} \nu \right) {\le} \left(1{+}\frac{2}{\epsilon}\right)^{\!d}  \max_{x \in \cN}\PP\left( \tn{W x}
    > (1 - \epsilon)\nu \right). 
    \end{align*}
    }
    Furthermore, if $W$ is a $d \times d$ symmetric random matrix, and $\epsilon \in (0,1/2)$, then,
    \ifthenelse{\equal{\version}{arxiv}}{
    \begin{align*}
        \PP\left( \Vert W \Vert_{\op} > \nu \right) &\le \left(1+\frac{2}{\epsilon}\right)^d  \max_{x \in \cN}\PP\left( \vert x^\top W x\vert  > (1 - 2\epsilon)\nu \right),
    \end{align*}
    }{
    \begin{align*}
        \PP\left( \Vert W \Vert_{\op} {>} \nu \right) {\le} \left(1{+}\frac{2}{\epsilon}\right)^{\!d}  \max_{x \in \cN}\PP\left( \vert x^\top W x\vert  > (1 - 2\epsilon)\nu \right)
    \end{align*}
    }
\end{lemma}

\begin{lemma}[$\epsilon$-net argument for the minimum eigenvalue] \label{lem:one sided net}Let $W$ be a $d \times d$ symmetric random matrix, and $\epsilon \in (0,1/2)$.
    Let $\cN$ be an $\epsilon$-net of $\cS^{d-1}$ with maximal cardinality.  Then, for all $\nu > 0$, we have
    \ifthenelse{\equal{\version}{arxiv}}{
    \begin{align*}
        &\PP\left( \lambda_{\min}(W) < \nu  \right)  \le  \left(1+\frac{2}{\epsilon}\right)^{d} \max_{x \in \cN} \PP(  x^\top W x < \nu + 2\epsilon \Vert W \Vert_{\op} ) 
    \end{align*}
    }{
    \begin{align*}
        &\PP\left( \lambda_{\min}(W) {<} \nu  \right)  {\le}  \left(1{+}\frac{2}{\epsilon}\right)^{\!d} \max_{x \in \cN} \PP(  x^\top W x {<} \nu {+} 2\epsilon \Vert W \Vert_{\op} ) 
    \end{align*}
    }
\end{lemma}

The proof of Lemma~\ref{lem:two-sided-net-nonsym} is standard and is therefore omitted. It can be found in \cite{vershynin2018high,wainwright2019high,tao2023topics}. 
\ifthenelse{\equal{\version}{arxiv}}{
The proof of Lemma~\ref{lem:one sided net} is deferred to Appendix~\ref{app:covering}.
}{
The proof of Lemma~\ref{lem:one sided net} is deferred to Appendix E.3 in~\cite{sattar2026learning}
}

\subsection{Vector-valued Martingales \& Concentration}\label{subsec:martingale}

An important class of vector-valued random variables that we will be using in our exposition is that of subgaussian random variables.  
\begin{definition}[Subgaussian random vector~\citep{vershynin2010introduction}]\label{def:subGaussian RV} A random vector $x$ taking values in $\R^{d}$ is said to be $\sigma^2$-subgaussian with variance proxy parameter $\sigma>0$, if for every $\lambda \in \R^{d}$ we have 
    \begin{align}
        \E[\exp( \lambda^\top x  )] \leq \exp\left(\frac{\sigma^2 \tn{\lambda}^2}{2}\right). \label{eqn:subgaussian_RV_mgf}
    \end{align}
\end{definition}

\medskip 
Next, we present a version of Azuma-Hoeffding's inequality which is useful for analyzing sums of causally dependent random variables. In our case, the need for such a bound arises when upper bounding the truncation bias term. 

\begin{proposition}[Sub-Gaussian martingale]\label{lem:freedman}
    Let $\{\Fcal_t\}_{t\ge 1}$ be a filtration. Let $\{\eta_t\}_{t \ge 1}$ be a sequence of conditionally zero-mean, $\sigma^2$-subgaussian random vectors (i.e., $\E[\eta_t \bgl \Fcal_{t-1}] = 0$, and $\E[\exp(\lambda^\T \eta_t)\bgl \Fcal_{t-1}] \leq \exp(\sigma^2\tn{\lambda}^2/2)$) taking values in $\R^{d_{\eta}}$, such that $\eta_t$ is $\Fcal_t$-measurable for all $t\ge 1$. Let $\{x_t\}_{t\geq 1}$ be a sequence of random vectors taking values in $\R^{d_{\eta}}$ such that for all $t \ge 1$, $x_t$ is $\Fcal_{t-1}$-measurable and $\tn{x_t} \le K$ almost surely for some $K > 0$. Then for all $\delta \in (0,1)$, and $T \ge 1$, we have
    \begin{align*}
        \PP\left( \bigg| \sum_{t=1}^T x_t^\T \eta_t \bigg| \le  \sigma K \sqrt{2 T \log(2/\delta)}\right) \ge 1 - \delta.
    \end{align*}
\end{proposition}
\medskip 
Another key inequality is the self-normalized martingale concentration bound due to \cite{abbasi2011improved}. Below, we present a version of this bound that can be found in \cite{ziemann2023tutorial}.

\begin{proposition}[Self-normalized martingale]\label{prop:self-normalized vector martingale}
    Let $\lbrace \cF_t \rbrace_{t\ge 0}$ be a filtration. Let $\lbrace \eta_t \rbrace_{t\ge 1}$ be a stochastic process taking values in $\RR^{d_\eta}$, such that for all $t \ge 1$, $\eta_t$ is $\cF_t$-measurable, and conditionally on $\cF_{t-1}$, $\sigma^2$-subgaussian, i.e., for all $\lambda \in \RR^{d_\eta}$, $\EE[\exp(\lambda^\top \eta_t) \bgl \cF_{t-1}] \le \exp(\tn{\lambda}^2 \sigma^2/2)$.  
    Let $\lbrace x_t\rbrace_{t\ge 1 }$ be a sequence of random vectors taking values in $\RR^{d_x}$, such that for all $t \ge 1$, $x_t$ is $\cF_{t-1}$-measurable.  Let $\Lambda \in \RR^{d_x \times d_x}$ be symmetric and positive definite.
    Then, for $T \ge 1$ and $\delta \in (0,1)$, the event:
    \ifthenelse{\equal{\version}{arxiv}}{
    \begin{align*}
        & \opn{\left(\sum_{t=1}^T \eta_t  x_t ^\top \right)\left(\sum_{t=1}^{T} x_{t} x_{t}^\top + \Lambda \right)^{-{1}/{2}}}^2   \le  4\sigma^2 \log\left(\frac{\det(\sum_{t=1}^T x_t x_t^\top + \Lambda)}{\det (\Lambda)}\right) + 8 \sigma^2 \log\left(\frac{5^{d_\eta}}{\delta}\right)
    \end{align*}
    }{
    \begin{align*}
        & \opn{\left(\sum_{t=1}^T \eta_t  x_t ^\top \right)\left(\sum_{t=1}^{T} x_{t} x_{t}^\top + \Lambda \right)^{\!\!\!\!-\frac{1}{2}}}^2  \\
        & \quad \le  4\sigma^2 \log\left(\frac{\det(\sum_{t=1}^T x_t x_t^\top + \Lambda)}{\det (\Lambda)}\right) + 8 \sigma^2 \log\left(\frac{5^{d_\eta}}{\delta}\right)
    \end{align*}
    }
    holds with probability at least $1-\delta$.
\end{proposition}

\subsection{Notes}

The topic of concentration of measure has been the subject of extensive research in the past century \citep{boucheron2003concentration, van2014probability}, and this tutorial cannot do it justice. 
The reader may also refer to the tutorials \citep{matni2019tutorial} and \citep{ziemann2023tutorial} that give a gentle exposition of such topics tailored for the control community. 
The concentration inequalities we present are somewhat standard, and there are many references where they can be found \citep{boucheron2003concentration, van2014probability, vershynin2018high, wainwright2019high, rigollet2023high}.
The idea of using covering or net arguments to analyze the operator norm of a random matrix or its minimum eigenvalue can be found in numerous handbooks \citep{van2014probability, vershynin2018high, tao2023topics}.
The exposition we follow is inspired by that in \cite{vershynin2018high}. 
It is worth mentioning that there exists other approaches to deal with the operator norm or minimum eigenvalue, notably generic chaining \citep{van2014probability, vershynin2018high} and the PAC-Bayes approach \citep{oliveira2016lower}. 
The self-normalized martingale bound we present is due to \cite{abbasi2011regret}, but the original ideas for deriving it are due to \cite{de2004self, de2009self}.

 \usetikzlibrary{positioning,calc,decorations.pathreplacing}

\begin{figure*}[!htbp]
\centering
\begin{tikzpicture}[
    every node/.style={font=\small},
    circ/.style={circle, draw, fill=gray!12, minimum size=8mm, inner sep=0pt},
    map/.style={rectangle, draw, rounded corners=2pt, fill=blue!10,
                minimum width=1.4cm, minimum height=8mm},
    arrow/.style={->, thick, dashed},
    brace/.style={decorate, decoration={brace, amplitude=8pt}}
]

\def\xstep{1.0}
\def\ytop{2.2}
\def\ymid{0}
\def\ybot{-2.2}

\foreach \i/\lab in {
0/{\zeta_0},
1/{\zeta_1},
2/{\zeta_2},
3/{\zeta_3},
4/{\cdots},
5/{\zeta_{\tau}},
6/{\zeta_{\tau+1}},
7/{\zeta_{\tau+2}},
8/{\zeta_{\tau+3}},
9/{\cdots},
10/{\zeta_{2\tau}},
11/{\cdots}
}{
    \node[circ, fill=orange!15] (z\i) at ({\i*\xstep},\ytop) {\(\lab\)};
}

\node[map] (p0) at ({3*\xstep},\ymid) {\(\psi_\tau~~\)};
\node[map] (p1) at ({4*\xstep},\ymid) {\(\psi_{\tau+1}~~\)};
\node[map] (p2) at ({5*\xstep},\ymid) {\(\psi_{\tau+2}~~\)};
\node[map] (p3) at ({6*\xstep},\ymid) {\(\psi_{\tau+3}~~\)};
\node[map] (p4) at ({7*\xstep},\ymid) {\(\cdots~~\)};
\node[map] (p5) at ({8*\xstep},\ymid) {\(\psi_{2\tau}~~\)};
\node[map] (p6) at ({9*\xstep},\ymid) {\(\cdots\)};

\foreach \i/\lab in {
0/{u_0},
1/{u_1},
2/{u_2},
3/{u_3},
4/{\cdots},
5/{u_{\tau}},
6/{u_{\tau+1}},
7/{u_{\tau+2}},
8/{u_{\tau+3}},
9/{\cdots},
10/{u_{2\tau}},
11/{\cdots}
}{
     \node[circ, fill=green!12] (u\i) at ({\i*\xstep},\ybot) {\(\lab\)};
}

\node[left=0.45cm of z0] {\(\{\zeta_t\}_{t = 0}^T \distas \Dcal_\zeta\)};
\node[left=0.45cm of p0] {\(\{\psi_t\}_{t=\tau}^T\)};
\node[left=0.45cm of u0] {\(\{u_t\}_{t=0}^T \distas \Dcal_u\)};

\draw[brace]
    ($(z5.south east)+(0,-0.25)$) --
    ($(z1.south west)+(0,-0.25)$);

\draw[brace]
    ($(z6.south east)+(0,-0.5)$) --
    ($(z2.south west)+(0,-0.5)$);

\draw[brace]
    ($(z7.south east)+(0,-0.75)$) --
    ($(z3.south west)+(0,-0.75)$);

\draw[brace]
    ($(u1.north west)+(0,0.25)$) --
    ($(u5.north east)+(0,0.25)$);

\draw[brace]
    ($(u2.north west)+(0,0.5)$) --
    ($(u6.north east)+(0,0.5)$);

\draw[brace]
    ($(u3.north west)+(0,0.75)$) --
    ($(u7.north east)+(0,0.75)$);

\draw[arrow] ($(z5.south)+(0,-0.85)$) -- (p2.north);
\draw[arrow] ($(u5.north)+(0,0.85)$) -- (p2.south);

\draw[arrow] ($(z4.south)+(0,-0.65)$) -- (p1.north);
\draw[arrow] ($(u4.north)+(0,0.65)$) -- (p1.south);

\draw[arrow] ($(z3.south)+(0,-0.45)$) -- (p0.north);
\draw[arrow] ($(u3.north)+(0,0.45)$) -- (p0.south);

\end{tikzpicture}
\caption{Block-dependent structure induced by sliding windows. The processes \(\{\zeta_t\}_{t=0}^T\) and \(\{u_t\}_{t=0}^T\) are i.i.d., while \(\{\psi_t\}_{t=\tau}^T\) is a nonlinear function $\psi$
of length-\(\tau\) windows of both processes, i.e., $\psi_t:= \psi\!\left(
\zeta_t,\ldots,\zeta_{t-\tau+1},u_t,\ldots,u_{t-\tau+1}\right)$. For bilinear systems, the covariate process $\{\phi_t\}_{t=\tau}^T$ as well as the noise process $\{\eta_t\}_{t=\tau}^T$ in \eqref{eqn:feature_map} can be (approximately) modeled as the block-dependent process $\{\psi_t\}_{t=\tau}^T$ defined above.} 
\label{fig:sliding-window-block-dependent}
\end{figure*}
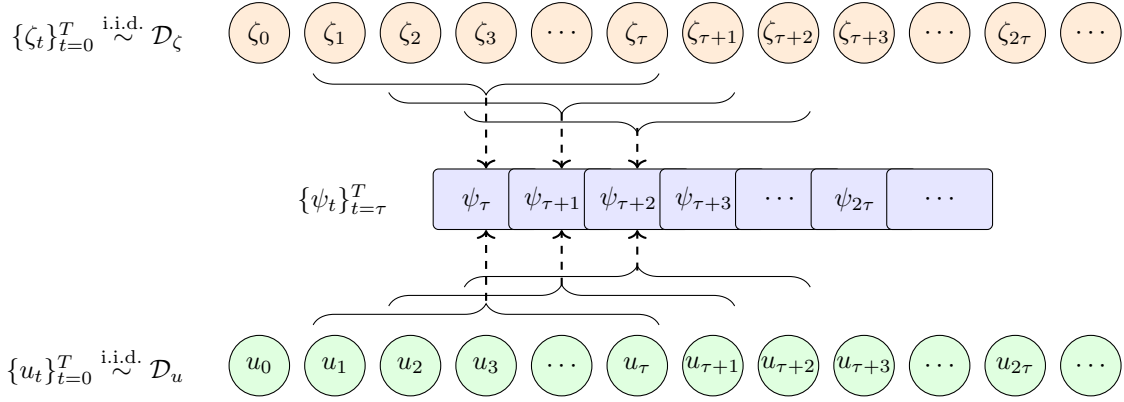 
\section{Persistence of 
Excitation with Heavy-tailed Covariates}\label{sec:onesided}
In this section, we provide a self-contained analysis of the lower tail of the empirical covariance matrix $\widehat{\Sigma}_T := \sum_{t=1}^T \phi_t \phi_t^\top$ appearing in our error bound in \eqref{eqn:estimation_error}. 
One advantage of analyzing the lower spectrum of the empirical covariance is that it decouples stability and persistence of excitation.

\subsection{Distributional Assumptions on the Covariates}

Recall that $\phi_t$ can be dependent, correlated, and heavy-tailed. 
More specifically, we assume the following dependence structure (as illustrated by Figure~\ref{fig:sliding-window-block-dependent}).
\begin{definition}[Block-dependence]\label{def:blockdepend}\ifthenelse{\equal{\version}{arxiv}}{}{\!\!\!}
A sequence of random vectors $\{x_t\}_{t \geq 0}$, taking values in $\R^{d_x}$, is said to have a block-dependent structure if there exists an integer $\tau\geq 1$ such that $x_{t_1}$ and $x_{t_2}$ are independent if and only if $|t_1 - t_2| > \tau$.
\end{definition}

One of the major sources of difficulty in establishing the persistence of excitation result is the nonlinear dependence of $\phi_{t}$  on $u_{t-\tau+1}, \dots, u_{t}$ for all $t \ge \tau$. 
We need to understand how distributional properties of the input impact the lower spectrum of $\widehat{\Sigma}_T$. 
To that end, we start by introducing the property of hyper-contractivity, which guarantees a subgaussian-type lower tail of the empirical covariance matrix $\widehat{\Sigma}_T$, despite heavy-tailed and block-dependent covariates.

\begin{definition}[Hypercontractivity]\label{def:hypercont}
A $d_x$-dimensional random vector $x $ is  $(4,2, \gamma)
$-hypercontractive, if  $\EE[(u^\top x)^4] \le \gamma \EE[(u^\top x)^2]^2$, for all $u \in \R^{d_x}$. 
\end{definition}

The $(4,2,\gamma)$-hypercontractivity property is satisfied by many classical distributions. Notably, a $d_x$-dimensional standard Gaussian random vector satisfies it with $\gamma = 3$, while a $d_x$-dimensional random vector sampled uniformly from $\sqrt{d_x} \cdot \cS^{d_x-1}$ satisfies it with $\gamma = 3/(1+2/d_x)$. In \S\ref{sec:sysid}, we show that, in the case of bilinear systems, the nonlinear covariates $\{\phi_t\}_{t=1}^T$ are $(4,2,\gamma)$-hyper-contractive. Specifically, we will derive upper bounds on $\gamma$ for various bilinear systems. Interestingly, the $(4,2,\gamma)$-hypercontractivity also implies the small ball condition via Paley-Zygmund inequality.

\begin{assumption}[Distributional properties of $\phi_t$] \label{assump:phi_t distribution}
    $\lbrace \phi_t \rbrace_{t\ge 1}$ is a sequence of block-dependent, zero-mean, isotropic\footnote{A $d_{\phi}$-dimensional random vector $\phi$ is isotropic if for all $x \in \RR^{d_\phi}$, $\EE[(x^\top \phi)^2] = \Vert x \Vert_{\ell_2}^2$.}, and $(4, 2, \gamma)$-hypercontractive for some $\gamma > 1$, random vectors taking values in $\R^{d_{\phi}}$. 
\end{assumption}

Assumption \ref{assump:phi_t distribution} covers a wide range of input distributions that may even be heavy-tailed, as it only requires conditions on the first four moments of the distribution. This contrast with classical assumptions that require the input distribution to have subgaussian tails, and is also consistent with the intuition that bounding the smallest singular value of random matrix requires weaker moment conditions  \citep{koltchinskii2015bounding}.

\subsection{Lower Bound on the Spectrum of the Empirical Covariance}
We are now ready to present our main result on the  persistence of excitation.

\begin{theorem}[Persistence of Excitation]\label{thm:persistence}
Suppose the covariate process $\lbrace \phi_{t}\rbrace_{t \ge 1}$ satisfies Assumption~\ref{assump:phi_t distribution}.
Then, for all $\delta \in (0,1)$, the event:
\ifthenelse{\equal{\version}{cdc}}{\vspace{-6pt}}{}
\begin{equation}
    \lambda_{\min}\left( \sum_{t=1}^{T} \phi_t \phi_t^\T \right) \geq T/4,
\end{equation}
holds with probability at least $1- \delta$, provided that 
\begin{equation}
    T \gtrsim \tau\,\gamma\left (\log\left(\frac{2\tau}{\delta}\right) + d_{\phi} \log\left(1+\frac{16d_{\phi}}{\delta}\right) \right). \label{eqn:trajectory_size_main}
\end{equation}
\end{theorem}

\ifthenelse{\equal{\version}{arxiv}}{
The proof of Theorem \ref{thm:persistence} is deferred to Appendix \ref{app:persistence excitation}. 
}{
The proof of Theorem \ref{thm:persistence} is deferred to Appendix A.1 in \cite{sattar2026learning}. 
}
Interestingly, despite the presence of nonlinearities and dependencies in the covariate process $\{\phi_t\}_{t \geq 1}$, persistence of excitation is still guaranteed.

\subsection{Notes}

In general, to derive high-probability lower bounds on the smallest singular value of a random matrix, there are two approaches, namely a two-sided approach and one-sided approach. The two-sided approach aims at controlling both sides of the spectrum of the random matrix simultaneously (see~\cite[Chapter~4]{vershynin2018high}).
For this purpose, one often relies on the net argument in Lemma~\ref{lem:two-sided-net-nonsym}.
\cite{jedra2020finite} leveraged this two-sided approach and used the Hanson inequality~\citep{rudelson2013hanson} to derive tight concentration bounds on the entire spectrum of the covariates matrix in the case of linear systems. The two-sided approach only provides meaningful bounds when dealing with stable systems. 
The one-sided approach aims at controlling directly the minimum eigenvalue and therefore typically requires weaker assumptions.
Standard implementations of this approach also rely on a net argument, but would require a finer lower tail control for each fixed direction. In particular, such control may follow from a small-ball condition, as first observed in~\cite{Mendel2}. This idea was subsequently used to establish persistence of excitation for marginally stable linear systems~\citep{simchowitz2018learning}. It is also worth mentioning that other techniques, such as PAC-Bayes methods~\citep{oliveira2016lower} and matrix Freedman inequalities~\citep{tropp2015introduction}, can yield lower-tail bounds for the minimum eigenvalue without relying on a net argument.

The approach we present in this section for establishing persistence of excitation follows that of \cite{sattar2025finite_arXiv}, which is a one-sided approach.
It relies on the net argument of Lemma~\ref{lem:one sided net}, then uses a one-sided Bernstein's inequality given in Proposition \ref{corr:onesided_bernstein} together with a blocking trick~\cite{yu1994rates}. This gives an additional scaling by $\tau$ in our sample complexity bound. In the case of LTI systems, this can be avoided by using the partial circulant or Toeplitz structure of the design matrix~\cite{oymak2021revisiting,lee2022improved}. In the case of nonlinear covariates, these structures vanish. Hence, it is not immediately clear how to avoid scaling by $\tau$ in the sample complexity bound.

 \section{Bilinear System Identification}\label{sec:sysid}
In this section, we provide end-to-end guarantees for learning bilinear dynamical systems in different settings (depending on whether the bilinearity appears in the state update equation or the observation model).
Specifically, we use the tools and results from Sections~\ref{sec:prels} and~\ref{sec:onesided} to analyze the least-squares algorithm for estimating the \emph{Markov-like} parameter matrix $G$ in each setting, using a single input-output trajectory $\{(u_t,y_t)\}_{t=0}^T$, and derive finite-sample learning guarantees.

\subsection{Linear Dynamical Systems with Bilinear Observations}\label{subsec:sysid_partial}
The first class of dynamical systems of the form \eqref{eqn:general_dynamics} which is of interest to us is the linear dynamical systems with bilinear observations (LDS-BO).
The state-space representation of LDS-BO is given by
\begin{equation} \label{eqn:LDS-BO}
    \begin{aligned}
         x_{t+1} &= A x_t + B u_t +  w_t, \\
         y_t &= C_0 x_t + \sum_{k=1}^{d_u} (u_t)_k C_k x_t + v_t,
    \end{aligned}
    \tag{LDS-BO}
\end{equation}
where $A \in \R^{d_x \times d_x}$ is the state matrix, $B \in \R^{d_x \times d_u}$ is the input matrix, and $C_0, C_1, \dots, C_{d_u}$ are the $d_u +1$ observation matrices taking values in $\R^{d_y \times d_x}$. 
Note that, unlike partially observed linear dynamical systems, the inputs $\{u_t\}_{t \geq 0}$ in \eqref{eqn:LDS-BO} directly interact with the states $\{x_t\}_{t \geq 0}$ to affect the observation $y_t$. 
Fixing a history length $\tau>0$, the system \eqref{eqn:LDS-BO} gives the non-linear input-output map $y_t = G \, \phi_t + z_t + \eta_t$ in \eqref{eqn:feature_map} with

\renewcommand{\Gbar}{\Gamma}
\ifthenelse{\equal{\version}{arxiv}}{
\begin{equation}
\begin{aligned} \label{eqn:G_phi_r_eta_LDS_BO}
    G &= \begin{bmatrix} C_0 \Gbar & C_1 \Gbar & \cdots & C_{d_u} \Gbar\end{bmatrix}, \quad\quad
    \phi_t  = \util_t \otimes \begin{bmatrix} u_{t-1} \\ u_{t-2}  \\ \vdots \\ u_{t-\tau+1} \end{bmatrix}, \\
  z_t &= (u_{t} \circ C)  A^{\tau-1} x_{t-\tau+1}, \quad\quad \quad \quad\quad \eta_t = H_t \, \omega_t, 
\end{aligned}
\end{equation}
}{
\begin{equation}
\begin{aligned} \label{eqn:G_phi_r_eta_LDS_BO}
    G &= \begin{bmatrix} C_0 \Gbar & C_1 \Gbar & \cdots & C_{d_u} \Gbar\end{bmatrix}, \quad
    \phi_t  = \util_t \otimes \begin{bmatrix} u_{t-1} \\ u_{t-2}  \\ \vdots \\ u_{t-\tau+1} \end{bmatrix}, \\
  z_t &= (u_{t} \circ C)  A^{\tau-1} x_{t-\tau+1}, \quad\quad \quad \quad \eta_t = H_t \, \omega_t, 
\end{aligned}
\end{equation}
}

where $\Gbar := \begin{bmatrix} B  & A B & \cdots & A^{\tau-2} B\end{bmatrix} \in \R^{d_x \times (\tau-1) d_u}$ and $G \in \R^{d_y \times d_u(d_u+1)(\tau-1)}$. 
The parameters $\lbrace \lbrace C_k B \rbrace_{k=0}^{d_u},$ $ \lbrace C_k A B \rbrace_{k=0}^{d_u}, \dots, \lbrace C_k A^{\tau-1} B \rbrace_{k=0}^{d_u} \rbrace$ are what we refer to as the \emph{Markov-like} parameters of the system\footnote{Note that though the form of the Markov parameters is the same as for LTI systems, the relationship between inputs and outputs differs due to the bilinear observation.}. $\util_t$ and $\omega_t$ are as defined in \eqref{eqn:util_omega_def}, and lastly, the matrix $H_t$ appearing in the noise process $\{\eta_t\}_{t \geq 0}$ is given by
\begin{equation}
\begin{aligned}\label{eqn:Ht_epsilont}
    H_t  := \begin{bmatrix}
     I_{d_y} &
    (u_{t} \circ C)   &
    (u_{t} \circ C) A  &
    \cdots &
    (u_{t} \circ C) A^{\tau-2} 
 \end{bmatrix}. 
\end{aligned}
\end{equation}
The goal here is to estimate the \emph{Markov-like} parameter matrix $G$ by regressing $\{y_t\}_{t=\tau}^T$ on $\{\phi_t\}_{t=\tau}^T$.
Before that, we state our assumptions on the system~\eqref{eqn:LDS-BO}.
\begin{assumption}[Stability]\label{assump:stability_linear}
    The dynamical system in \eqref{eqn:LDS-BO} is strictly stable, i.e., $\rho(A) < 1$. 
\end{assumption}
It is well-known that, when $\rho(A) < 1$, there exist $\rho \in (\rho(A), 1)$ and $\kappa \geq 1$ such that $\|A^k\| \leq \kappa \rho^k$ for all $k \in \mathbb{Z}_+$. The quantity $\kappa := \sup_{k \in \mathbb{Z}_+}(\norm{A^k}/ \rho^k)$ is finite by Gelfand's formula, and it measures the transient response of the system and can be upper bounded by its $\mathcal{H}_\infty$ norm~\citep{tu2017non}.
This decay condition is important for showing that the residual term $z_t$ is small when $\tau$ is large enough, and is a relatively common assumption~\citep{oymak2021revisiting, lee2022improved}. In the case of linear observations, one can relax Assumption~\ref{assump:stability_linear} to marginal stability ($\rho(A) \leq 1$), and use different approaches such as method-of-moments~\cite{bakshi2023new}, or Kalman filtering~\cite{tsiamis2019finite,ghai2020no}. Extending these approaches to bilinear observations is still an open question. 
Next, we state our assumptions on the excitation $\{u_t\}_{t\geq 0}$ and the noise processes $\{w_t\}_{t\geq 0}$ and $\{v_t\}_{t\geq 0}$.

\begin{assumption}[Excitation/noise processes]\label{assump:input+noise}
    \emph{(\textbf{a})} $\{u_t\}_{t=0}^T$ are sampled uniformly at random from a sphere of radius $\sqrt{d_u}$, i.e., $\{u_t\}_{t=0}^T \distas \mathrm{Unif}(\sqrt{d_u} \cdot\Scal^{d_u-1})$. \emph{(\textbf{b})} $\lbrace w_t \rbrace_{t=0}^T$ and $\lbrace v_t \rbrace_{t=0}^T$ are sequences of independent, zero-mean, $\sigma^2$-subgaussian random vectors taking values in $\R^{d_x}$ and $\R^{d_y}$, respectively. 
\end{assumption}

\subsubsection{Learning Markov-like Parameter Matrix}\label{subsec:est_error_decomp_LDS_BO}
In this subsection, we derive finite-sample guarantees on learning the \emph{Markov-like} parameter matrix $G$ given by \eqref{eqn:G_phi_r_eta_LDS_BO} via regressing $\{y_t\}_{t=\tau+1}^T$ on $\{\phi_t\}_{t=\tau+1}^T$. 
Now, recall from Section~\ref{subsec:LSE} that estimating $G$ via least-squares regression gives the estimation error decomposition \eqref{eqn:estimation_error} and \eqref{eqn:estimation_error_decomposition} which are upper bounded to get the following result.

\begin{theorem}[Learning Markov-like parameters]\label{thm:main_LDS_BO}
    \ifthenelse{\equal{\version}{arxiv}}{
    Fix $\delta \in (0,1)$ and a history length $\tau > 0$. 
    Suppose Assumptions~\ref{assump:stability_linear} and~\ref{assump:input+noise} hold, and we are given a single input-output trajectory $\{(u_t,y_t)\}_{t=0}^T$ of the system~\eqref{eqn:LDS-BO}. 
    Construct $\{\phi_t\}_{t=\tau+1}^T$ according to \eqref{eqn:G_phi_r_eta_LDS_BO}.
    Let $\Ghat = ( \sum_{t=\tau}^T   y_t \phi_{t}^\top) (\sum_{t=\tau}^{T} \phi_{t} \phi_{t}^\top )^{\dagger}$ be the least-squares estimator (LSE) of $G$ in \eqref{eqn:G_phi_r_eta_LDS_BO}. Setting $d_{\phi} := d_u(d_u+1)(\tau-1)$, the event:
    }{
    Fix $\delta {\in} (0,1)$ and a history length $\tau > 0$. 
    Suppose Assumptions~\ref{assump:stability_linear} and~\ref{assump:input+noise} hold, and we are given a single input-output trajectory $\{(u_t,y_t)\}_{t=0}^T$ of the system~\eqref{eqn:LDS-BO}. 
    Construct $\{\phi_t\}_{t=\tau+1}^T$ according to \eqref{eqn:G_phi_r_eta_LDS_BO}.
    Let $\Ghat = ( \sum_{t=\tau}^T   y_t \phi_{t}^\top) (\sum_{t=\tau}^{T} \phi_{t} \phi_{t}^\top )^{\dagger}$ be the least-squares estimator (LSE) of $G$ in \eqref{eqn:G_phi_r_eta_LDS_BO}. Setting $d_{\phi} := d_u(d_u+1)(\tau-1)$, the event:
    }
    \ifthenelse{\equal{\version}{arxiv}}{
    \begin{equation}
    \begin{aligned}\label{eq:error bound}
\norm {\Ghat  - G }_\op  &\leq \frac{c_1\sigma + c_2 \sqrt{d_\phi} \rho^{\tau-1}}{1-\rho}  \sqrt{\frac{\left(d_\phi  + d_y + \log\left(\frac{\tau}{\delta}\right)\right)}{T-\tau+1}}
    \end{aligned}
    \end{equation}
    }{
    \begin{equation}
    \begin{aligned}\label{eq:error bound}
         \norm {\Ghat  - G }_\op  &\leq \frac{c_1\sigma + c_2 \sqrt{d_\phi} \rho^{\tau-1}}{1-\rho}  \sqrt{\frac{\left(d_\phi  + d_y + \log\left(\frac{\tau}{\delta}\right)\right)}{T-\tau+1}} 
    \end{aligned}
    \end{equation}
    }
    holds with probability at least $1-\delta$, provided that 
\ifthenelse{\equal{\version}{arxiv}}{
    \begin{align}
    T - \tau+1 \gtrsim  \tau \left(d_\phi \log\left(d_\phi\right) + \log\left(\frac{\tau}{\delta}\right)  \right), \label{eqn:trajectory_size_main_thm}
    \end{align}
    }{
    \begin{align}
   T - \tau+1 \gtrsim  \tau \left(d_\phi \log\left(d_\phi\right) + \log\left(\frac{\tau}{\delta}\right)  \right),  \label{eqn:trajectory_size_main_thm}
    \end{align}
    }
where $c_1$, $c_2$ are system dependent constants, and there exists a universal constant $c_0 >0$, such that $c_1 = c_0\left(1 + {\kappa \beta_C}\right)$ and $c_2 = c_0 \kappa \beta_C \sqrt{\sigma^2 + \Vert B \Vert_\op^2}$. Here, we define $\beta_C := \sup_{t \leq T}\opn{u_t \circ C}$.
\end{theorem}
\ifthenelse{\equal{\version}{arxiv}}{
The proof of Theorem~\ref{thm:main_LDS_BO} is deferred to Appendix~\ref{app:LDS-BO}. 
}{
The proof of Theorem~\ref{thm:main_LDS_BO} is deferred to Appendix D.1 in \cite{sattar2026learning}. 
}
Theorem~\ref{thm:main_LDS_BO} extends the results in \cite{sattar2024learning}, \cite{choi2025explore} to a more general bilinear observation model with $y_t = (u_t \circ C) x_t + v_t$ (unlike $y_t = u^\T C x_t + v_t$ in  \cite{sattar2024learning}, \cite{choi2025explore}). The proof of Theorem~\ref{thm:main_LDS_BO} relies on showing that $\phi_t$ in \eqref{eqn:G_phi_r_eta_LDS_BO} satisfies the hypercontractivity condition (Def.~\ref{def:hypercont}) with $\gamma=9$. This is combined with Theorem~\ref{thm:persistence} to guarantee persistence of excitation.

\subsection{Fully Observed Bilinear Dynamical Systems}\label{subsec:sysid_full} 
In this subsection, we consider the class of state-observed bilinear dynamical systems (BDS), which is also a special case of the dynamical systems of the form~\eqref{eqn:general_dynamics}. The state-space representation of BDS is given by
\begin{equation} \label{eqn:BDS-FO}
     x_{t+1} = A_0 x_t + \sum_{k=1}^{d_u} (u_t)_k A_k x_t + B u_t +  w_t, \tag{BDS}
\end{equation}
where $A_0, A_1, \dots, A_{d_u}$ denote the $d_u+1$ state matrices taking values in $\R^{d_x \times d_x}$, and $B$ denote the input matrix taking values in $\R^{d_x \times d_u}$.
Note that \eqref{eqn:BDS-FO} can also be viewed as a dynamical system with input-dependent state matrix $u_t \circ A = A_0  + \sum_{k=1}^{d_u} (u_t)_k A_k $.
The learning problem in this setting can also be formulated as a least-squares regression.
Specifically, note that \eqref{eqn:BDS-FO} can alternately be written as
\begin{equation}
\begin{aligned} \label{eqn:G_phi_t_BDS_FO}
    x_{t+1} &=  G \,\phi_t + w_t,  \\  
   \text{where} \quad G &= \begin{bmatrix} B & A_0 & A_1 & \cdots & A_{d_u}\end{bmatrix}, \\
   \text{and} \quad \phi_t &=  \begin{bmatrix} u_{t} \\
    \util_{t} \otimes x_{t}\end{bmatrix}, 
\end{aligned}
\end{equation}
where $\util_t$ is as defined in \eqref{eqn:util_omega_def}.
Note that when the state $x_t$ is known, the approximation error satisfies $z_t = 0$ and, thus, and we estimate $G$ in \eqref{eqn:G_phi_t_BDS_FO} by regressing $\{x_{t+1}\}_{t=1}^{T}$ on $\{\phi_t\}_{t=1}^{T}$.

\subsubsection{Input Choice \& Stability of Bilinear Dynamical Systems}\label{sec:stability}

Stability of bilinear dynamical systems is typically input dependent. 
To see that, we can unroll the state dynamics in \eqref{eqn:BDS-FO} to write for all $t \ge 0$
\begin{equation}
\begin{aligned}\label{eqn:bilinear sys state}
        x_{t+1} & = \sum_{\ell = 0}^t \left(\prod_{k = 0}^{\ell-1} (u_{t-k} \circ A)\right) \left(B u_{t-\ell}  +  w_{t-\ell}\right). 
\end{aligned}
\end{equation} 
Observe that the products of matrices $\prod_{k = 0}^{\ell-1} (u_{t-k} \circ A)$ may grow exponentially in norm if we consistently choose large inputs.
This is precisely why stability in bilinear dynamical systems is more challenging than other classes of systems such as linear dynamical systems or switched systems.

Traditionally, notions like Mean Square Stability (MSS) have been considered to reason about the stability behavior of bilinear systems \citep{kubrusly1985mean, pardalos2010optimization, sattar2022finite}. Typically, these notions are asymptotic in nature, require distributional assumptions on the inputs,
permit diverging trajectories with nonzero probability, and may not allow us to obtain tight guarantees. We introduce an alternative notion of stability that naturally generalizes the classical notion of stability in standard LTI systems.

\paragraph{Uniform stability in bilinear dynamical systems} First, let us recall that the \emph{joint spectral radius} of a set of matrices $\Mcal \subseteq \RR^{d \times d}$ can be defined as  
\begin{equation}
    \rho(\Mcal)  := \lim_{k \to \infty}\sup_{M_{1}, \dots, M_k \in \Mcal}  \norm {M_{1} M_{2} \cdots M_{k}}_\op^{1/k}.
\end{equation}
For $\rho > 0$, we define
$
    \kappa(\Mcal, \rho) := \sup_{k \ge 1, M_1, \dots, M_k \in \Mcal} \frac{\norm{M_{1} M_{2} \cdots M_{k}}_\op}{ \rho^k}.
$
The quantity $\kappa(\Mcal, \rho)$ is defined in similar vein to that by \citet{mania2019certainty} for LDS, and it captures the transient behavior of a system with state transition matrices varying in $\Mcal$. Note that if $\rho(\Mcal) < 1$, then for any $\rho > \rho(\Mcal)$, the quantity $\kappa(\Mcal, \rho)$ is finite. Now, given a set $\Ucal \subseteq \RR^{d_u}$, we denote  $\Ucal \circ A := \lbrace  A_0 + \sum_{i=1}^{d_u} (u)_i A_i: u \in \Ucal \rbrace$ and introduce the following definition of stability.

\begin{definition}[$( \Ucal, \kappa, \rho)$-uniform-stability]\label{def:stability} 
    Let $\Ucal \subseteq \RR^{d_u}$, $\kappa \ge 1$, and $ 0< \rho < 1$. We say that a bilinear dynamical system (as defined in \eqref{eqn:BDS-FO}) with state-transition matrices $\Acal := \lbrace A_0, \dots, A_{d_u} \rbrace$ is $( \Ucal, \kappa, \rho)$-uniformly-stable, if the joint spectral radius of the set $\Ucal \circ \Acal$ satisfies: (i) $\rho(\Ucal \circ \Acal) < \rho< 1$; and (ii) $\kappa(\Ucal\circ \Acal, \rho)  \le \kappa $.
\end{definition}
Note that, if there exists a nonempty and bounded set $\Ucal \subseteq \RR^p$ such that $\rho(\Ucal \circ \Acal)< 1$, then for any $\rho(\Ucal \circ \Acal) < \rho < 1$, the system is $(\Ucal, \kappa, \rho)$-uniformly stable with $\kappa = \kappa(\Ucal \circ \Acal, \rho) \vee 1$.  
Furthermore, we note that Definition \ref{def:stability} naturally generalizes that introduced by \citet{monfared2023stabilization}.
Indeed, there the authors assume that there exists $u^\star$ such that $\rho(u^\star \circ A) < 1$. This is equivalent to assuming that their system is $(\lbrace u^\star\rbrace, \kappa, \rho)$-uniformly-stable for some $\kappa \ge 1$ and $\rho(u^\star \circ A) <\rho < 1$. 
We need stronger requirements on the stability of the system in comparison with \citet{monfared2023stabilization} because we are concerned with the task of identification. 
This requirement stems from the need to have persistence of excitation so that estimation is possible. 
\paragraph{Input choice} We consider that the inputs $\lbrace u_t \rbrace_{t\ge 0}$ are sampled in an i.i.d. manner from some distribution $\Dcal_{u}$ on $\RR^{d_u}$. For ease of exposition, we will focus on the case where inputs are sampled uniformly at random from a sphere of radius $\sqrt{d_u}$, i.e., $u_t \sim \mathrm{Unif}(\sqrt{d_u} \cdot\Scal^{d_u-1})$.
More generally, as long as the inputs are isotropic and are bounded with high probability, our results will still hold at the expense of longer proofs. 
Putting together this input choice with the stability definition, we are now ready to present the assumption we make on the stability of the bilinear system \eqref{eqn:BDS-FO}.
\begin{assumption}[Stability]\label{assump:stability}
    There exists $\kappa \ge 1$ and $\rho \in (0,1)$ such that the bilinear dynamical system \eqref{eqn:BDS-FO} is $(\sqrt{d_u}\cdot\Scal^{d_u-1}, \kappa, \rho)$-uniformly stable. 
\end{assumption} 
In view of Assumption \ref{assump:stability}, choosing inputs uniformly at random from $\sqrt{d_u}\cdot \cS^{d_u-1}$ guarantees stability almost surely. 
More generally, we can choose to sample inputs from any set $\Ucal$, as long as the system is stable under such a set in the sense of Definition \ref{def:stability}. However, the quality of estimation depends on whether inputs sampled from $\Ucal$ are persistently exciting or not (see \textsection\ref{sec:onesided}).

\subsubsection{Learning State-Space Parameter Matrix}\label{subsec:est_error_decomp_BDS_FO}
In this subsection, we derive finite-sample guarantees on learning the state-space parameter matrix $G {=} \begin{bmatrix} B & A_0 & A_1 & \cdots & A_{d_u}\end{bmatrix}$ via regressing $\{x_{t+1}\}_{t=1}^T$ on $\{\phi_t\}_{t=1}^T$. 
Before we present our main result on learning $G$, we introduce a few quantities which appears in our sample complexity and error bounds. 
Note that in the case of \eqref{eqn:BDS-FO}, $\phi_t$ also contains $x_t$.
Hence, the properties of $x_t$ directly affect the learning guarantees. Looking at \eqref{eqn:bilinear sys state}, under Assumption~\ref{assump:input+noise} and $\E[w_t w_t^\top] = \sigma^2 I_{d_x}$, the covariance of $x_t$ takes the form
\begin{equation}
    \Gamma_x^{(t)}:= \E[x_t x_t^\T]  = \sum_{\ell = 1}^{t} \bar{G}_\ell \bar{G}_\ell^\top + \sigma^2 \sum_{\ell = 1}^{t} \bar{F}_\ell \bar{F}_\ell^\top, \label{eqn:BDS_covariance}
\end{equation}
\ifthenelse{\equal{\version}{arxiv}}{
where $\bar{G}_1 := B $, $\bar{G}_{\ell} := \lbrace A_{i_1} A_{i_2} \cdots A_{i_{\ell-1}} B  \rbrace_{i_1,i_2 \dots, i_{\ell -1}  \in \lbrace 0, \dots, d_u\rbrace} $ for $\ell \in \lbrace 2, \dots, t \rbrace$. 
$\{\bar{G}_\ell\}_{\ell=1}^{t}$ contains $\lbrace B, \lbrace A_{i_1} B \rbrace_{i_1 \in \lbrace 0, \dots, d_u \rbrace},$ $ \dots, \lbrace A_{i_1} A_{i_2}  \cdots A_{i_{t-1}} B \rbrace_{i_1, i_2 \dots, i_{t - 1}  \in \lbrace 0, \dots, d_u\rbrace} \rbrace $. The matrices $\{\bar{F}_\ell\}_{\ell=1}^{t}$ are defined similarly, with $B$ replaced by $I_{d_x}$. Similarly, the covariance of $\phi_t$ can be expressed as,
}{
where $\bar{G}_1 {:=} B $, $\bar{G}_{\ell} {:=} \lbrace A_{i_1}\cdots A_{i_{\ell-1}} B  \rbrace_{i_1 \dots, i_{\ell -1}  \in \lbrace 0, \dots, d_u\rbrace} $ for $\ell \in \lbrace 2, \dots, t \rbrace$. 
$\{\bar{G}_\ell\}_{\ell=1}^{t}$ contains $\lbrace B, \lbrace A_{i_1} B \rbrace_{i_1 \in \lbrace 0, \dots, d_u \rbrace},$ $ \dots, \lbrace A_{i_1} A_{i_2}  \cdots A_{i_{t-1}} B \rbrace_{i_1, i_2 \dots, i_{t - 1}  \in \lbrace 0, \dots, d_u\rbrace} \rbrace $. The matrices $\{\bar{F}_\ell\}_{\ell=1}^{t}$ are defined similarly, with $B$ replaced by $I_{d_x}$. Similarly, the covariance of $\phi_t$ can be expressed as,
}
\begin{align}
    \Gamma_\phi^{(t)}:= \E[\phi_t\phi _t^\T]  = \diag\left(I_{d_u}, I_{d_u+1} \otimes \Gamma_x^{(t)}\right). \label{eqn:BDS_phi_covariance}
\end{align}
We are now ready to use the estimation error decomposition in \eqref{eqn:estimation_error_decomposition} to derive finite-sample learning guarantees for the bilinear system~\eqref{eqn:BDS-FO} as follows.

\begin{theorem}[Learning state-space matrices]\label{thm:main_BDS_FO}
    Fix $\delta \in (0,1)$. Suppose Assumptions~\ref{assump:input+noise} and \ref{assump:stability} hold, and $\E[w_t w_t^\top] = \sigma^2 I_{d_x}$ for all $t \in [0,T]$.
    Given a single input-state trajectory $\{(u_t,x_t)\}_{t=0}^T$ of the system~\eqref{eqn:BDS-FO}, construct $\{\phi_t\}_{t=1}^T$ according to \eqref{eqn:G_phi_t_BDS_FO}.
    Let $\Ghat = ( \sum_{t=1}^T   x_{t+1} \phi_{t}^\top) (\sum_{t=1}^{T} \phi_{t} \phi_{t}^\top )^{\dagger}$ be the least-squares estimator (LSE) of $G$ in \eqref{eqn:G_phi_t_BDS_FO}. Fix an integer $\tau>0$, and let $\Gamma_x^{(\tau)}$, $\Gamma_\phi^{(\tau)}$ be as in \eqref{eqn:BDS_covariance} and \eqref{eqn:BDS_phi_covariance}, respectively. 
    Setting $d_{\phi} := d_u + (d_u+1) d_x$, the event:
\ifthenelse{\equal{\version}{arxiv}}{
    \begin{equation}\label{eq:error bound}
         \norm {\Ghat  - G }_\op \leq c_0 \sigma \mysqrt[1pt]{\frac{d_\phi\log\left(\frac{d_\phi}{\delta}\right) + \log\left(\det\left({\Gamma_\phi^{(T)}}/{\lambda_{\min}\left(\Gamma_{\phi}^{(\tau)}\right)}\right)\right)}{T\left(1\land \lambda_{\min}\left(\Gamma_{x}^{(\tau)}\right)\right)}}
    \end{equation}
    }{
    \begin{equation}
          \begin{aligned}\label{eq:error bound}
          &\norm {\Ghat  - G }_\op \leq \\
          &c_0 \sigma \mysqrt[1pt]{\frac{d_\phi\log\left(\frac{d_\phi}{\delta}\right) + \log\left(\det\left({\Gamma_\phi^{(T)}}/{\lambda_{\min}\left(\Gamma_{\phi}^{(\tau)}\right)}\right)\right)}{T\left(1\land \lambda_{\min}\left(\Gamma_{x}^{(\tau)}\right)\right)}}
          \end{aligned}
    \end{equation}
    }
    holds with probability at least $1-\delta$, provided that
\ifthenelse{\equal{\version}{arxiv}}{
    \begin{align}
    T &\gtrsim  \tau \bar{\gamma} \bigg(\log\left(\frac{\tau}{\delta}\right) + d_\phi \log\bigg(1+\frac{d_u + 2d_u\tr(\Gamma_x^{(\tau)})}{\delta}\bigg)  \bigg),  \label{eqn:trajectory_size_main_thm II}
    \end{align}
    }{
    \begin{align}
    T &\gtrsim  \tau \bar{\gamma} \bigg(\log\left(\frac{\tau}{\delta}\right) + d_\phi \log\bigg(1+\frac{d_u + 2d_u\tr(\Gamma_x^{(\tau)})}{\delta}\bigg)  \bigg), \nn  \label{eqn:trajectory_size_main_thm II}
    \end{align}
    }
where $c_0>0$ is a universal constant, and $\tau$, $\bar \gamma$ are system dependent constants defined as
\begin{align*}
    \tau &:= \left\lceil 1 + \polylog\left(d_x,d_u,\kappa,\rho,\sigma,\opn{B},\lambda_{\min}\left(\Gamma_x^{(\tau)}\right)\right) \right\rceil, \\
    \bar \gamma &:= \bigg(1 +  \frac{\kappa^2(\sigma +\norm{B}_\op)^2}{(1-\rho)^2\lambda_{\min}\left(\Gamma_x^{(\tau)}\right)} \bigg)^2. 
\end{align*}  
\end{theorem}
\ifthenelse{\equal{\version}{arxiv}}{
The proof of Theorem~\ref{thm:main_BDS_FO} is deferred to Appendix~\ref{app:BDS-FO}. 
}{
The proof of Theorem~\ref{thm:main_BDS_FO} is deferred to Appendix D.2 in \cite{sattar2026learning}. 
}
The result in Theorem~\ref{thm:main_BDS_FO} is novel, and it is derived using the results and tools introduced in the previous sections.
As compared to~\cite{sattar2022finite} (which requires Gaussian inputs/noise, mean square stability, and $B=0$), Theorem~\ref{thm:main_BDS_FO} holds for sub-Gaussian noise, non-zero activation matrix $B$, and uniform stability.
Specifically, \cite{sattar2022finite} shows that for Gaussian inputs/noise, the nonlinear features $\xtil_t = \util_t \otimes x_t$ satisfy the block Martingale small ball condition~\citep{simchowitz2019learning} with a block of length $k=1$. 
On the other hand, Theorem~\ref{thm:main_BDS_FO} is derived by proving that the nonlinear features $\phi_t = [u_t^\T~~\util_t^\T \otimes x_t^\T]$ satisfy the hypercontractivity condition (Def.~\ref{def:hypercont}) with $\gamma \leq \bar \gamma$ under uniform stability.
Combining this with the results of \S\ref{sec:onesided} guarantees persistence of excitation, which is then combined with self-normalized martingale bounds for vector-valued noise processes to get the desired result.
The errors due to noise and truncation are bounded using martingale-based concentration arguments.
Combining these gives the desired learning guarantees for~\eqref{eqn:LDS-BO}.

\subsection{Partially Observed Bilinear Dynamical Systems}\label{subsec:sysid_partial_bds}
In this subsection, we consider the class of partially-observed bilinear dynamical systems, which admits the state-space representation~\eqref{eqn:BDS-PO}. Recall from Section~\ref{subsec:bilinear_feature_map} that, fixing a history length $\tau{>}0$, \eqref{eqn:BDS-PO} can be alternately represented as the nonlinear input-output map $y_t = G \, \phi_t + z_t + \eta_t$ with

\ifthenelse{\equal{\version}{arxiv}}{
\begin{equation}
    \begin{aligned}\label{eqn:G_phi_r_eta_BDS_PO}
        G &= \begin{bmatrix} D & G_1 & G_2 & \cdots & G_{\tau-1}\end{bmatrix}, \quad\quad
        \phi_t = \begin{bmatrix} u_{t} \\u_{t-1} \\  \util_{t-1} \otimes u_{t-2} \\ \util_{t-1} \otimes \util_{t-2} \otimes u_{t-3} \\ \vdots \\  \util_{t-1} \otimes \util_{t-2} \otimes \cdots \otimes u_{t-\tau+1}  \end{bmatrix},
        \\
         z_t &=  C  \left( \prod_{\ell = 1}^{\tau-1} (u_{t-\ell}\circ A )  \right) x_{t-\tau+1}, \quad \quad ~
         \eta_t = F_t \, \omega_t,
    \end{aligned}
\end{equation}
}{
\begin{equation}
    \begin{aligned}\label{eqn:G_phi_r_eta_BDS_PO}
        G &= \begin{bmatrix} D & G_1 & G_2 & \cdots & G_{\tau-1}\end{bmatrix}, \\
        \phi_t &= \begin{bmatrix} u_{t} \\u_{t-1} \\  \util_{t-1} \otimes u_{t-2} \\ \util_{t-1} \otimes \util_{t-2} \otimes u_{t-3} \\ \vdots \\  \util_{t-1} \otimes \util_{t-2} \otimes \cdots \otimes u_{t-\tau+1}  \end{bmatrix},
        \\
         z_t &=  C  \left( \prod_{\ell = 1}^{\tau-1} (u_{t-\ell}\circ A )  \right) x_{t-\tau+1},
         \\
         \eta_t &= F_t \, \omega_t,
    \end{aligned}
\end{equation}
}
\ifthenelse{\equal{\version}{arxiv}}{
where $G_1 := C B $, $G_{\ell} := \lbrace C A_{i_1} A_{i_2} \cdots A_{i_{\ell-1}} B  \rbrace_{i_1, i_2 \dots, i_{\ell -1}  \in \lbrace 0, \dots, d_u\rbrace}$ for $\ell \in \lbrace 2, \dots, \tau - 1 \rbrace$, and $\tau >0$. 
The parameters $\lbrace C B,$ $\lbrace C A_{i_1} B \rbrace_{i_1 \in \lbrace 0, \dots, d_u \rbrace},\dots,$ $\lbrace C A_{i_1} A_{i_2} \cdots A_{i_{\tau-2}} B \rbrace_{i_1, i_2, \dots, i_{\tau - 2}  \in \lbrace 0, \dots, d_u\rbrace} \rbrace $ are what we refer to as the \emph{Markov-like} parameters of the system. The matrix $G$ takes values in $\R^{d_y \times ((d_u+1)^{\tau-1} + d_u-1)}$. 
}{where $G_1 {:=} C B $, $G_{\ell} {:=} \lbrace C A_{i_1}\cdots A_{i_{\ell-1}} B  \rbrace_{i_1, \dots, i_{\ell -1}  \in \lbrace 0, \dots, d_u\rbrace}$ for $\ell \in \lbrace 2, \dots, \tau - 1 \rbrace$, and $\tau >0$. 
The parameters $\lbrace C B,$ $\lbrace C A_{i_1} B \rbrace_{i_1 \in \lbrace 0, \dots, d_u \rbrace},\dots,$ $\lbrace C A_{i_1} A_{i_2} \cdots A_{i_{\tau-2}} B \rbrace_{i_1, i_2, \dots, i_{\tau - 2}  \in \lbrace 0, \dots, d_u\rbrace} \rbrace $ are what we refer to as the \emph{Markov-like} parameters of the system. The matrix $G$ takes values in $\R^{d_y \times ((d_u+1)^{\tau-1} + d_u-1)}$. 

}
The vectors $\util_t$ and $\omega_t$ are as defined in \eqref{eqn:util_omega_def}.
Lastly, the matrix $F_t$ appearing in the noise process $\{\eta_t\}_{t \geq 0}$ is given by
\begin{equation}
\begin{aligned}\label{eqn:Ft_omegat}
    F_t  &:= \begin{bmatrix}
     I_{d_y} &
    C   &
    C \,  (u_{t-1}\circ A )  &
    \cdots &
    C \prod_{\ell = 1}^{\tau-2} (u_{t-\ell}\circ A ) 
 \end{bmatrix}.  
\end{aligned}
\end{equation}

\subsubsection{Learning Markov-like Parameter Matrix}\label{subsec:est_error_decomp_BDS_PO}
In this subsection, we derive finite-sample guarantees on learning \emph{Markov-like} parameter matrix $G$ (given by \eqref{eqn:G_phi_r_eta_BDS_PO}) via regressing $\{y_t\}_{t=\tau+1}^T$ on $\{\phi_t\}_{t=\tau+1}^T$.
Since the covariates and the residual/noise processes $\{\phi_t\}_{t=\tau+1}^T$, $\{z_t\}_{t=\tau+1}^T$, and $\{\eta_t\}_{t=\tau+1}^T$ are highly dependent, we use the estimation error decomposition in \eqref{eqn:estimation_error} and \eqref{eqn:estimation_error_decomposition} which are upper bounded to get the following result.

\begin{theorem}[Learning Markov-like parameters]\label{thm:main_BDS_PO}
    Fix $\delta \in (0,1)$ and a history length $\tau >0$. Suppose Assumptions~\ref{assump:input+noise} and \ref{assump:stability} hold, and we are given a single input-output trajectory $\{(u_t,y_t)\}_{t=0}^T$ of the system~\eqref{eqn:BDS-PO}. Let $\{\phi_t\}_{t=\tau+1}^T$ be the nonlinear covariates as in \eqref{eqn:G_phi_r_eta_BDS_PO}, and $\Ghat = ( \sum_{t=\tau}^T   y_t \phi_{t}^\top) (\sum_{t=\tau}^{T} \phi_{t} \phi_{t}^\top )^{\dagger}$ be the least-squares estimator (LSE) of $G$ in \eqref{eqn:G_phi_r_eta_BDS_PO}.  Setting $\bar d_{\phi} := 2(d_u+1)^{\tau-1}$, the event:
    \ifthenelse{\equal{\version}{arxiv}}{
    \begin{equation}
    \begin{aligned}\label{eq:error bound}
        \norm {\Ghat  - G }_\op  &\leq \frac{c_3\sigma + c_4 \sqrt{\bar d_\phi} \rho^{\tau-1}}{1-\rho}  \sqrt{\frac{\left(\bar d_\phi  + d_y + \log\left(\frac{\tau}{\delta}\right)\right)}{T-\tau+1}}
    \end{aligned}
    \end{equation}
    }{
    \begin{equation}
    \begin{aligned}\label{eq:error bound}
         \norm {\Ghat  - G }_\op  &\leq \frac{c_3\sigma + c_4 \sqrt{\bar d_\phi} \rho^{\tau-1}}{1-\rho}  \sqrt{\frac{\left(\bar d_\phi  + d_y + \log\left(\frac{\tau}{\delta}\right)\right)}{T-\tau+1}}
    \end{aligned}
    \end{equation}
    }
    holds with probability at least $1-\delta$, provided that
    \ifthenelse{\equal{\version}{arxiv}}{
    \begin{equation*}
         T - \tau+1 \gtrsim  \tau(\tau-1) 3^{\tau} \left(\bar d_{\phi} \log\left(\bar d_{\phi}\right) + \log\left(\frac{\tau}{\delta}\right) \right), 
\end{equation*}
    }{
    \begin{equation*}
        \begin{aligned}
              T - \tau+1 \gtrsim  \tau(\tau-1) 3^{\tau} \left(\bar d_{\phi} \log\left(\bar d_{\phi}\right) + \log\left(\frac{\tau}{\delta}\right) \right), 
        \end{aligned}
\end{equation*}
    }
where $c_3$, $c_4$ are system dependent constants, and there exists a universal constant $c_0 >0$, such that $c_1 = c_0\left(1 + {\kappa \opn{C}}\right)$ and $c_2 = c_0 \kappa \opn{C} \sqrt{\sigma^2 + \Vert B \Vert_\op^2}$.
\end{theorem}
\ifthenelse{\equal{\version}{arxiv}}{
The proof of Theorem~\ref{thm:main_BDS_PO} is deferred to the Appendix~\ref{app:BDS-PO}.
}{
The proof of Theorem~\ref{thm:main_BDS_PO} is deferred to the Appendix D.3 in \cite{sattar2026learning}.
}
From the bound in Theorem~\ref{thm:main_BDS_PO}, we see the recovery error $\Vert \Ghat - G \Vert_\op $ scales as, ignoring all other dependencies,  $\tilde{\cO} (\sqrt{(d_u+1)^{\tau}/(T-\tau)})$.
This contrasts with partially observed linear systems, where typically we only have a polynomial dependence in the history length $\tau$, and also reflects the difficulty in learning bilinear systems from partial observations.
To recover the unknown matrices $C, A_0, \dots, A_{d_u}, B$, we require $\tau$ large enough, typically larger than $2d_{x}$ (see Remark 3 in \cite{sattar2025finite}).

\subsection{Notes}

The early literature on bilinear system identification has for the most part focused on developing estimation procedures. These include least squares and variants thereof \citep{fnaiech1987recursive}, augmented by subspace methods for partially observed systems \citep{favoreel1999subspace,verdult2001identification,verdult2005kernel}, as well as expectation--maximization algorithms \citep{gibson2005maximum}. The asymptotic properties of these identification procedures have received comparatively less attention. A notable contribution in this direction is \citet{chen1996strong}, which established strong consistency and convergence-rate guarantees for extended least-squares estimation of discrete-time stochastic bilinear systems. Alongside the choice of estimator, the role of input design has been investigated in the continuous-time setting. Specifically, \citet{juang2005continuous} developed identification procedures using specially designed inputs, and \citet{SontagWangMegretski2009} characterized input classes sufficient for generic input--output identifiability. This overview is by no means exhaustive, and we refer the reader to the cited works and the references therein for further developments.

In contrast to this established identification literature, non-asymptotic learning guarantees for bilinear dynamical systems remain relatively sparse. \cite{sattar2022finite} provided the first finite-time guarantees for learning state-observed bilinear systems under mean-square stability~\citep{kubrusly1985mean, pardalos2010optimization} and Gaussian inputs/noise.
Specifically, \cite{sattar2022finite} used the tools for deriving learning guarantees for linear dynamical systems, such as block martingale small ball methods \cite{simchowitz2018learning}, and self-normalized martingales~\citep{sarkar2018fast}. More recently, \citet{chatzikiriakos2026endtoend} derived finite-sample identification bounds from i.i.d. data using constant-input experiments that reduce bilinear identification to linear and affine subproblems, and incorporated these bounds into robust controller design. These results also motivate further investigation of input design for learning partially observed bilinear systems. In the linear setting, carefully designed control inputs, rather than i.i.d. Gaussian inputs, can yield strong statistical estimation rates~\citep{sun2022finite}. We believe that for partially observed bilinear dynamical systems, the exponential growth of the number of \emph{Markov-like} parameters, leading to an error rate of $\tilde{\cO} (\sqrt{(d_u+1)^{\tau-1}/(T-\tau)})$, can be avoided by the input design. Lastly, the learning guarantees in \S\ref{sec:sysid} require uniformly stable bilinear systems. In the linear setting, non-asymptotic identification guarantees extend to marginally stable systems under both full and partial observation \citep{simchowitz2018learning, sarkar2018fast, simchowitz2019learning, bakshi2023new}, and to certain classes of unstable systems under full observation \citep{sarkar2018fast,faradonbeh2018finite}. Extending the presented guarantees beyond uniformly stable bilinear systems, even to the marginally stable case, remains an open problem.

\newpage
\part{Control}\label{part:control}

We now turn our discussion to the control of bilinear systems.
First, we discuss the design of optimal controllers under quadratic costs and partial observations.
Unlike for linear systems, 
bilinearity, even only in the observations, complicates the optimal control solution.
In particular, we show that the well-known \emph{separation principle} of linear quadratic control does not apply when the observations are bilinear.
Naively following this principle may even lead to failures of stabilization.
Instead, we propose and numerically investigate a method based on receding horizon control in belief space.

Second, we turn to stabilization of bilinear dynamics under perfect state observation.
Turning from challenges that arise due to noise and partial observation, with bilinearity only in the measurements, we instead focus on the challenges of controlling imperfectly modeled deterministic bilinear dynamics under full state observation. 
We review strategies based on semidefinite programming, using linear matrix inequality (LMI)-based techniques and sum-of-squares (SOS) optimization, for designing controllers and Lyapunov functions.
Finally, we discuss connections between bilinear control and general classes of nonlinear dynamics through Koopman theory, highlighting the broad applications.

\section{Control of Linear Dynamics from Partial Bilinear Observations}

\renewcommand{\vxhat}{\hat{x}}
\renewcommand{\vx}{x}
\renewcommand{\vw}{w}
\renewcommand{\vz}{z}
\renewcommand{\vu}{u}
\renewcommand{\vy}{y}
\renewcommand{\vSigma}{\Sigma}
\renewcommand{\vA}{A}
\renewcommand{\vC}{C}
\renewcommand{\vB}{B}
\renewcommand{\vL}{L}
\renewcommand{\vK}{K}
\renewcommand{\vQ}{Q}
\renewcommand{\vR}{R}
\renewcommand{\vP}{P}
\renewcommand{\vM}{M}
\renewcommand{\vF}{F}
\renewcommand{\Ical}{\mathcal{I}}
\renewcommand{\Bcal}{\mathcal{B}}
\renewcommand{\Gcal}{\mathcal{G}}
\renewcommand{\Ncal}{\mathcal{N}}
\renewcommand{\cN}{\mathcal{N}}

We consider quadratic control of \eqref{eqn:LDS-BO} with Gaussian noise (BO-LQG), defined in the finite-horizon optimal control problem
\ifthenelse{\equal{\version}{arxiv}}{
\begin{equation}
\begin{aligned}\label{eqn:bilinear LQG}
    \min_{\mu_{0:T-1}} &\E \left[ x_T^\T Q_T x_T + \sum_{t=0}^{T-1} \left( x_t^\top Q x_t {+} u_t^\top R u_t \right) \right]  \\
	\text{s.t.} \quad &\eqref{eqn:LDS-BO}, \quad \text{and} \quad u_t = \mu_t(\vu_0, \dots, \vu_{t-1} ,\vy_0, \dots, \vy_{t-1}),
\end{aligned} \tag{BO-LQG}
\end{equation}
}{
\begin{equation}
\begin{aligned}\label{eqn:bilinear LQG}
    \min_{\mu_{0:T-1}} &\E \left[ x_T^\T Q_T x_T + \sum_{t=0}^{T-1} \left( x_t^\top Q x_t {+} u_t^\top R u_t \right) \right]  \\
	\text{s.t.} \quad &\eqref{eqn:LDS-BO}, \\
    &u_t = \mu_t(\vu_0, \dots, \vu_{t-1} ,\vy_0, \dots, \vy_{t-1}),
\end{aligned} \tag{BO-LQG}
\end{equation}
}
where $Q_T,Q \in \R^{d_x \times d_x}$ are given state-cost matrices and $R \in \R^{d_u \times d_u}$ is a given input-cost matrix. 
We additionally define the shorthand
$\Ical_t := \{\vu_0, \dots, \vu_{t-1} ,\vy_0, \dots, \vy_{t-1} \}$
as the information available at time $t \geq 0$.
Moreover, we assume the following.
 \begin{assumption}\label{assump noise, initial, cost} We have: (i) Gaussian noise $\{\vw_t\}_{t=0}^{T-1} \distas \Ncal(0, \vSigma_w)$, $\{v_t\}_{t=0}^{T-1} \distas \Ncal(0, \vSigma_v)$, (ii) initial state $\vx_0 \distas \Ncal(\vxhat_0, \vSigma_0)$, and (iii) positive definite cost matrices, i.e., $\vQ_T, \vQ, \vR \succ 0$.    
 \end{assumption}

The problem in \eqref{eqn:bilinear LQG} is a slight variation on the classical linear-quadratic Gaussian (LQG) problem; if the matrices $C_1{=}\cdots{=}C_{d_u}{=}0$, it reduces to LQG.
In classical LQG, the optimal controller follows the separation principle, i.e., state estimation and control design can be solved independently.
In particular, the state is estimated via Kalman filtering \citep{kalman1960new}, and the control input is obtained by applying the optimal linear feedback gain to the state estimate.
In the following, we explore how the separation principle can fail in this simple departure from standard LQG.

\subsection{Kalman Filtering}

State estimation for bilinear systems follows the same logic as for linear time-varying systems.
We thus begin by defining the Kalman filter for~\eqref{eqn:LDS-BO}.
For notational convenience, we define the input-dependent observation matrix $C(u_t) := C_0 + \sum_{k=1}^{d_u}(u_t)_k C_k$. 
Let $\vxhat_{t|t- 1} := \E [\vx_t | \Ical_t]$ be the estimated state at time $t$, and $\vSigma_{t|t- 1} := \E \left[(\vx_t - \E [\vx_t | \Ical_t])(\vx_t - \E [\vx_t | \Ical_t])^\T\right]$ be the estimation error covariance.
Then, starting from $\vxhat_{0|- 1} = \vxhat_0$ and $\vSigma_{0|- 1} = \vSigma_0$, the estimated state and the covariance follow the recursion:
\ifthenelse{\equal{\version}{arxiv}}{
\begin{equation}
\begin{aligned}\label{eq:kf_bo}
    \vxhat_{t + 1|t} &= \vA \vxhat_{t|t-1} + \vB \vu_t - \vL(\vu_t)\left(\vy_t - \vC(\vu_t) \vxhat_{t|t-1} \right), \\
    \vSigma_{t+ 1|t} &= \vA \vSigma_{t|t - 1} \vA^\T  + \vL(\vu_t)\vC(\vu_t) \vSigma_{t|t - 1} \vA^\T + \vSigma_w,\\
    \text{where,} \quad \vL(\vu_t) &:= -  \vA \vSigma_{t|t - 1} \vC(\vu_t)^\T \big( \vC(\vu_t) \vSigma_{t|t - 1} \vC(\vu_t)^\T + \vSigma_v \big)^{-1}.
\end{aligned}\tag{KF-BO} 
\end{equation}
}{
\begin{equation}
\begin{aligned}\label{eq:kf_bo}
    \vxhat_{t\!+\!1|t} &= \vA \vxhat_{t|t\!-\!1} {+} \vB \vu_t {-} \vL(\vu_t)\left(\vy_t {-} \vC(\vu_t) \vxhat_{t|t\!-\!1} \right), \\
    \vSigma_{t\!+\! 1|t} &= \vA \vSigma_{t|t\!-\! 1} \vA^\T  {+} \vL(\vu_t)\vC(\vu_t) \vSigma_{t|t\!-\! 1} \vA^\T {+} \vSigma_w,\\
    \vL(\vu_t) &:= {-}  \vA \vSigma_{t|t\!-\! 1} \vC(\vu_t)^\T \big( \vC(\vu_t) \vSigma_{t|t\!-\! 1} \vC(\vu_t)^\T {+} \vSigma_v \big)^{\!-1}.
\end{aligned}\tag{KF-BO} 
\end{equation}
}
The matrix $\vL(\vu_t)$ is the input-dependent Kalman gain at time $t$.  
Unlike Kalman filtering from linear measurements, both the Kalman gain $\vL(\vu_t)$ and the error covariance $\vSigma_{t + 1|t}$ depend on the control inputs up to time $t$. 
Nonetheless, it is easy to show that the Kalman filter provides the full posterior state distribution. 
The following lemma readily follows from~\cite{kalman1960new} by replacing $\vC_t$ with $\vC(\vu_t)$, which is a deterministic quantity conditioned on $\vu_t$. 

\begin{lemma}[Optimality of KF] \label{lemma:optimality of KF} Under Assumption~\ref{assump noise, initial, cost}(i) \& (ii), the Kalman filtering algorithm gives
the full posterior distribution of the state $\vx_t$ given the information $\Ical_t$, $\vx_t | \Ical_t \sim \Ncal(\vxhat_{t|t- 1}, \vSigma_{t|t- 1})$.
\end{lemma}

Lemma~\ref{lemma:optimality of KF} readily follows from~\cite{kalman1960new} by replacing $\vC_t$ with $\vC(\vu_t)$ which is a deterministic quantity conditioned on $\vu_t$. Moreover, since the posterior distribution of the state $\vx_t$ is Gaussian, using similar argument as~\cite[Section 4.3.1]{bertsekas2012dynamic}, we note that the values $\vxhat_{t|t- 1}, \vSigma_{t|t- 1}$ computed by the Kalman filter are sufficient statistics for any policy.

\begin{figure*}[t]
\setlength{\abovecaptionskip}{0pt}
\setlength{\belowcaptionskip}{0pt}
\begin{center}
\begin{tabular}{ c @{\hspace{0.25cm}} c @{\hspace{0.25cm}} c @{\hspace{0.25cm}} c}
\includegraphics[scale=0.20]{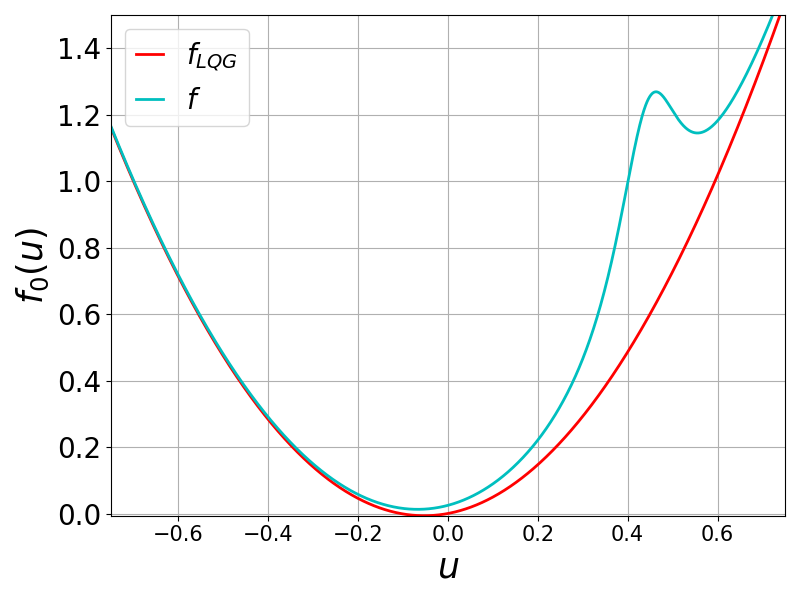} &
\includegraphics[scale=0.20]{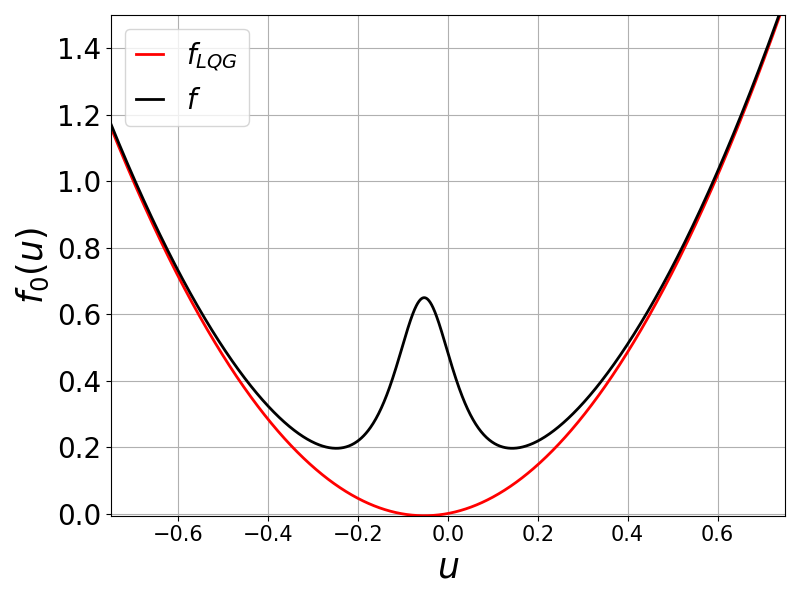} &
\includegraphics[scale=0.20]{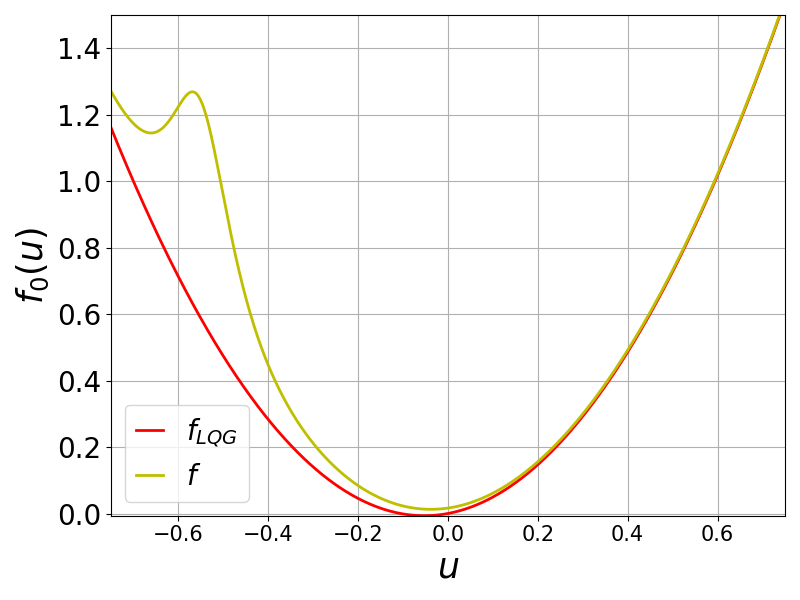} &
\includegraphics[scale=0.20]{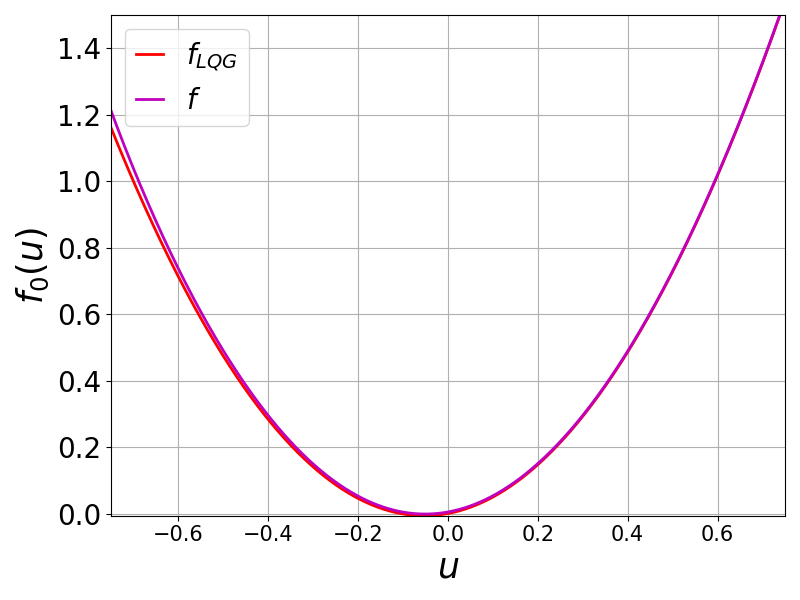} \\
\small \quad (a) $ u_\text{LQG}= -\frac{C_0}{C_1} - 0.5$ &      \small \quad (b) $ u_\text{LQG} = -\frac{C_0}{C_1} $ & \small \quad (c) $ u_\text{LQG}= -\frac{C_0}{C_1} + 0.5$ &      \small \quad(d) $u_\text{LQG} = -\frac{C_0}{C_1} + 1.0$ 
\end{tabular}
\vspace{5pt}
\caption{The landscape of $f(u)=f_\text{LQG}(u) + g(u)$ depends on the distance between the global minima of $f_\text{LQG}(u)$ and the global maxima of $g(u)$, which depends on how far the LQG input is from an input which leads to loss of observability, which in the scalar case is given by $-\frac{C_0}{C_1}$.
}
\label{figure1}
\end{center}
\vspace{-12pt}
\end{figure*} 

\subsection{Separation Principle} \label{subsec:sp}

The certainty equivalent separation principle treats the state estimate $\hat x_{t|t-1}$ as if it were true.
As a result, the policy follows the linear-quadratic regulation (LQR) state-feedback rule.
We refer to this as the LQG controller.
At every time $0 \leq t \leq T{-}1$, the LQG controller is given by
\begin{equation}
\begin{aligned}\label{eqn:LQG Policy}
    \vu_t^{LQG} &= \vL_t \vxhat_{t|t\!-\! 1}, \\
    \text{where}\quad \vL_t &= -(\vB^\T \vK_{t+1} \vB + \vR)^{-1}\vB^\T \vK_{t+1} \vA,
\end{aligned}
\end{equation}
and the matrices $\vK_t$, starting from $\vK_T = \vQ_T$, are given recursively by the Riccati equation
\begin{equation}
\begin{aligned}\label{eqn:LQG Riccati}
     \vP_t &= \vA^\T \vK_{t+1} \vB (\vB^\T \vK_{t+1} \vB + \vR)^{-1}\vB^\T \vK_{t+1} \vA, \\
     \vK_t &= \vA^\T \vK_{t+1} \vA - \vP_t + \vQ.
\end{aligned}
\end{equation}

This strategy computes actions according only to state estimates with no accounting for the effect of the control input on observation.
It is immediate to observe failure modes of this approach, especially when $C_0=0$.
First, consider the case that $\hat x_0$, the mean of the initial state's prior distribution, is equal to zero. 
Then, we have $u_0=0$, $C(0)=0$, and, hence, $L(u_t)=0$ and $\hat x_{1|0}=0$.
This pattern continues, and no control input is applied.
If the matrix $A$ has spectral radius larger than one, the state grows exponentially over the time horizon and so does the accumulated cost.
Even when $\hat x_0\neq 0$, similar issues can emerge, as illustrated in the following example.

\ifthenelse{\equal{\version}{arxiv}}{
\begin{example}
\label{ex:KC-zero}
Consider a system with $A=1$, $B=\begin{bmatrix}0 & 1\end{bmatrix}$, $C_0=C_2=0$ and $C_1=1$.
Suppose
$Q = I$ and $R=1$. 
The optimal state-feedback control gains satisfy $K_t e_1 = 0$ since the first control input incurs a cost but does not affect the state.
As a result, $(u_t)_1 = 0$ and thus $C(u_t) = 0$ for all $t \geq 0$.
\end{example}
}{
\begin{example}
\label{ex:KC-zero}
Consider a system with $A{=}1$, $B{=}\begin{bmatrix}0 & 1\end{bmatrix}$, $C_0{=}C_2{=}0$ and $C_1{=}1$.
Suppose
$Q {=} I$ and $R{=}1$. 
The optimal state-feedback control gains satisfy $K_t e_1 {=} 0$ since the first control input incurs a cost but does not affect the state.
As a result, $(u_t)_1 {=} 0$ and thus $C(u_t) {=} 0$ for all $t {\geq} 0$.
\end{example}
}

In this example, no information about the state is ever observed.
As a result, using the separation principle is basically equivalent to applying the optimal open-loop sequence of actions. 
It is not difficult to verify that alternative linear policies can achieve a lower quadratic cost by preventing this loss of observability.

We now show that the sub-optimality of the separation principle indeed holds more generally by analyzing the dynamic programming procedure.
Using the same proof technique as \citet[Section 4.2]{bertsekas2012dynamic}, the optimal policy for the last stage is given by $\vu_{T-1}^\star = \vu_{T-1}^{LQG}$, which is the LQG controller given by~\eqref{eqn:LQG Policy}. 
However, things change for $t \leq T-2$.
Following the dynamic programming algorithm and applying the Woodbury matrix identity, \ifthenelse{\equal{\version}{arxiv}}{
\begin{align*}
    &\vu_{T-2}^\star = \arg \min_{\vu_{T-2}} \left\{ f(\vu_{T-2}):= f_{LQG}(\vu_{T-2}) + g(\vu_{T-2}) \right\},\\
     & f_{\rm LQG}(\vu_{T-2}) := \vu_{T-2}^\T \Acal \vu_{T-2} + 2 \vxhat_{T-2|T-3}^\T\Bcal^\T \vu_{T-2}, \\
      &g(\vu_{T-2}):=  \tr \big(\Gcal \big(\vSigma_{T-2|T-3}^{-1} + \vC(\vu_{T-2})^\T \vSigma_v^{-1}\vC(\vu_{T-2})\big)^{-1} \big),
\end{align*}
}{
\begin{align*}
    &\vu_{T\!-\!2}^\star = \arg \min_{\vu_{T\!-\!2}} \left\{ f(\vu_{T\!-\!2}):= f_{LQG}(\vu_{T\!-\!2}) + g(\vu_{T\!-\!2}) \right\},\\
     & f_{\rm LQG}(\vu_{T\!-\!2}) := \vu_{T\!-\!2}^\T \Acal \vu_{T\!-\!2} {+} 2 \vxhat_{T\!-\!2|T\!-\!3}^\T\Bcal^\T \vu_{T\!-\!2}, \\
      &g(\vu_{T\!-\!2}):=  \tr \big(\Gcal \big(\vSigma_{T\!-\!2|T\!-\!3}^{\!-1} {+} \vC(\vu_{T\!-\!2})^\T \vSigma_v^{-1}\vC(\vu_{T\!-\!2})\big)^{\!-1} \big),
\end{align*}
}
where we define $\Acal := \vB^\T \vK_{T-1}\vB + \vR$, $\Bcal :=  \vB^\T \vK_{T-1} \vA$, and $\Gcal := \vA^\T\vP_{T-1}\vA$ for notational convenience. 
The second term $g(\vu_{T-2})$ corresponds to the estimation error covariance at time $T-1$, and it depends on control input $\vu_{T-2}$ via the input dependent observation matrix $\vC(\vu_{T-2})$.
This dependence is the precise violation of the separation principle.

\begin{theorem}\label{thm:summary_sep}
The LQG controller may locally maximize the cost, and
    the optimal control policy $\mu^\star_t(\mathcal I_t)$ is not affine in the estimated state $\hat x_{t|t-1}$ in general.
\end{theorem}

\begin{proof}[Proof Sketch]
    Consider $T=2$ and the dynamic programming algorithm described above. 
    Let $u_\text{LQG}= \mathcal A^{-1}  \mathcal B \hat x$ and suppose that $C(u_\text{LQG})=0$.
    In this case, the minimum of $f_{LQG}$ coincides with the maximum of $g$. 
    Depending on their relative magnitudes, this can result in a local maximum for their sum $f=f_{LQG}+g$.
    Concretely, in the scalar case $n=m=p=1$, we may compute the second derivative of $f$ and evaluate it at $u_\text{LQG}$:
$
f''(u_\text{LQG}) = 2\Acal - (2\Gcal C_1^2 \Sigma_{0}^2)/\Sigma_v.
$ This is negative when $\mathcal G C_1^2 \Sigma_{0}^2 > \mathcal A \Sigma_v$, i.e.\ when the measurement noise is small relative to the state uncertainty.
In the scalar setting, one may further characterize all critical points of $f$ as the roots of a degree 5 polynomial. By verifying that there does not exist an affine control law which satisfies this critical point equation for all $\hat x$, we conclude that the optimal policy is not generally affine.
\end{proof}

Theorem~\ref{thm:summary_sep} summarizes Theorems 1 and 2 in~\cite{sattar2025sub}.
Finding an analytical expression for the optimal controller is challenging in general because of: (i) non-convexity, (ii) nonlinearity, and (iii) the existence of multiple critical points.
We investigate the trade-off numerically in Figure~\ref{figure1} by plotting $f_\text{LQG}$ and $f$ in a one dimensional setting.
We see cases where the LQG control is optimal, sub-optimal, or even a local maximum. 
These cases vary based on how close the LQG input is to an input which causes loss of observability.

Thus we conclude that control designed based on the separation principle is not generally reliable for BO-LQG.
It is natural to ask whether iteratively linearizing the dynamics around the current state estimate and following LQG control could recover a better controller.
We argue that it cannot: in the BO-LQG setting, this iterative linearization scheme converges after a single step to the LQG controller analyzed above.
The reason is that, for any \emph{fixed} input $u_t$, the observation model is already exactly linear in $x_t$; the state dynamics are also already linear.
Consequently, linearizing introduces no additional error, so the ``iterative'' procedure terminates immediately, and the resulting controller coincides with the suboptimal policy analyzed in this section.

\subsection{Model Predictive Control in Belief Space} \label{sec:Blief-space MPC}

An alternative notion of ``separation'' in optimal control is that of separating the computation of the posterior distribution of the state from the computation of the control input.
This idea motivates us to reformulate \eqref{eqn:bilinear LQG} in terms of the values $(\hat{\vx}_{t|t-1}, \vSigma_{t|t-1})=:b_t $ computed by the Kalman filter. 
By a similar argument as~\citet[Section 4.3.1]{bertsekas2012dynamic}, this \emph{belief state} is a sufficient statistic for any control policy.
Belief space planning is a common strategy in robotics to handle practical issues like path planning under perception uncertainties \citep{platt2010belief}.
We take inspiration from this perspective and reformulate the original optimal control problem into a belief space control problem.

First, the original finite-horizon quadratic cost \eqref{eqn:bilinear LQG} can be
rewritten in terms of the belief state
\begin{align*}
    \E \big[ x_t^\top Q x_t \mid \Ical_t \big] &= \hat x_{t|t-1}^\top Q \hat x_{t|t-1} + \tr(Q\Sigma_{t|t-1}) .
\end{align*}
This leads to the equivalent optimal control problem in belief space
\begin{equation}
\begin{aligned}\label{eqn:bilinear LQG II}
    \min_{\mu_{0:T-1}} &\E \Big[ \underbrace{\hat x_{T|T-1}^\top Q_T \hat x_{T|T-1} + \tr(Q\Sigma_{t|t-1}) }_{\phi(b_T)}\Big] \\
    &+ \E \Big[ \sum_{t=0}^{T-1}  \underbrace{\left(\hat x_{t|t-1}^\top Q \hat x_{t|t-1} + \tr(Q\Sigma_{t|t-1})  + u_t^\top R u_t\right) }_{\ell(b_t,u_t)} \Big]  \\
    \text{s.t.} \quad &\eqref{eqn:LDS-BO}, \\
     &\eqref{eq:kf_bo}, \\
    \text{and} \quad &u_t = \mu_t(\hat x_{t|t-1}, \Sigma_{t|t-1}).
\end{aligned} 
\end{equation}
Notice that this is a stochastic optimal control problem with a fully observed (belief) state. Thus, the challenge has transformed from reasoning about partial observation to synthesizing a controller for nonlinear dynamics, which arise from the nonlinear covariance evolution.

We now propose a belief-space model predictive control (MPC) method, \texttt{B-MPC}, that plans over a deterministic surrogate of the belief dynamics induced by the Kalman filter.
Notice that the belief state update is stochastic and depends on the bilinear observation $y_t = C(u_t) x_t + v_t$.
Consider its deterministic surrogate
\begin{equation} \label{eqn:det-belief-update}
    \begin{aligned}
        \bar x_{t+1} &= A \bar x_t + Bu_t, \\
        \bar \Sigma_{t+1} &= A \bar\Sigma_t A^\top +\bar L(u_t)C(u_t) \bar\Sigma_{t}A^\top + \Sigma_w,
    \end{aligned}
\end{equation}
where 
$\bar L(u_t)$ is defined analogously to $L(u_t)$ in \eqref{eq:kf_bo} with $\Sigma_{t|t{-}1}$ replaced by $\bar\Sigma_{t}$.
Let $\bar b_t := (\bar x_t, \bar \Sigma_t)$ and define the finite horizon MPC objective at time $t \geq 0$ as
\begin{equation}\label{eq:mpc_cost}
    J_t^{\text{MPC},H}(u_t,\dots,u_{t{+}H{-}1}; b_t) = \sum_{\tau = t}^{t+H-1} \ell(\bar b_\tau, u_\tau) + \phi(\bar b_{t+H}),
\end{equation} 
where the trajectory $(\bar b_\tau)_{\tau=t}^{t+H}$ is generated by \eqref{eqn:det-belief-update} 
as a function of inputs $u_t,\dots,u_{t+H-1}$, 
starting from $\bar b_t = (\hat x_{t|t-1},\Sigma_{t|t-1})$, the true current belief state.
Algorithm~\ref{alg:B-MPC} presents the MPC strategy.

\begin{algorithm}[ht]
\caption{\texttt{B-MPC}} \label{alg:B-MPC}
\begin{algorithmic}[1]
\Require Trajectory length $T$, planning horizon $H$, initial mean $\hat x_0$ and covariance $\Sigma_0$.
\State Initialize $b_0$ with $\hat x_{0|-1} = \hat x_0$ and $\Sigma_{0|-1} = \Sigma_0$.
\For{$t=0,1,2,\dots,T-1$}
    \State Minimize the cost $J_t^{\text{MPC},H}(~\cdot~;b_t)$  \eqref{eq:mpc_cost} over $(u_t,\dots,u_{t+H-1})$ to obtain  $(u_t^\star,\dots,u_{t+H-1}^\star)$.
    \State Apply $u_t^\star$ in the true system, and observe $y_t$.
    \State Update the belief state $b_{t+1} = (\hat x_{t+1|t}, \Sigma_{t+1|t})$ via \eqref{eq:kf_bo} using $(u_t^\star,y_t)$.
\EndFor
\end{algorithmic}
\end{algorithm}

\subsection{Numerical Experiments}
\ifthenelse{\equal{\version}{arxiv}}{
\begin{figure*}[ht]
    \centering
\begin{subfigure}[b]{0.6\linewidth}
            \centering
            \includegraphics[width=\linewidth]{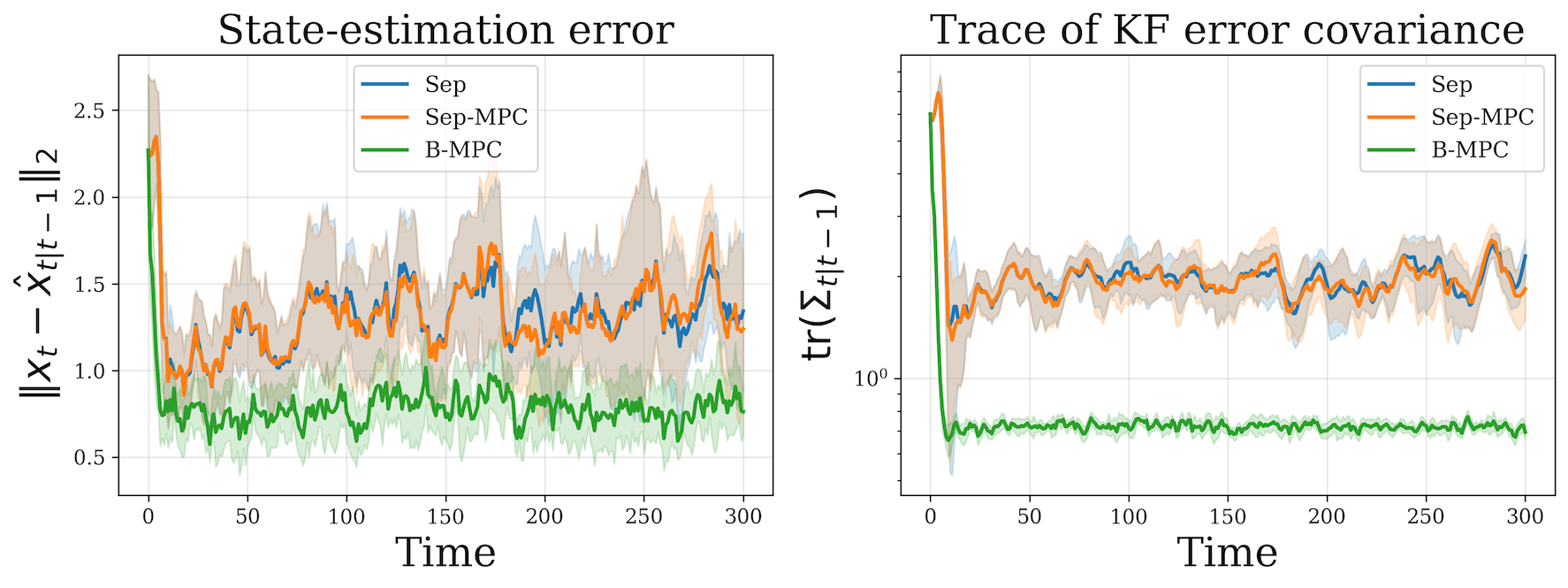}
            \caption{Kalman filter diagnostics.}
            \label{fig:kf-rollout}
        \end{subfigure}
        ~
~
\begin{subfigure}[b]{0.28\linewidth}
            \centering
            \adjincludegraphics[width=\linewidth, trim={0.5\width} 0 0 0, clip]{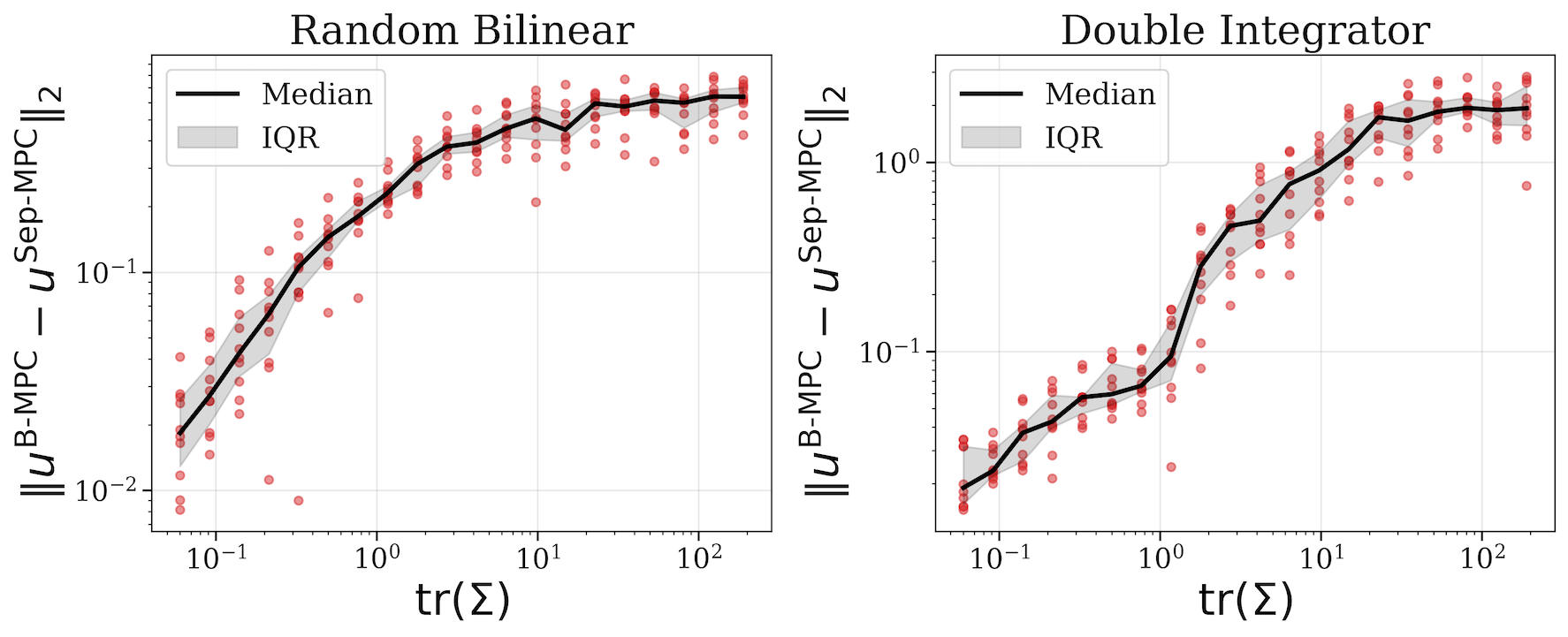}
            \caption{Action difference vs $\text{tr}(\Sigma)$.}
            \label{fig:action-gap-synthetic}
        \end{subfigure}
\begin{subfigure}[b]{0.9\linewidth}
            \centering
\adjincludegraphics[width=\linewidth, trim={0 0 0 0.5\height}, clip]{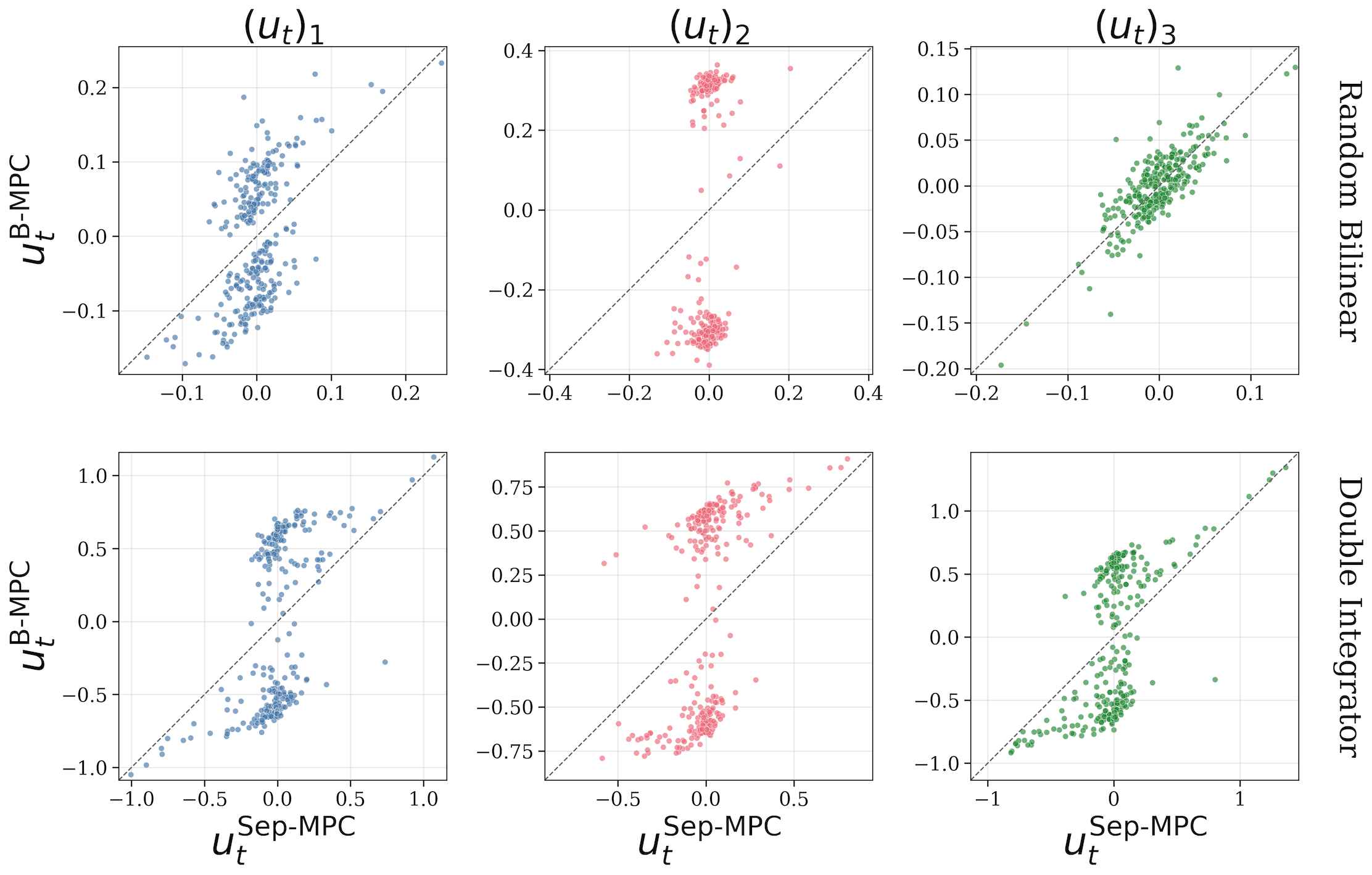}
            \caption{Counterfactual action comparison.}
            \label{fig:input-scatter}
        \end{subfigure}

    \caption{Comparisons between \texttt{Sep}, \texttt{Sep-MPC}, and \texttt{B-MPC}: (a) Kalman filter state estimation error and covariance trace over 10 trials (mean and $95\%$ confidence interval). (b) Action difference ($\ell_2$ norm) versus covariance trace for synthetic belief states. (c) Comparison of counterfactual actions taken by \texttt{Sep-MPC} vs. \texttt{B-MPC} along the same trajectory.}
    \label{fig:main-combined-figure}
\end{figure*}
}{
\begin{figure*}[ht]
    \centering
\begin{subfigure}[b]{0.5\linewidth}
            \centering
            \includegraphics[width=\linewidth]{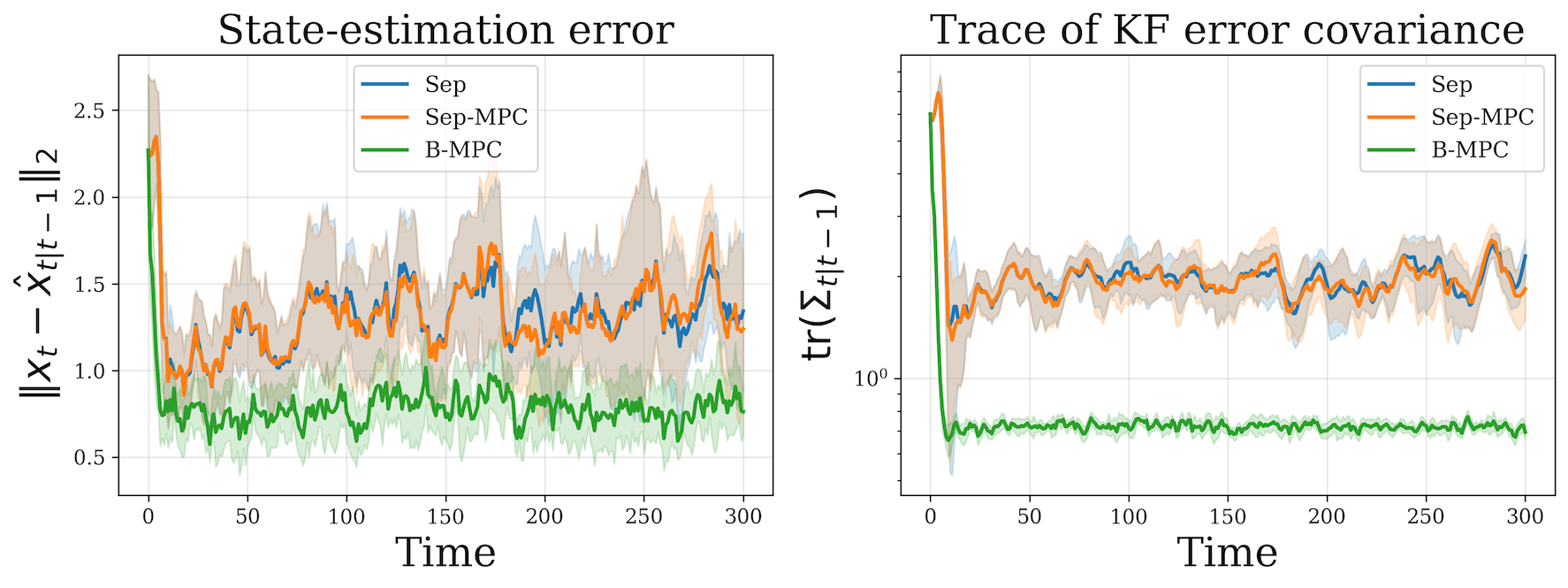}
            \caption{Kalman filter diagnostics.}
            \label{fig:kf-rollout}
        \end{subfigure}
        ~
~
\begin{subfigure}[b]{0.23\linewidth}
            \centering
            \adjincludegraphics[width=\linewidth, trim={0.5\width} 0 0 0, clip]{Figures/action_gap_vs_tr_sigma_synthetic_large.png}
            \caption{Action difference vs $\text{tr}(\Sigma)$.}
            \label{fig:action-gap-synthetic}
        \end{subfigure}
\begin{subfigure}[b]{0.8\linewidth}
            \centering
\adjincludegraphics[width=\linewidth, trim={0 0 0 0.5\height}, clip]{Figures/counterfactual_input_scatter_on_sep_rollout_per_dim_largev2.png}
            \caption{Counterfactual action comparison.}
            \label{fig:input-scatter}
        \end{subfigure}

    \caption{Comparisons between \texttt{Sep}, \texttt{Sep-MPC}, and \texttt{B-MPC}: (a) Kalman filter state estimation error and covariance trace over 10 trials (mean and $95\%$ confidence interval). (b) Action difference ($\ell_2$ norm) versus covariance trace for synthetic belief states. (c) Comparison of counterfactual actions taken by \texttt{Sep-MPC} vs. \texttt{B-MPC} along the same trajectory.}
    \label{fig:main-combined-figure}
\end{figure*}
}

We conclude with numerical comparisons of the performance of controllers based on the separation principle with that of \texttt{B-MPC}.
In particular, we consider \texttt{Sep}, which applies the finite-horizon state-feedback policy $u_t=L_t \hat x_{t|t-1}$ (the separation principle policy in \eqref{eqn:LQG Policy}), and \texttt{Sep-MPC} which computes the state-feedback gain for the same horizon $H$ as \texttt{B-MPC}.
The controller performance is compared on a multi-block double integrator system with $d_x{=}6$, $d_u{=}3$, and $d_y{=}3$, where the scale of $C_0$ is $10^{-2}$.
We use $H{=}15$ and approximately solve the \texttt{B-MPC} minimization using L-BFGS.
Full experimental details and results are presented in~\cite{cao2026dual}. 

The first observation presented in~\cite{cao2026dual} is that \texttt{B-MPC} can achieve cost reductions of close to $40\%$. 
The improvement is due to a reduction in the state component of the cost,
and it comes at the expense of a slightly increased input cost.
Figure~\ref{fig:kf-rollout} explains the performance difference: \texttt{B-MPC} maintains lower state estimation errors by controlling the size of the posterior covariance $\tr(\Sigma_{t|t-1})$. To understand this difference,
we explore how the control inputs differ between the controllers.
First, we generate a single trajectory using \texttt{Sep-MPC} to produce a sequence of belief states $(b_t)_{t=0}^{T-1}$ and inputs $(u_t^\texttt{Sep-MPC})_{t=0}^{T-1}$. Then, for each $b_t$, we compute the counterfactual \texttt{B-MPC} control input  $u_t^\texttt{B-MPC}$.
Figure~\ref{fig:input-scatter} shows that \texttt{B-MPC} inputs are spread away from zero, which can be viewed as information-gathering actions. Recalling that $C_0$ is small, near-zero inputs cause the input-dependent measurement matrix $C(u_t)\approx 0$, leading to reduced state estimation accuracy. 
Finally, we explore reasons for disagreement by sampling synthetic belief states, solving both \texttt{B-MPC} and \texttt{Sep-MPC}, and recording the Euclidean distance between the inputs.
Figure~\ref{fig:action-gap-synthetic} shows the result.
For both systems, the gap between \texttt{B-MPC} and \texttt{Sep-MPC} increases monotonically as $\tr(\Sigma)$ grows.
The larger deviation when uncertainty is high aligns with the information-gathering nature of \texttt{B-MPC} actions.

\subsection{Notes}

Bilinear dynamical systems, including those with partial observations, are classical models.
However, to our knowledge, the particular BO-LDS model was only recently formulated by~\cite{sattar2024learning,liu2025probabilistic} and first studied in the context of system identification.
Control of BO-LDS was studied by~\cite{choi2025explore}, but only for open-loop policies.
In this section, we reviewed the results of \cite{sattar2025sub} which establish that the optimal closed-loop controller does not follow the separation principle.
It remains an open question how to design controllers which guarantee stability, let alone bounded sub-optimality, for BO-LDS. 
The preliminary numerical results from the belief space MPC method proposed by \cite{cao2026dual} provide one path forward.

Belief space planning was coined by \cite{kaelbling1998planning}
in the context of generic partially observed Markov decision processes.
In general, it can be intractable except in relatively small discrete settings.
Nonetheless, it became a common strategy in robotics to handle practical issues like path planning under perception uncertainties.
\cite{platt2010belief} and \cite{van2012motion} formulate belief space planning approximations based on iterative linearization in iLQG. 
More broadly, optimal control under decision-relevant and reducible uncertainty is studied as 
dual control, a term introduced by
\cite{feldbaum1960dual}.
When control inputs have such a dual effect (both regulating and probing),
\cite{bar1974dual} show that the separation principle will not hold. \cite{heirung2017dual} present an MPC strategy for dual adaptive control.

\ifthenelse{\equal{\version}{arxiv}}{
\begin{table*}[tb]
    \centering
    \caption{Summary of different control methods for bilinear systems.}
    \label{tab:various_control}
    \renewcommand{\arraystretch}{1.25}
    \begin{adjustbox}{max width=\columnwidth}
    \begin{tabular}{@{}llll@{}}
        \toprule
       \textbf{Control algorithm}  & \textbf{System class} & \textbf{Objective} & \textbf{Strengths; Limitations}\\
        \midrule
        Separation-based LQG        & \eqref{eqn:LDS-BO} &  minimize quadratic cost & closed-form; can lead to zero information\\ 
        Belief-space MPC  & \eqref{eqn:LDS-BO} &  minimize quadratic cost & maximize information; computationally expensive \\
        LMI-based control         & \eqref{eqn:BDS-FO}
                              & stabilizing while maximizing the ROA & scalable; conservatism from overapproximation \\
        SOS (Polynomial Lyanunov)           & \eqref{eqn:BDS-FO}
                              & stabilizing while maximizing the ROA & expressive; poor scalability in dimension/degree \\
        Koopman-lifted control        & \eqref{eqn:BDS-FO}
                              & stabilizing through a lifted bilinear surrogate & generic; depends on learned dictionary and residuals \\
        \bottomrule
    \end{tabular}
    \end{adjustbox}
\end{table*}
}{
\begin{table*}[tb]
    \centering
    \caption{Summary of different control methods for bilinear systems.}
    \label{tab:various_control}
    \renewcommand{\arraystretch}{1.25}
    \begin{tabular}{@{}llll@{}}
        \toprule
       \textbf{Control algorithm}  & \textbf{System class} & \textbf{Objective} & \textbf{Strengths; Limitations}\\
        \midrule
        Separation-based LQG        & \eqref{eqn:LDS-BO} &  minimize quadratic cost & closed-form; can lead to zero information\\ 
        Belief-space MPC  & \eqref{eqn:LDS-BO} &  minimize quadratic cost & maximize information; computationally expensive \\
        LMI-based control         & \eqref{eqn:BDS-FO}
                              & stabilizing while maximizing the ROA & scalable; conservatism from overapproximation \\
        SOS (Polynomial Lyanunov)           & \eqref{eqn:BDS-FO}
                              & stabilizing while maximizing the ROA & expressive; poor scalability in dimension/degree \\
        Koopman-lifted control        & \eqref{eqn:BDS-FO}
                              & stabilizing through a lifted bilinear surrogate & generic; depends on learned dictionary and residuals \\
        \bottomrule
    \end{tabular}
\end{table*}
}

\section{Control of Bilinear Dynamics}\label{sec:control-bilinear-systems}

Bilinear dynamical systems of the form~\eqref{eqn:BDS-FO} represent a central system class in nonlinear control theory. 
They are the simplest systems in which state and input interact multiplicatively, yet they are expressive enough to approximate broad classes of nonlinear dynamics via Carleman linearization or Koopman-operator lifting~\citep{bilinearbook,koopman1931hamiltonian,carleman1932application}.
Controlling~\eqref{eqn:BDS-FO} is fundamentally harder than controlling a linear time-invariant (LTI) system because the effective system matrix $u_t \circ \mathcal{A}$ depends on the applied input, so standard pole-placement and LQR design do not apply directly.

Every method presented in this section faces the same challenge: the bilinear term $\sum_{k=1}^{d_u} (u_t)_k A_k x_t$ couples state and input, making the closed-loop system nonlinear even under linear feedback.
Two broad families of tractable controller-design methods have emerged in the literature, differing in \emph{how} they handle this coupling.
The first (\S\ref{sec:control-bilinear-LMI-overapprox}) treats the bilinear term as a \emph{structured uncertainty} appended to a linear nominal model and exploits robust-control tools, most prominently Petersen's lemma and ellipsoidal state bounds, to reduce the design to a semidefinite program (SDP).
The second (\S\ref{sec:control-bilinear-SOS}) embraces the \emph{polynomial} (in fact bilinear) character of the dynamics and searches directly for polynomial Lyapunov functions and rational controllers via sum-of-squares (SOS) optimization, recovering less conservative guarantees at the cost of higher computational effort.
Table~\ref{tab:comp_control} at the end of this section summarizes the trade-offs at a glance, which may be helpful before diving into the details.

Although data collected during learning may be corrupted by process noise (cf.\ Part~\ref{part:learning}), the controller design focuses on the noise-free part of the dynamics, i.e., the control objective is nominal stabilization of~\eqref{eqn:BDS-FO} with $w_t = 0$.
In particular, we assume that the only uncertainty in the system dynamics arises from the identification error, which is common in the literature on (stochastic) data-driven control; see, e.g.,~\citet{vanWaarde2022noisy,martin2023guarantees,faulwasser2023behavioral} and the references therein.
More precisely, we consider the identified bilinear dynamics
\begin{subequations}\label{eq:residual_dynamics}
    \begin{align}
        x_{t+1} 
        &= \hat{A}_0 x_t + \sum_{k=1}^{d_u}(u_t)_k \hat{A}_k x_t + \hat{B} u_t + r(x_t, u_t)
        \\
        &= \hat{A}_0 x_t + \widehat{\mathcal{N}} (I_{d_u}\otimes x_t) u_t + \hat{B} u_t + r(x_t,u_t)
    \end{align}
\end{subequations}
with $\widehat{\mathcal{N}} := \begin{bmatrix}\hat{A}_1 & \cdots & \hat{A}_{d_u}\end{bmatrix} \in \mathbb{R}^{d_x \times d_u d_x}$, where $r(x_t, u_t) := (A_0 - \hat{A}_0)x_t + \sum_{k=1}^{d_u}(u_t)_k(A_k - \hat{A}_k)x_t + (B - \hat{B})u_t$ is the \emph{identification residual}.

Here, the central challenge is that the identification error of the identified bilinear dynamics must be \emph{structurally compatible} with the subsequent controller synthesis, i.e., the \emph{shape} of this error bound matters.
In particular, the robust controller designs in \S\ref{sec:control-bilinear-LMI-overapprox} and~\S\ref{sec:control-bilinear-SOS}, which establish end-to-end guarantees for unknown bilinear systems, require the residual $r$ to be bounded by a quadratic expression that vanishes at the origin $(x_t,u_t)=(0,0)$ and grows as the state and input move away from it.

Recall the finite time error bounds from Theorem~\ref{thm:main_BDS_FO}, which readily gives high probability point-wise error bounds,\begin{subequations}\label{eq:error-bound-individual}
    \begin{align}
        \|\hat{A}_k - A_k\| &\leq \varepsilon_{A_k}, \qquad k=0,...,d_u, \\
        \|\hat{B} - B\| &\leq \varepsilon_{B},
    \end{align}
\end{subequations}
or ellipsoidal error bounds
\begin{equation}\label{eq:error-bound-ellipsoidal}
  (\hat{A}_0 - A_0) \in \mathcal{E}_A,
  \qquad
  \begin{bmatrix}\widehat{A_0 + A_i} - (A_0 + A_i)\\ \hat{b}_i - b_i\end{bmatrix} \in \mathcal{E}_{B_i}
\end{equation}
for $i=1,\ldots,d_u$. In particular, these bounds ensure that after sufficiently many observations, we may write a single residual identification error bound 
\begin{equation}\label{eq:residual-quadratic-bound}
    r(x_t,u_t)^\top r(x_t,u_t) 
    \leq 
    \begin{bmatrix} x_t \\ u_t\end{bmatrix}^\top 
    Q_\Delta 
    \begin{bmatrix} x_t \\ u_t\end{bmatrix}
\end{equation}
that holds with probability at least $1-\delta$ for a failure probability $\delta > 0$.
Here, the positive semidefinite matrix $Q_\Delta$ in constructed as in~\citet[Propositions~11 and~12]{chatzikiriakos2026endtoend}.
Thus, the learning error bounds of Part~\ref{part:learning} on the identification of the true system~\eqref{eqn:BDS-FO} are structurally compatible with control, i.e., the identification residual $r$ in~\eqref{eq:residual_dynamics} satisfies the quadratic bound which will be exploited for the robust controller synthesis.

Having established the quadratic bound~\eqref{eq:residual-quadratic-bound}, we now turn to controller synthesis.
The next two subsections each propose a different way to handle the bilinear term $\widehat{\mathcal{N}}(I_{d_u}\otimes x_t)u_t$ in~\eqref{eq:residual_dynamics}.
\S\ref{sec:control-bilinear-LMI-overapprox} overapproximates it as a structured uncertainty and solves an LMI, while \S\ref{sec:control-bilinear-SOS} keeps the polynomial structure intact and uses SOS optimization.

\subsection{LMI-Based Design via Uncertainty Overapproximation}
\label{sec:control-bilinear-LMI-overapprox}

The guiding idea of this subsection is simple.
In particular, we \emph{pretend the bilinear coupling is a linear unknown}.
More precisely, the term $\widehat{\mathcal{N}}(I_{d_u}\otimes x_t)u_t$ is rewritten as $\widehat{\mathcal{N}}\Delta(x_t)u_t$, where $\Delta(x_t)=I_{d_u}\otimes x_t$ is treated as an unknown matrix confined to a bounded set.
This converts the nonlinear stabilization problem into a robust-control problem for an uncertain linear system, which can be solved with standard semidefinite programming.
The price is additional conservatism, i.e., the LMI is feasible for \emph{all} bounded $\Delta$, not just those arising from actual state trajectories.
Figure~\ref{fig:lmi_diagram} illustrates the overall pipeline.
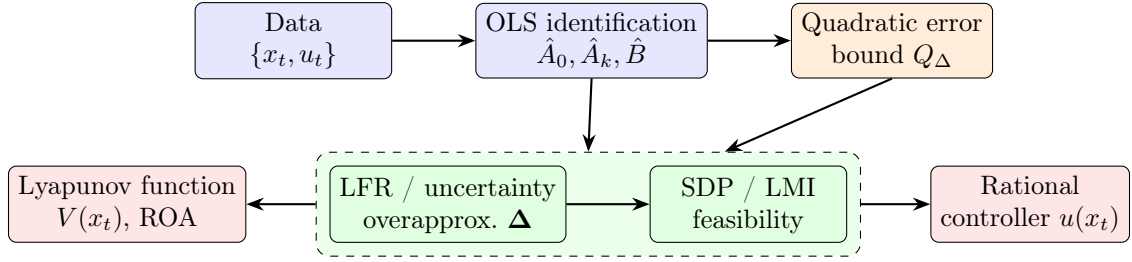
\begin{figure*}[t]
    \centering
    \begin{tikzpicture}[
      node distance=0.55cm and 1.1cm,
      box/.style={draw, rounded corners=3pt, minimum width=2.6cm, minimum height=0.85cm, align=center, font=\small},
      bigbox/.style={draw, dashed, rounded corners=5pt, inner sep=6pt},
      arr/.style={-Stealth, thick},
      label/.style={font=\scriptsize\itshape}
    ]
\node[box, fill=blue!10]  (data)   {Data\\$\{x_t,u_t\}$};
        \node[box, fill=blue!10, right=of data]  (id)     {OLS identification\\$\hat A_0,\hat A_k,\hat B$};
        \node[box, fill=orange!15, right=of id]  (bound)  {Quadratic error\\bound $Q_\Delta$};
\node[box, fill=green!12, below=1.15cm of bound,xshift=-1.9cm]  (lmi)   {SDP / LMI\\feasibility};
        \node[box, fill=green!12, left=of lmi] (lfr) {LFR / uncertainty\\overapprox.\ $\mathbf\Delta$};
        \node[box, fill=red!10,   right=of lmi]        (ctrl)   {Rational\\controller $u(x_t)$};
        \node[box, fill=red!10,   left=of lfr]         (lyap)   {Lyapunov function\\$V(x_t)$, ROA};
\begin{scope}[on background layer]
          \node[bigbox, fit=(lfr)(lmi), fill=green!8, inner sep=5pt] (synth) {};
        \end{scope}
\draw[arr] (data) -- (id);
        \draw[arr] (id)   -- (bound);
\draw[arr] (id.south)    -- ($(synth.north west)!0.5!(synth.north east)$);
        \draw[arr] (bound.south) -- ($(synth.north west)!0.75!(synth.north east)$);
\draw[arr] (lfr) -- (lmi);
\draw[arr] (synth.west) -- (lyap.east);
        \draw[arr] (synth.east) -- (ctrl.west);
    \end{tikzpicture}
    \caption{Pipeline for the LMI-based controller design. The collected data are used to identify a nominal bilinear model and derive a quadratic error bound characterized by $Q_\Delta$.
    The bilinear coupling is then overapproximated as a structured uncertainty $\mathbf{\Delta}$, reducing controller and Lyapunov certificate synthesis to a single SDP.}
    \label{fig:lmi_diagram}
\end{figure*}

To make this precise, we define $\Delta(x_t) := I_{d_u} \otimes x_t \in \mathbb{R}^{d_u d_x \times d_u}$ and write~\eqref{eq:residual_dynamics} as
\begin{equation}\label{eq:bil_compact}
  x_{t+1} = \hat{A}_0 x_t + \widehat{\mathcal{N}} \Delta(x_t) u_t + \hat{B} u_t + r(x_t,u_t).
\end{equation}
Note that~\eqref{eq:bil_compact} represents the original bilinear interaction.
In this form, $\Delta(x_t)$ plays the role of an artificial uncertainty.
Here, the key observation is that $\Delta(x_t)=I_{d_u}\otimes x_t$ is not an arbitrary matrix, but it is entirely determined by the state $x_t$.
If we know that $x_t$ lies in a pre-specified ellipsoid $\mathcal{X}$, we can characterize $\Delta(x_t)$ \emph{exactly}, i.e., without additional conservatism, via a suitable multiplier class, recasting the bilinear dynamics as a \emph{linear fractional representation} (LFR) that is linear in both the state and the uncertainty channel.
The identification residual $r$ then enters as a second, separate uncertainty channel in the same LFR, and both are handled simultaneously in one LMI, yielding end-to-end guarantees for the true unknown system.
This robust-LFR viewpoint is developed in~\citet{strasser2023robust,strasser2023control} for discrete-time bilinear systems and extended to continuous-time systems in~\citet{strasser2025koopman}.
It builds on the same fundamental objective as earlier LMI-based ROA methods~\citep{tarbouriech2009lmi,amato2009stabilization}, namely, simultaneous controller synthesis and local ROA certification.
The key difference is that, in these earlier approaches, the analysis is performed directly on the nonlinear closed-loop system induced by a fixed linear state-feedback law, whereas the framework in~\citet{strasser2023robust,strasser2023control,strasser2025koopman} recasts the bilinear dynamics as an LFR with structured uncertainty.
We note that an LFR is a common representation of uncertain systems~\citep{zhou1996robust}, which enables the use of modern robust and gain-scheduling synthesis techniques.
A controller is then designed that renders the system stable for all uncertainty realizations consistent with $\mathcal{X}$, and a Lyapunov sublevel set $\mathcal{X}_{\mathrm{RoA}} \subseteq \mathcal{X}$ is certified as a positively invariant ROA.

Crucially, the LFR framework is particularly well-suited for incorporating the identification residual $r$ from~\eqref{eq:residual_dynamics}.
To this end, the nominal bilinear controller design of~\citet{strasser2023control} has been extended in~\citet{strasser2026safedmd,chatzikiriakos2026endtoend} to account for the learning error of bilinear systems, yielding end-to-end guarantees for the underlying unknown bilinear system.
In particular, the residual $r(x_t, u_t)$ enters the LFR~\eqref{eq:lfr_bilinear} as an additional input channel, and the quadratic bound~\eqref{eq:residual-quadratic-bound} is incorporated as a second structured uncertainty block.
The LMI synthesis then simultaneously certifies stability against both the state-dependent bilinear coupling and the identification error, leading to a controller that is guaranteed to exponentially stabilize the \emph{true} (unknown) bilinear system with high probability, not just the estimated one.

\paragraph{LFR of the bilinear system}
We define the pre-specified ellipsoidal state region $\mathcal{X}$ via the quadratic inequality
\begin{equation}\label{eq:Z_set}
  \mathcal{X} 
  := \left\{ 
    x \in \mathbb{R}^{d_x} 
    \;\middle|\;
    \begin{bmatrix} x \\ 1 \end{bmatrix}^\top
    \begin{bmatrix} Q_x & S_x \\ S_x^\top & R_x \end{bmatrix}
    \begin{bmatrix} x \\ 1 \end{bmatrix} \geq 0 \right\},
\end{equation}
where $Q_x \prec 0$ and $R_x \succ 0$. 
Further, we assume the existence of
\begin{equation}
    \begin{bmatrix}
        \tilde{Q}_x & \tilde{S}_x \\ \tilde{S}_x^\top & \tilde{R}_x
    \end{bmatrix} 
    = \begin{bmatrix}
        Q_x & S_x \\ S_x^\top & R_x
    \end{bmatrix}^{-1}.
\end{equation}
The simple choice $Q_x = -I$, $S_x = 0$, $R_x = c$ recovers the ball $\|x\|^2 \leq c$.
Rewriting~\eqref{eq:bil_compact} by collecting the auxiliary signal
$\zeta_t := \Delta(x_t) u_t$, one obtains the LFR
\begin{equation}\label{eq:lfr_bilinear}
  x_{t+1} 
  = \hat{A}_0 x_t 
  + \hat{B} u_t 
  + \widehat{\mathcal{N}} \zeta_t
  + r(x_t,u_t),
  \qquad 
  \zeta_t = \Delta(x_t) u_t,
\end{equation}
which is linear in $(x_t, u_t, \zeta_t, r(x_t,u_t))$.
Here, in the feedback path appears both the residual $r(x_t,u_t)$ satisfying~\eqref{eq:residual-quadratic-bound} as well as the artificially introduced state-dependent uncertainty $\Delta(x_t) \in \mathbf{\Delta}\subseteq\mathbb{R}^{d_u d_x \times d_u}$.
The set $\mathbf{\Delta}$ characterizes the uncertainty and, due to its link to the state $x_t$ via $\Delta(x_t) = I_{d_u}\otimes x_t$, should be chosen in line with the pre-specified ellipsoidal region $\mathcal{X}$.

\paragraph{Uncertainty characterization}
The key observation of~\citet{strasser2023control} is that the structure $\Delta(x_t) = I_{d_u} \otimes x_t$ with $x_t \in \mathcal{X}$ can be characterized \emph{exactly} and without conservatism via the multiplier class
\begin{equation}\label{eq:multiplier}
  \boldsymbol{\Pi}_\Delta 
  := \left\{
    \Pi_\Delta 
    = \begin{bmatrix} 
        \Lambda \otimes Q_x & \Lambda \otimes S_x \\
        \Lambda \otimes S_x^\top & \Lambda \otimes R_x
     \end{bmatrix}
    \;\middle|\; 
    0 \preceq \Lambda \in \mathbb{R}^{d_u \times d_u}
  \right\}.
\end{equation}
More precisely, $\Delta = I_{d_u} \otimes x_t$ with $x_t \in \mathcal{X}$ if and only if $\begin{bmatrix} \Delta^\top & I \end{bmatrix} \Pi_\Delta
\begin{bmatrix} \Delta^\top & I \end{bmatrix}^\top \succeq 0$ for all
$\Pi_\Delta \in \boldsymbol{\Pi}_\Delta$~\citep[Proposition~2]{strasser2023control}.
This is an \emph{exact} (non-conservative) characterization for multi-input
systems, reducing to the standard ellipsoidal description $\mathbf{\Delta}=\mathcal{X}$ for scalar inputs ($d_u = 1$).

\paragraph{Controller design and stability guarantee}
Substituting a full-information control law~(cf.~\citet{doyle:glover:khargonekar:francis:1989,packard:zhou:pandey:leonhardson:balas:1992,astolfi:1997}) parametrized as $u_t = K x_t + K_w w_t$ (which reduces to linear feedback $u_t = Kx_t$ when $K_w = 0$) into~\eqref{eq:lfr_bilinear} and applying the multiplier characterization~\eqref{eq:multiplier} together with a quadratic Lyapunov function $V(x) = x^\top P^{-1} x$, one obtains via the dualization lemma~(cf.~\citet{scherer:weiland:2000}) a convex LMI feasibility problem in $(P, L, L_w, \Lambda, \tau)$, where $L := KP$, $L_w$ encodes $K_w$, and $\tau$ is a multiplier for the S-procedure~\citep{boyd1994lmi}.
Concretely, the LMI feasibility problem reads: find $P\in\mathbb{R}^{d_x\times d_x}$, $L \in \mathbb{R}^{d_u \times d_x}$, $L_w \in \mathbb{R}^{d_u \times d_u d_x}$, $\Lambda \in \mathbb{R}^{d_u \times d_u}$, and $\tau > 0$ such that $P\succ 0$, $\Lambda \succeq 0$, and LMI~\eqref{eq:lfr_lmi}
\begin{figure*}[t]\footnotesize
    \begin{equation}\label{eq:lfr_lmi}
        \begin{bmatrix}
            -P - \tau I_{d_x} 
            & -\widehat{\mathcal{N}} (\Lambda \otimes \tilde{S}_x) - \hat{B}L_w(I_{d_u}\otimes \tilde{S}_x) 
            & 0 
            & \hat{A}_0 P + \hat{B} L
            & \widehat{\mathcal{N}} (\Lambda \otimes I_{d_x}) + \hat{B} L_w
            \\
            \star
            & \Lambda\otimes \tilde{R}_x - L_w(I_{d_u}\otimes \tilde{S}_x) - (I_{d_u} \otimes \tilde{S}_x^\top)L_w^\top
            & -(I_{d_u}\otimes \tilde{S}_x^\top) \begin{bmatrix}0\\L_w\end{bmatrix}
            & L
            & L_w
            \\
            \star 
            & \star
            & \tau Q_\Delta^{-1}
            & \begin{bmatrix} P \\ L \end{bmatrix}
            & - \begin{bmatrix} 0 \\ L_w \end{bmatrix}
            \\
            \star 
            & \star 
            & \star 
            & P 
            & 0
            \\
            \star
            & \star
            & \star 
            & \star 
            & -\Lambda \otimes \tilde{Q}_x^{-1}
        \end{bmatrix} 
        \succ 0
    \end{equation}
    \normalsize
    \medskip
    \vspace*{-0.35\baselineskip}
    \hrule
    \vspace*{-0.35\baselineskip}
\end{figure*}holds together with a second LMI enforcing $\mathcal{X}_{\mathrm{RoA}} \subseteq \mathcal{X}$ with
\begin{equation}\label{eq:ZRoA}
  \mathcal{X}_{\mathrm{RoA}} := \{x : x^\top P^{-1} x \leq 1\};
\end{equation}
see~\citet[Theorem~4]{strasser2023control} and~\citet[Theorem~13]{chatzikiriakos2026endtoend} for the precise statement.
The decision variables are $(P, L, L_w, \Lambda)$, with total matrix size $\mathcal{O}((d_x + d_u + d_u d_x)^2)$ variables, and the controller gains are recovered as $K = LP^{-1}$ and $K_w = L_w(\Lambda^{-1} \otimes I_{d_x})$.

If feasible, the recovered controller
\begin{align}\label{eq:rational_controller}
  u(x_t) 
  &= K x_t + K_w w_t 
  = K x_t + K_w (I_{d_u} \otimes x_t) u(x_t) 
  \nonumber\\
  &= \left(I - K_w(I_{d_u} \otimes x_t)\right)^{-1} K x_t
\end{align}
is a \emph{rational} function of $x_t$, and asymptotically stabilizes the bilinear system for all $x_0 \in \mathcal{X}_{\mathrm{RoA}}$.
Here, $\mathcal{X}_{\mathrm{RoA}}\subseteq \mathcal{X}$ ensures that the uncertainty characterization~\eqref{eq:multiplier} remains valid along all closed-loop trajectories.
Positive invariance of $\mathcal{X}_{\mathrm{RoA}}$ follows from the Lyapunov decrease $\Delta V(x_t) < 0$ for all $x_t \in \mathcal{X} \supseteq \mathcal{X}_{\mathrm{RoA}}$, and the volume of $\mathcal{X}_{\mathrm{RoA}}$ is maximized by solving $\max \log\det(P)$ or $\max \mathrm{tr}(P)$ subject to the LMI constraints.

\paragraph{Quadratic performance}
The framework extends naturally to quadratic performance (e.g., $\mathcal{L}_2$-gain bounds) by augmenting the LFR~\eqref{eq:lfr_bilinear} with a performance channel.
Concretely, an exogenous disturbance input $d_t$ and a regulated output $z_t = C_z x_t + D_z u_t + \mathcal{N}_z (I_{d_u}\otimes x_t) u_t$ are introduced as an additional input-output channel in the LFR, alongside the existing uncertainty channels for the bilinear coupling and the identification residual.
The quadratic performance condition is then imposed as a further linear constraint in the LMI feasibility problem, following standard robust-control techniques~\citep{boyd1994lmi}; see~\citet[Section~IV]{strasser2023control} for the precise statement and proof.
The performance channel does not introduce new decision variables, i.e., the same $(P, L, L_w, \Lambda)$ remain, but augments the LMI with additional rows and columns corresponding to the disturbance input and regulated output dimensions, leading to a moderately larger SDP.
The result is, for instance, a certified local $\mathcal{L}_2$-gain bound $\|z\|_{\ell_2} \leq \gamma \|d\|_{\ell_2}$ valid for all trajectories starting in $\mathcal{X}_{\mathrm{RoA}}$.
 
\paragraph{Summary and limitations.}
The discussed LFR approach yields a rational state-feedback controller with end-to-end stability guarantees and a certified ellipsoidal ROA, obtained from a single SDP of size $\mathcal{O}((d_x + d_u + d_u d_x)^2)$.
The conservatism relative to the true ROA stems from three sources: 
(i)~the bilinear term $\Delta(x_t)u_t$ is treated as an uncertainty ranging over all of $\mathbf{\Delta}$, even though actual trajectories are only a subset of those realizations;
(ii)~the Lyapunov function is restricted to the quadratic class $V(x)=x^\top P^{-1}x$;
and (iii)~the ROA must be contained in the pre-specified region $\mathcal{X}$ used to construct the multipliers.
All three limitations are addressed by the SOS approach of \S\ref{sec:control-bilinear-SOS}, at the cost of a significantly larger SDP.

\subsection{Sum-of-Squares Optimization Exploiting the Bilinear Structure}
\label{sec:control-bilinear-SOS}
The LMI approaches of \S\ref{sec:control-bilinear-LMI-overapprox} achieve tractability by replacing the state-dependent factor $I_{d_u}\otimes x_t$ in the bilinear term with a worst-case bounding uncertainty $\Delta(x_t)\in\mathbf{\Delta}$.
However, this over-approximation is conservative: the true state trajectory can only follow paths consistent with the bilinear dynamics, whereas the LMI is feasible for \emph{all} trajectories of all uncertainty realizations in $\mathbf{\Delta}$, most of which never occur.
Sum-of-squares (SOS) optimization offers a different path.
Instead of over-approximating the bilinear term, it searches within the class of polynomial Lyapunov functions for a certificate that holds \emph{exactly} along bilinear trajectories.
The bilinear dynamics~\eqref{eq:residual_dynamics} is a polynomial system, so a polynomial feedback law yields a polynomial closed-loop, and checking Lyapunov decrease reduces to checking nonnegativity of a polynomial, which is where SOS comes in.
Figure~\ref{fig:sos_diagram} illustrates this pipeline.

The following subsections introduce the SOS framework, explain why bilinear systems are a particularly natural fit for it, and describe how bilinear controller synthesis can be cast as a single convex program via a rational controller parameterization.
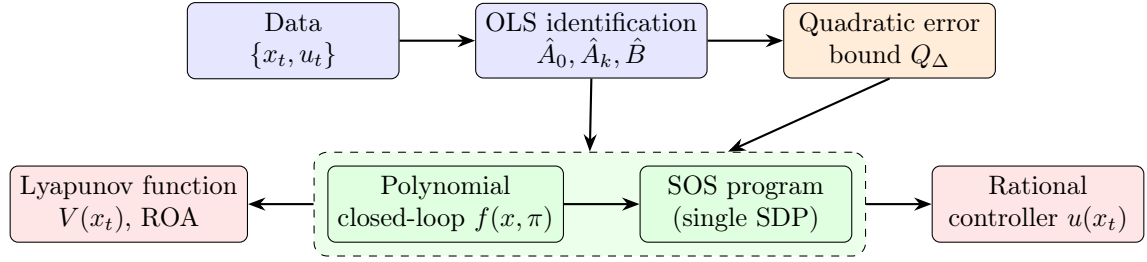
\begin{figure*}[t]
    \centering
    \begin{tikzpicture}[
      node distance=0.55cm and 1.0cm,
      box/.style={draw, rounded corners=3pt, minimum width=2.8cm, minimum height=0.85cm, align=center, font=\small},
      bigbox/.style={draw, dashed, rounded corners=5pt, inner sep=6pt},
      arr/.style={-Stealth, thick},
      label/.style={font=\scriptsize\itshape}
    ]        
\node[box, fill=blue!10]  (data)   {Data\\$\{x_t,u_t\}$};
        \node[box, fill=blue!10, right=of data]  (id)  {OLS identification\\$\hat A_0,\hat A_k,\hat B$};
        \node[box, fill=orange!15, right=of id]  (bound)  {Quadratic error\\bound $Q_\Delta$};
\node[box, fill=green!12, below=1.15cm of bound,xshift=-1.9cm] (sos) {SOS program\\(single SDP)};
        \node[box, fill=green!12, left=of sos]  (poly)   {Polynomial\\closed-loop $f(x,\pi)$};
\node[box, fill=red!10, left=1.05 of poly]   (lyap)   {Lyapunov function\\$V(x_t)$, ROA};
        \node[box, fill=red!10, right=1.05 of sos]   (ctrl)   {Rational\\controller $u(x_t)$};
\begin{scope}[on background layer]
          \node[bigbox, fit=(poly)(sos), fill=green!8, inner sep=5pt] (synth) {};
        \end{scope}
\draw[arr] (data) -- (id);
        \draw[arr] (id)   -- (bound);
\draw[arr] (id.south)    -- ($(synth.north west)!0.5!(synth.north east)$);
        \draw[arr] (bound.south) -- ($(synth.north west)!0.75!(synth.north east)$);
\draw[arr] (poly) -- (sos);
\draw[arr] (synth.west) -- (lyap.east);
        \draw[arr] (synth.east) -- (ctrl.west);
    \end{tikzpicture}
    \caption{Pipeline for the SOS-based controller design. The bilinear closed-loop is treated directly as a polynomial system. A single SOS program, formulated using the rational controller parameterization~\eqref{eq:rational_K} and the quadratic error bound characterized by $Q_\Delta$, simultaneously searches for a polynomial Lyapunov function and the controller, yielding a less conservative ROA than the LMI approach.}
    \label{fig:sos_diagram}
\end{figure*}

\subsubsection{Background: SOS Relaxation of Polynomial Nonnegativity}
A polynomial $p\in\mathbb{R}[x,2d]$ of degree $2d$ is \emph{sum-of-squares} (SOS) if there exist polynomials $q_1,\ldots,q_k$ of degree $d$ such that $p = \sum_i q_i^2$.
Every SOS polynomial is nonnegative, and the converse holds for univariate polynomials or polynomials of degree at most two, but fails in general.
SOS is therefore a sufficient, computationally tractable condition for nonnegativity.
Concretely, $p$ is SOS if and only if there exists a positive semidefinite Gram matrix $Q \succeq 0$ such that $p(x) = z(x)^\top Q\, z(x)$, where $z(x)$ is the vector of monomials up to degree $d$.
For the controller synthesis, we require a \emph{matrix} extension of the SOS notion.
A \emph{polynomial matrix} $T \in \mathbb{R}[x,d]^{m\times n}$ is a matrix whose every entry is a polynomial in $x$.
We say $T$ has degree $d$ if the maximum degree of its entries is $d$.
Further, a polynomial matrix $S\in\mathbb{R}[x,2d]^{n\times n}$ of degree $2d$ is \emph{SOS} if it can be decomposed as $S(x) = T(x)^\top T(x)$.
We write $S \in \mathrm{SOS}[x, 2d]^n$ and $S \in \mathrm{SOS}_+[x,2d]^n$ if additionally $S(x) \succ 0$ for all $x \neq 0$.
An SOS matrix satisfies $S(x) \succeq 0$ for all $x$.
Verifying $S \in \mathrm{SOS}[x,2d]^n$ reduces to a semidefinite feasibility problem~\citep{parrilo2000structured} in terms of an LMI and is solvable by standard SDP solvers (e.g., MOSEK~\citep{mosek:2026}) interfaced through toolboxes such as SOSTOOLS~\citep{prajna2002introducing,sostools2013} or YALMIP~\citep{lofberg:2004,lofberg:2009}.

\subsubsection{Bilinear Systems as Polynomial Systems}
The use of SOS for polynomial Lyapunov functions in nonlinear control was pioneered by Parrilo~\citep{parrilo2000structured} and developed into a systematic framework by Prajna, Papachristodoulou, and Parrilo~\citep{prajna2002introducing,papachristodoulou2005tutorial}.
This tool has found notable success in, e.g., systems biology, where nonlinear ODE models of gene regulatory networks are naturally polynomial or admit polynomial approximations.
For instance,~\citet{el2003model} apply SOS techniques to model validation and robust stability analysis of the bacterial heat shock response, demonstrating that parametric robustness questions in gene regulatory networks can be addressed algorithmically via semidefinite programming.
Moreover, the G-protein signaling cascade in yeast is analyzed in~\citet{yi2005application} using SOS methods. 
A further application, closer to the setting of the present paper, appear in~\citet{prajna2007framework, elsamad2006stochastic}, who construct barrier certificates via SOS methods to compute bounds on the probability that a stochastic biological process reaches an unsafe region of the state space in finite time, illustrated on the bacteriophage-$\lambda$ genetic switch.

Because the bilinear dynamics~\eqref{eq:residual_dynamics} is polynomial in $(x_t, u_t)$, any polynomial feedback $u_t = \pi(x_t)$ yields a \emph{polynomial} closed-loop vector field
\begin{equation}\label{eq:cl_vector_field}
  f(x, \pi(x)) := \hat{A}_0 x + \sum_{k=1}^{d_u}(\pi(x))_k \hat{A}_k x + \hat{B}\pi(x) + r(x,\pi(x))
\end{equation}
making the Lyapunov decrease condition $V(f(x,\pi(x))) - V(x) < 0$ a polynomial inequality in $x$ for any fixed polynomial $V$ and (polynomial) uncertainty bound on $r$.
With a linear feedback $\pi(x) = Kx$, the closed-loop map~\eqref{eq:cl_vector_field} is degree-2 in $x$, so a degree-4 Lyapunov function $V \in \mathbb{R}[x,4]$ already suffices to certify decrease~\citep{tan2006stability}.
More generally, for degree-$(2\alpha-1)$ feedback, the Lyapunov degree must be at least $4\alpha$, but the key observation is that the SDP size is determined solely by the degree of the bilinear term and the controller, not by any additional conservatism from an uncertainty overapproximation.
We emphasize that the decrease condition and the uncertainty characterization can be combined using the S-procedure~\citep{boyd1994lmi} with a polynomial multiplier~\citep[Lemma 2.1]{tan:2006}.

\paragraph{Rational Controller Parameterization and a single Synthesis SDP}
The rational controller parameterization of~\citet{strasser2025sos} is the key to turning the joint synthesis into a single convex SOS program.
The central idea is to write the control law as
\begin{equation}\label{eq:rational_K}
  u(x_t) = \frac{1}{\kappa(x_t)} L(x_t) x_t
\end{equation}
with a polynomial matrix $L \in \mathbb{R}[x,2\alpha-1]^{d_u \times d_x}$ and a strictly SOS scalar polynomial $\kappa \in \mathrm{SOS}_+[x, 2\alpha]$.
Substituting~\eqref{eq:rational_K} into the closed-loop map~\eqref{eq:cl_vector_field} and multiplying the Lyapunov decrease condition through by $\kappa(x) > 0$ yields a condition that is \emph{affine} in the decision variables $(P, L, \tau, \rho)$ for a fixed $\kappa$, where $P \succ 0$ encodes the Lyapunov function via $V(x) = x^\top P^{-1} x$, $\tau$ is a polynomial multiplier for the S-procedure, and $\rho$ is a slack variable.
This condition takes the form of a polynomial matrix in $x$ required to belong to $-\mathrm{SOS}_+[x, 2\alpha]^{3d_x+d_u}$, i.e., to be strictly negative semidefinite as a polynomial matrix~\citep{strasser2025sos}.
This leads to a single convex SOS program in $(P, L, \tau, \rho)$ solvable as one SDP.
In particular, we have an iterative solution strategy, where we fix $\kappa$ and solve for $(P, L, \tau, \rho)$, then iterate, decomposing the problem into a sequence of convex sub-problems~\citep{tan2006stability,henrion2005polynomial}.
More precisely, the SOS synthesis problem is as follows.
Find $P \succ 0$, $L \in \mathbb{R}[x,2\alpha-1]^{d_u \times d_x}$, $\tau\in\mathrm{SOS}_+[x,2\alpha]$, $\rho>0$, and a fixed $\kappa \in \mathrm{SOS}_+[x,2\alpha]$ such that
\begin{equation}
    \begin{bmatrix}
        \kappa P - \tau I_{d_x}
        & 0 
        & \kappa \hat{A}_0 P + \hat{B} L + \widehat{\mathcal{N}} (L \otimes x)
        \\
        \star 
        & \tau Q_\Delta^{-1}
        & \begin{bmatrix} \kappa P \\ L \end{bmatrix}
        \\
        \star 
        & \star
        & \kappa (P-\rho I_{d_x})
    \end{bmatrix}
\end{equation}
is in $\mathrm{SOS}[x,2\alpha]^{3d_x + d_u}$, where we refer to~\citet[Theorem~2]{strasser2025sos} and~\citet[Theorem~14]{chatzikiriakos2026endtoend} for the precise formulation and details.
The resulting controller $u(x_t)$ in~\eqref{eq:rational_K} is a rational function of the state and globally exponentially stabilizes~\eqref{eq:residual_dynamics} if the residual error bound holds globally as well.
As discussed in~\citet{chatzikiriakos2026endtoend}, this bound, however, only holds within a compact input set $\mathbb{U}$, i.e., the stability guarantees are valid for all initial conditions $x_0 \in \mathcal{X}_{\mathrm{RoA}}(c) := \{x : x^\top P^{-1} x \leq c\}$ with $c\geq 0$ chosen such that $u(x_t)\in\mathbb{U}$ for all $x\in\mathcal{X}_{\mathrm{RoA}}(c)$. 

\paragraph{Relation to the uncertainty-based LMI design}
The rational parameterization~\eqref{eq:rational_K} is structurally analogous to the gain-scheduling controller~\eqref{eq:rational_controller} of \S\ref{sec:control-bilinear-LMI-overapprox}.
More precisely, both exploit the bilinear structure to express the control law as a rational function of the state.
The key difference is that the uncertainty-based LMI design works with a rational control law with linear polynomials as numerator and denominator, whereas the numerator and denominator of the SOS-controller are polynomials of any chosen degree.
Here, the complexity of the LMI design scales better with state dimension, but is typically significantly more conservative due to the over-approximation of the bilinear term.
The SOS design, on the other hand, directly exploits the bilinearity to find larger ROA estimates at the cost of a more demanding SDP.
We emphasize that a larger polynomial degree enlarges the ROA estimate but increases the computational complexity.

\paragraph{SOS controller synthesis within model predictive control}
The approach is further extended to data-driven min-max MPC in~\citet{xie2025minmax}, where the SOS program minimizes the worst-case cost over all bilinear systems consistent with a set of noisy measurements.
In particular, an infinite-horizon min-max optimal control problem is formulated over the set of bilinear systems consistent with the collected data under bounded noise, and the worst-case cost is upper-bounded and minimized using an SOS relaxation of the resulting polynomial program following the same rational controller parameterization of~\eqref{eq:rational_K}.
The resulting Lyapunov function is obtained directly from the SOS solution, and the scheme is shown to certify robust closed-loop stability and constraint satisfaction for all noise realizations satisfying the assumed bound.

\paragraph{Summary and limitations}
The SOS approach produces less conservative stability certificates than the LMI methods because it works directly with the polynomial structure of the closed-loop rather than over-approximating the bilinear term.
The rational controller parameterization~\eqref{eq:rational_K} is the key that keeps the synthesis convex.
In particular, for a \emph{fixed} $\kappa$, all remaining decision variables enter the SOS condition affinely such that the design is solvable via an SDP.
The principal limitation is computational.
More precisely, the SDP size grows as $\mathcal{O}(d_x^{2\alpha})$ with state dimension $d_x$ and controller degree $\alpha$, making the approach practical mainly for $d_x \lesssim 10$--$15$ at $\alpha=1$.
Additionally, the stability guarantee is \emph{global} subject to the residual bound holding globally.
If the bound is only valid on a compact input set $\mathbb{U}$, the certified ROA $\mathcal{X}_{\mathrm{RoA}}(c)$ must be chosen such that the input constraint is respected for all states in the ROA.

\subsection{Numerical Experiments}
We conclude this section by comparing both design paradigms in a numerical experiment.
Here, we consider the two-dimensional bilinear dynamics investigated in~\citet{chatzikiriakos2026endtoend}.
Data are collected from two independent experiments for inputs $u_t\equiv 0$ and $u_t\equiv 1$, where identification error bounds together with the quadratic residual bound~\eqref{eq:residual-quadratic-bound} are derived following the procedure of Part~\ref{part:learning} with data length $T_0 = T_1 = T$.
A key observation from~\citet[Table~1]{chatzikiriakos2026endtoend} is the trade-off between data requirements, ROA size, and computational cost across the two design paradigms.
For the LMI-based design, feasibility depends on the size of the prescribed region $\mathcal{X}$.
In particular, a small ball $\|x\|_2^2 \leq 0.1$ requires as few as $T = 33$ samples, while certifying a larger region $\|x\|_2^2 \leq 0.9$ demands up to $T = 3999$ samples, with a computation time below $0.01\,\text{s}$ in all cases; see Fig.~\ref{fig:exmp-control-bilinear} for the corresponding certified ROAs.
\begin{figure}[tb]
    \centering
    \begin{tikzpicture}

\definecolor{iBlue}{HTML}{0072B2}
\definecolor{iOrange}{HTML}{E69F00}
\definecolor{iRed}{HTML}{D55E00}

\begin{axis}[width=0.5\linewidth,
height=0.5\linewidth,
scale only axis,
xmin=-1,
xmax=1,
xtick distance=1,
xlabel={$(x_t)_1$},
ymin=-1,
ymax=1,
ytick distance=1,
ylabel={$(x_t)_2$},
tick label style={font=\small},
xmajorgrids,
ymajorgrids,
legend style={at={(1,1)}, anchor=north east, align=left, draw=white!15!black, font=\scriptsize, inner sep=5pt}
]
\addplot [color=iRed, dashed, thick]
  table[row sep=crcr]{0.11238982	-0.29558165\\
  0.12372868	-0.29058911\\
  0.13456933	-0.28442646\\
  0.14486812	-0.27711853\\
  0.15458357	-0.26869474\\
  0.16367657	-0.25918901\\
  0.17211051	-0.24863962\\
  0.17985141	-0.23708905\\
  0.18686812	-0.2245838\\
  0.19313237	-0.21117423\\
  0.19861895	-0.19691434\\
  0.20330576	-0.18186154\\
  0.20717393	-0.16607645\\
  0.21020789	-0.14962263\\
  0.21239541	-0.13256633\\
  0.21372769	-0.11497624\\
  0.21419936	-0.096923173\\
  0.21380853	-0.078479833\\
  0.21255676	-0.059720482\\
  0.21044911	-0.040720658\\
  0.20749405	-0.021556867\\
  0.20370349	-0.00230627309999998\\
  0.19909268	0.016953607\\
  0.1936802	0.036145221\\
  0.18748783	0.055191291\\
  0.18054052	0.074015125\\
  0.17286624	0.092540926\\
  0.16449588	0.1106941\\
  0.15546316	0.12840154\\
  0.14580444	0.14559196\\
  0.13555862	0.16219613\\
  0.12476695	0.17814719\\
  0.1134729	0.19338092\\
  0.10172192	0.20783597\\
  0.089561346	0.22145414\\
  0.077040141	0.23418059\\
  0.064208722	0.24596408\\
  0.051118757	0.25675716\\
  0.037822955	0.26651637\\
  0.024374854	0.27520241\\
  0.010828603	0.28278031\\
  -0.00276125030000002	0.28921955\\
  -0.016339985	0.2944942\\
  -0.029852925	0.29858304\\
  -0.043245657	0.30146958\\
  -0.056464255	0.30314221\\
  -0.06945549	0.3035942\\
  -0.082167053	0.30282371\\
  -0.094547759	0.30083387\\
  -0.10654775	0.29763267\\
  -0.11811872	0.29323301\\
  -0.12921406	0.28765261\\
  -0.1397891	0.28091392\\
  -0.14980126	0.2730441\\
  -0.15921023	0.26407483\\
  -0.16797811	0.25404222\\
  -0.1760696	0.24298667\\
  -0.18345213	0.2309527\\
  -0.19009595	0.21798876\\
  -0.19597433	0.20414706\\
  -0.20106358	0.18948334\\
  -0.20534323	0.17405663\\
  -0.20879603	0.15792905\\
  -0.21140808	0.14116555\\
  -0.21316887	0.12383363\\
  -0.2140713	0.10600307\\
  -0.21411174	0.087745679\\
  -0.21329003	0.069134964\\
  -0.21160947	0.050245867\\
  -0.20907684	0.031154447\\
  -0.20570233	0.01193758\\
  -0.20149953	-0.0073273559\\
  -0.19648536	-0.026562787\\
  -0.19068002	-0.045691259\\
  -0.18410687	-0.064635749\\
  -0.17679239	-0.083319973\\
  -0.16876603	-0.1016687\\
  -0.16006011	-0.11960804\\
  -0.15070968	-0.13706576\\
  -0.1407524	-0.15397156\\
  -0.13022835	-0.17025738\\
  -0.11917993	-0.18585763\\
  -0.10765161	-0.2007095\\
  -0.095689809	-0.21475318\\
  -0.083342703	-0.22793212\\
  -0.070660006	-0.24019327\\
  -0.057692786	-0.25148724\\
  -0.044493257	-0.26176855\\
  -0.031114569	-0.27099582\\
  -0.017610595	-0.27913189\\
  -0.00403570819999999	-0.28614399\\
  0.00955542860000003	-0.29200389\\
  0.023108089	-0.29668799\\
  0.036567701	-0.30017743\\
  0.049880068	-0.30245817\\
  0.062991586	-0.30352101\\
  0.075849459	-0.30336168\\
  0.088401913	-0.30198083\\
  0.1005984	-0.299384\\
  0.11238982	-0.29558165\\
};
\addlegendentry{$c=0.1$}
\addplot [color=iRed, solid, thick, forget plot]
  table[row sep=crcr]{0.23525707	0.18342575\\
  0.22245122	0.19887341\\
  0.20874964	0.21352028\\
  0.19420749	0.22730738\\
  0.17888335	0.24017919\\
  0.1628389	0.25208389\\
  0.14613875	0.26297353\\
  0.12885016	0.27280427\\
  0.11104274	0.28153652\\
  0.09278818	0.28913513\\
  0.074159999	0.29556949\\
  0.055233201	0.3008137\\
  0.036083999	0.30484664\\
  0.0167895	0.30765206\\
  -0.00257260520000002	0.30921868\\
  -0.021924351	0.30954019\\
  -0.041187816	0.30861529\\
  -0.060285431	0.3064477\\
  -0.079140298	0.30304616\\
  -0.097676496	0.29842435\\
  -0.11581938	0.2926009\\
  -0.13349591	0.28559924\\
  -0.15063489	0.27744758\\
  -0.16716732	0.26817874\\
  -0.18302663	0.25783003\\
  -0.19814895	0.24644313\\
  -0.2124734	0.2340639\\
  -0.22594229	0.22074217\\
  -0.23850139	0.20653159\\
  -0.25010013	0.19148938\\
  -0.26069181	0.17567611\\
  -0.27023377	0.15915545\\
  -0.2786876	0.14199393\\
  -0.28601925	0.12426066\\
  -0.2921992	0.10602702\\
  -0.29720257	0.087366459\\
  -0.3010092	0.0683541\\
  -0.30360378	0.049066503\\
  -0.30497586	0.029581333\\
  -0.3051199	0.00997704910000002\\
  -0.30403534	-0.0096674088\\
  -0.30172653	-0.029272939\\
  -0.29820277	-0.048760598\\
  -0.29347826	-0.068051915\\
  -0.28757202	-0.087069212\\
  -0.28050782	-0.10573591\\
  -0.27231412	-0.12397685\\
  -0.2630239	-0.14171858\\
  -0.25267458	-0.15888965\\
  -0.24130783	-0.17542094\\
  -0.22896942	-0.19124586\\
  -0.21570903	-0.20630071\\
  -0.20158005	-0.22052485\\
  -0.18663938	-0.23386102\\
  -0.17094719	-0.24625552\\
  -0.15456664	-0.25765843\\
  -0.13756372	-0.26802384\\
  -0.12000687	-0.27731001\\
  -0.10196679	-0.28547955\\
  -0.083516136	-0.29249957\\
  -0.064729188	-0.2983418\\
  -0.045681599	-0.30298271\\
  -0.026450066	-0.30640361\\
  -0.00711202789999998	-0.30859074\\
  0.012254648	-0.30953528\\
  0.031571978	-0.30923342\\
  0.05076218	-0.3076864\\
  0.06974798	-0.30490043\\
  0.088452929	-0.30088674\\
  0.10680171	-0.29566148\\
  0.12472044	-0.28924569\\
  0.14213696	-0.28166522\\
  0.15898115	-0.27295058\\
  0.17518517	-0.26313686\\
  0.19068379	-0.25226358\\
  0.20541459	-0.24037453\\
  0.21931826	-0.22751758\\
  0.23233881	-0.21374449\\
  0.24442382	-0.19911073\\
  0.25552462	-0.18367521\\
  0.26559651	-0.16750011\\
  0.27459893	-0.15065054\\
  0.28249565	-0.13319435\\
  0.28925485	-0.11520184\\
  0.29484933	-0.096745449\\
  0.29925655	-0.077899499\\
  0.30245877	-0.058739875\\
  0.30444309	-0.039343726\\
  0.30520154	-0.019789154\\
  0.30473104	-0.000154898160000028\\
  0.3030335	0.019479981\\
  0.30011575	0.039036422\\
  0.29598954	0.058435677\\
  0.29067148	0.077599632\\
  0.28418299	0.09645112\\
  0.2765502	0.11491423\\
  0.26780384	0.13291463\\
  0.25797912	0.15037982\\
  0.24711562	0.16723949\\
  0.23525707	0.18342575\\
};

\addplot [color=iOrange, dashed, thick]
  table[row sep=crcr]{0.36628655	-0.68252044\\
  0.38412985	-0.67117462\\
  0.4004264	-0.65712622\\
  0.41511057	-0.64043179\\
  0.42812324	-0.62115858\\
  0.43941201	-0.59938418\\
  0.44893142	-0.57519627\\
  0.45664314	-0.54869225\\
  0.46251612	-0.51997884\\
  0.46652671	-0.48917165\\
  0.46865877	-0.45639475\\
  0.46890369	-0.4217801\\
  0.46726051	-0.3854671\\
  0.46373584	-0.34760196\\
  0.45834387	-0.30833714\\
  0.45110631	-0.26783076\\
  0.4420523	-0.22624592\\
  0.4312183	-0.18375007\\
  0.41864794	-0.14051432\\
  0.40439183	-0.096712769\\
  0.38850738	-0.052521791\\
  0.37105854	-0.00811932540000004\\
  0.35211559	0.036315834\\
  0.33175479	0.080604762\\
  0.31005813	0.12456912\\
  0.28711297	0.16803189\\
  0.26301172	0.21081805\\
  0.2378514	0.25275532\\
  0.21173335	0.29367483\\
  0.18476271	0.33341182\\
  0.15704811	0.37180628\\
  0.12870112	0.4087036\\
  0.099835907	0.44395522\\
  0.070568685	0.47741919\\
  0.041017309	0.50896076\\
  0.01130077	0.53845292\\
  -0.018461273	0.56577693\\
  -0.048148979	0.59082275\\
  -0.077642807	0.61348953\\
  -0.10682399	0.63368601\\
  -0.13557504	0.65133086\\
  -0.16378017	0.66635304\\
  -0.19132582	0.67869204\\
  -0.21810106	0.68829819\\
  -0.24399809	0.6951328\\
  -0.26891263	0.69916836\\
  -0.29274434	0.70038862\\
  -0.31539728	0.69878866\\
  -0.33678023	0.69437492\\
  -0.35680708	0.68716518\\
  -0.3753972	0.67718847\\
  -0.39247572	0.66448496\\
  -0.40797389	0.6491058\\
  -0.42182928	0.63111293\\
  -0.43398612	0.61057878\\
  -0.44439545	0.58758605\\
  -0.45301535	0.56222732\\
  -0.45981113	0.53460469\\
  -0.4647554	0.50482941\\
  -0.46782827	0.47302135\\
  -0.46901736	0.4393086\\
  -0.46831788	0.40382691\\
  -0.46573266	0.36671916\\
  -0.46127209	0.32813475\\
  -0.45495414	0.28822906\\
  -0.44680425	0.24716278\\
  -0.43685524	0.20510125\\
  -0.42514716	0.16221386\\
  -0.41172717	0.11867329\\
  -0.3966493	0.074654866\\
  -0.37997426	0.030335831\\
  -0.3617692	-0.014105356\\
  -0.34210742	-0.058489745\\
  -0.3210681	-0.10263862\\
  -0.29873595	-0.1463742\\
  -0.2752009	-0.18952038\\
  -0.2505577	-0.23190344\\
  -0.2249056	-0.2733527\\
  -0.19834789	-0.31370126\\
  -0.1709915	-0.35278666\\
  -0.14294658	-0.39045152\\
  -0.11432607	-0.42654416\\
  -0.085245213	-0.46091926\\
  -0.055821101	-0.4934384\\
  -0.026172216	-0.52397064\\
  0.00358205519999999	-0.55239303\\
  0.033321903	-0.57859113\\
  0.062927574	-0.60245945\\
  0.092279859	-0.62390188\\
  0.12126057	-0.64283208\\
  0.149753	-0.65917382\\
  0.17764243	-0.67286129\\
  0.20481656	-0.6838394\\
  0.23116596	-0.69206392\\
  0.25658454	-0.69750174\\
  0.28096994	-0.70013098\\
  0.30422398	-0.69994103\\
  0.32625301	-0.69693266\\
  0.34696834	-0.691118\\
  0.36628655	-0.68252044\\
};
\addlegendentry{$c=0.6$}
\addplot [color=iOrange, solid, thick, forget plot]
  table[row sep=crcr]{-0.60828743	0.47956897\\
  -0.62854295	0.45135786\\
  -0.64626755	0.42132929\\
  -0.66138987	0.38960417\\
  -0.673849	0.35631025\\
  -0.68359477	0.3215816\\
  -0.69058796	0.28555805\\
  -0.69480038	0.24838467\\
  -0.6962151	0.21021112\\
  -0.6948264	0.17119113\\
  -0.69063987	0.13148181\\
  -0.68367239	0.091243062\\
  -0.673952	0.05063691\\
  -0.66151784	0.00982685999999999\\
  -0.64641998	-0.031022759\\
  -0.62871922	-0.07174746\\
  -0.60848683	-0.11218326\\
  -0.58580427	-0.15216734\\
  -0.56076289	-0.19153869\\
  -0.53346352	-0.23013878\\
  -0.50401607	-0.26781219\\
  -0.47253913	-0.30440721\\
  -0.43915945	-0.33977649\\
  -0.40401142	-0.37377761\\
  -0.36723658	-0.40627366\\
  -0.32898301	-0.43713379\\
  -0.28940474	-0.46623374\\
  -0.24866113	-0.49345633\\
  -0.20691626	-0.51869194\\
  -0.16433821	-0.54183896\\
  -0.12109843	-0.56280419\\
  -0.077371028	-0.58150321\\
  -0.03333208	-0.59786072\\
  0.010841085	-0.61181085\\
  0.054970597	-0.62329744\\
  0.098878762	-0.63227422\\
  0.14238878	-0.63870507\\
  0.18532544	-0.64256407\\
  0.22751587	-0.64383569\\
  0.26879017	-0.64251482\\
  0.30898214	-0.63860676\\
  0.34792996	-0.63212727\\
  0.38547678	-0.62310241\\
  0.42147142	-0.61156855\\
  0.45576895	-0.59757211\\
  0.48823126	-0.58116946\\
  0.51872763	-0.56242665\\
  0.54713527	-0.54141914\\
  0.57333979	-0.51823153\\
  0.59723567	-0.49295719\\
  0.61872669	-0.46569788\\
  0.63772632	-0.43656336\\
  0.65415805	-0.40567097\\
  0.66795572	-0.37314507\\
  0.67906376	-0.33911666\\
  0.68743746	-0.30372274\\
  0.69304309	-0.26710584\\
  0.69585807	-0.22941339\\
  0.69587109	-0.19079718\\
  0.69308207	-0.1514127\\
  0.68750226	-0.11141853\\
  0.67915412	-0.070975719\\
  0.66807127	-0.030247113\\
  0.65429832	0.010603288\\
  0.63789075	0.051410993\\
  0.61891462	0.092011683\\
  0.59744634	0.13224188\\
  0.57357235	0.17193958\\
  0.54738879	0.21094494\\
  0.51900109	0.2491009\\
  0.48852355	0.28625382\\
  0.45607891	0.32225409\\
  0.42179779	0.35695676\\
  0.38581824	0.3902221\\
  0.34828514	0.42191614\\
  0.30934962	0.45191128\\
  0.26916845	0.48008674\\
  0.22790344	0.50632905\\
  0.18572074	0.53053255\\
  0.14279021	0.55259979\\
  0.099284715	0.5724419\\
  0.055379435	0.58997899\\
  0.011251162	0.60514044\\
  -0.032922416	0.61786521\\
  -0.076963427	0.62810205\\
  -0.12069453	0.63580975\\
  -0.16393965	0.64095726\\
  -0.20652463	0.64352387\\
  -0.24827801	0.64349923\\
  -0.28903167	0.64088345\\
  -0.3286215	0.63568706\\
  -0.36688808	0.62793098\\
  -0.40367734	0.61764644\\
  -0.43884113	0.60487486\\
  -0.47223786	0.58966766\\
  -0.50373305	0.57208608\\
  -0.53319989	0.5522009\\
  -0.56051973	0.53009222\\
  -0.58558254	0.50584903\\
  -0.60828743	0.47956897\\
};

\addplot [color=iBlue, dashed, thick]
  table[row sep=crcr]{0.23683842	0.12663874\\
  0.20799001	0.179444\\
  0.17830409	0.2315267\\
  0.14790021	0.28267713\\
  0.11690078	0.33268932\\
  0.0854306420000001	0.38136188\\
  0.0536165	0.42849884\\
  0.021586464	0.47391038\\
  -0.0105304939999999	0.51741365\\
  -0.042605048	0.55883348\\
  -0.074508048	0.59800308\\
  -0.10611103	0.63476474\\
  -0.13728674	0.66897042\\
  -0.16790964	0.7004824\\
  -0.19785644	0.72917377\\
  -0.22700653	0.75492903\\
  -0.25524255	0.77764445\\
  -0.2824508	0.79722857\\
  -0.30852171	0.81360254\\
  -0.33335032	0.82670041\\
  -0.35683665	0.83646945\\
  -0.37888612	0.84287033\\
  -0.39940995	0.84587726\\
  -0.41832549	0.84547815\\
  -0.43555659	0.84167459\\
  -0.45103386	0.8344819\\
  -0.46469497	0.82392906\\
  -0.47648491	0.81005854\\
  -0.48635623	0.7929262\\
  -0.49426915	0.77260103\\
  -0.50019183	0.74916486\\
  -0.50410041	0.72271207\\
  -0.50597916	0.69334918\\
  -0.50582051	0.66119441\\
  -0.5036251	0.62637725\\
  -0.49940177	0.58903789\\
  -0.49316752	0.54932668\\
  -0.48494746	0.50740352\\
  -0.47477469	0.46343723\\
  -0.46269017	0.41760484\\
  -0.44874256	0.37009091\\
  -0.43298802	0.32108675\\
  -0.41548999	0.27078968\\
  -0.39631893	0.21940224\\
  -0.37555204	0.16713135\\
  -0.35327293	0.11418748\\
  -0.32957131	0.060783813\\
  -0.30454262	0.00713539210000003\\
  -0.27828765	-0.04654176\\
  -0.25091212	-0.1000315\\
  -0.22252624	-0.15311846\\
  -0.19324434	-0.20558886\\
  -0.16318431	-0.25723142\\
  -0.13246719	-0.30783821\\
  -0.10121667	-0.35720544\\
  -0.0695585910000001	-0.40513433\\
  -0.037620422	-0.45143188\\
  -0.00553076959999999	-0.49591168\\
  0.026581153	-0.53839462\\
  0.058586044	-0.57870964\\
  0.0903550280000001	-0.61669439\\
  0.12176019	-0.65219593\\
  0.15267506	-0.68507131\\
  0.18297516	-0.71518815\\
  0.21253849	-0.74242518\\
  0.241246	-0.76667273\\
  0.26898209	-0.78783315\\
  0.29563509	-0.80582125\\
  0.32109768	-0.82056459\\
  0.34526731	-0.83200381\\
  0.36804668	-0.84009284\\
  0.38934405	-0.84479911\\
  0.40907367	-0.84610368\\
  0.4271561	-0.84400129\\
  0.44351852	-0.8385004\\
  0.45809504	-0.82962316\\
  0.47082698	-0.81740533\\
  0.48166307	-0.80189609\\
  0.49055966	-0.7831579\\
  0.49748095	-0.76126621\\
  0.50239906	-0.73630917\\
  0.50529418	-0.70838727\\
  0.50615466	-0.67761295\\
  0.50497704	-0.64411011\\
  0.50176605	-0.60801368\\
  0.49653463	-0.56946899\\
  0.48930384	-0.52863124\\
  0.48010279	-0.48566489\\
  0.46896854	-0.44074293\\
  0.45594592	-0.39404626\\
  0.44108736	-0.3457629\\
  0.4244527	-0.29608728\\
  0.40610892	-0.24521942\\
  0.38612989	-0.19336415\\
  0.36459604	-0.14073027\\
  0.34159409	-0.087529713\\
  0.31721667	-0.033976708\\
  0.29156193	0.019713109\\
  0.26473317	0.073323548\\
  0.23683842	0.12663874\\
};
\addlegendentry{$c=0.9$}
\addplot [color=iBlue, solid, thick, forget plot]
  table[row sep=crcr]{-0.79346626	0.52001086\\
  -0.81286315	0.4869293\\
  -0.82898693	0.45188704\\
  -0.84177267	0.4150252\\
  -0.85116889	0.37649219\\
  -0.85713776	0.33644319\\
  -0.85965523	0.29503944\\
  -0.85871118	0.25244768\\
  -0.8543094	0.2088394\\
  -0.84646761	0.1643902\\
  -0.8352174	0.11927905\\
  -0.82060407	0.073687611\\
  -0.80268645	0.027799457\\
  -0.78153669	-0.018200635\\
  -0.75723997	-0.06412744\\
  -0.7298941	-0.10979603\\
  -0.69960921	-0.1550225\\
  -0.66650724	-0.19962476\\
  -0.63072148	-0.2434232\\
  -0.59239603	-0.28624145\\
  -0.55168521	-0.32790712\\
  -0.50875295	-0.36825242\\
  -0.46377212	-0.40711489\\
  -0.41692384	-0.44433806\\
  -0.36839675	-0.47977203\\
  -0.31838627	-0.51327414\\
  -0.26709375	-0.54470947\\
  -0.21472574	-0.57395144\\
  -0.16149311	-0.60088232\\
  -0.1076102	-0.62539365\\
  -0.053293981	-0.64738675\\
  0.00123683299999999	-0.66677305\\
  0.055762666	-0.68347448\\
  0.11006396	-0.69742381\\
  0.16392207	-0.70856486\\
  0.21712012	-0.71685277\\
  0.26944391	-0.72225416\\
  0.32068274	-0.72474729\\
  0.37063029	-0.72432212\\
  0.41908545	-0.72098036\\
  0.4658531	-0.71473547\\
  0.51074492	-0.70561259\\
  0.55358015	-0.69364846\\
  0.59418631	-0.67889125\\
  0.63239989	-0.66140038\\
  0.66806702	-0.64124629\\
  0.70104408	-0.61851012\\
  0.73119828	-0.59328344\\
  0.75840821	-0.56566781\\
  0.78256429	-0.53577443\\
  0.80356927	-0.50372368\\
  0.82133855	-0.46964462\\
  0.83580059	-0.43367446\\
  0.84689716	-0.39595804\\
  0.85458357	-0.35664725\\
  0.85882888	-0.31590036\\
  0.85961599	-0.27388145\\
  0.85694172	-0.23075972\\
  0.85081686	-0.1867088\\
  0.84126605	-0.14190607\\
  0.82832776	-0.096531931\\
  0.81205409	-0.050769094\\
  0.79251057	-0.0048018286\\
  0.76977588	0.041184772\\
  0.74394158	0.087005537\\
  0.71511169	0.13247596\\
  0.68340229	0.17741295\\
  0.64894108	0.22163556\\
  0.61186681	0.26496572\\
  0.57232877	0.30722896\\
  0.53048616	0.3482551\\
  0.48650747	0.38787893\\
  0.44056978	0.42594092\\
  0.39285808	0.46228779\\
  0.34356448	0.49677319\\
  0.29288746	0.52925826\\
  0.24103109	0.55961219\\
  0.18820418	0.58771277\\
  0.13461943	0.61344683\\
  0.080492614	0.63671076\\
  0.026041685	0.65741088\\
  -0.028514105	0.67546383\\
  -0.082955079	0.69079693\\
  -0.13706202	0.70334844\\
  -0.19061706	0.71306781\\
  -0.24340456	0.71991591\\
  -0.29521195	0.72386516\\
  -0.34583063	0.72489966\\
  -0.39505676	0.72301525\\
  -0.44269215	0.71821951\\
  -0.48854497	0.71053175\\
  -0.53243059	0.69998293\\
  -0.5741723	0.68661553\\
  -0.61360202	0.67048337\\
  -0.65056099	0.65165141\\
  -0.68490037	0.63019548\\
  -0.7164819	0.60620198\\
  -0.74517841	0.57976751\\
  -0.77087435	0.55099852\\
  -0.79346626	0.52001086\\
};

\end{axis}
\end{tikzpicture}    
 \caption{Certified ROAs for the LMI-based design with $\mathcal{X} = \{x \mid \|x\|_2^2 \leq c\}$     for $c \in \{0.1, 0.6, 0.9\}$, using different types of error bounds. Here, we consider individual (\ref{eq:error-bound-individual}, dashed) and ellipsoidal (\ref{eq:error-bound-ellipsoidal}, solid) error bounds.}
    \label{fig:exmp-control-bilinear}
\end{figure}
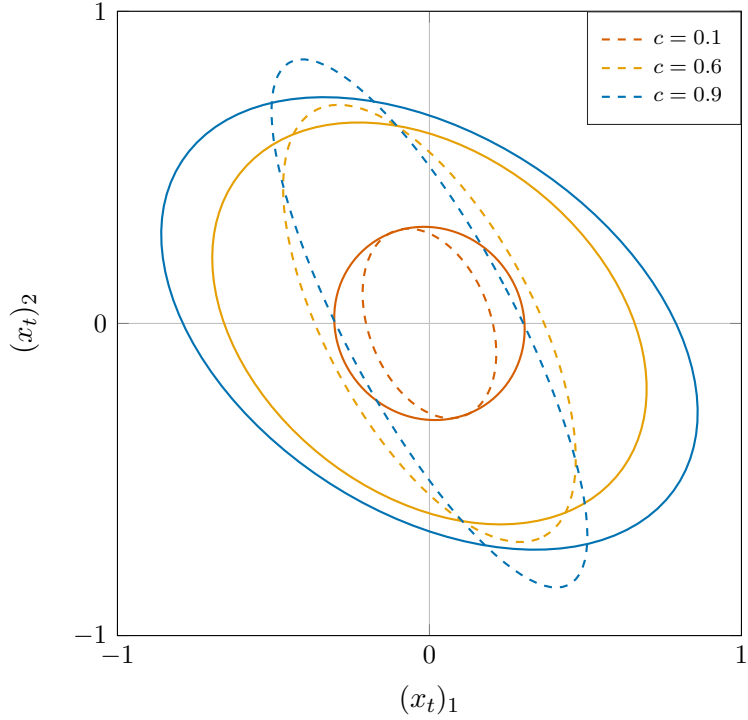
The SOS-based design with $\alpha=1$, which produces a globally stabilizing controller on all of $\mathbb{R}^2$, requires $T = 1040$ samples, at a modest increase in computation time to $0.0135\,\text{s}$, reflecting the higher complexity of the SOS program relative to the LMI.
For further numerical studies on the discussed bilinear controller designs and their feasibility and performance we refer to~\citet[Section~V-A]{strasser2023control} and~\citet[Section~3]{strasser2025sos}.

\subsection{Notes}
Various control approaches for bilinear control have been proposed in the literature~\citep{pedrycz1980stabilization,longchamp2003stable,gutman2003stabilizing,derese1980design,benallou1988optimal,lin1994kyp,khlebnikov2018discrete}.
Most of them, however, rely on inherently non-convex methods or over-approximate the bilinear term as in the LFR approach of \S\ref{sec:control-bilinear-LMI-overapprox}.
An alternative, cheaper over-approximation uses a \emph{norm-ball} uncertainty rather than the exact ellipsoidal multiplier class~\eqref{eq:multiplier}.
If the state stays inside an ellipsoid $\mathcal{E}$, $\Delta(x_t)$ is norm-bounded by a known scalar, and Petersen's lemma~\citep{petersen1987stabilization,khlebnikov2008petersens,bisoffi2020bilinear} converts the Lyapunov decrease condition into a single LMI via an auxiliary multiplier $\lambda$, reducing the design of a linear gain $u_t=Kx_t$ to a line search over $\lambda$ with a small SDP at each step.
This underlies the continuous- and discrete-time designs of~\citet{khlebnikov2016quadratic,khlebnikov2018discrete} and its extension to quadratic-bilinear systems~\citep{kafshgarkolaei:hemati:2025}, and the data-driven matrix-ellipsoidal variant of~\citet{bisoffi2022petersen}, further extended to setpoint stabilization in~\citet{bisoffi2024setpoint}.
While computationally cheaper than the LFR-based design, this approach is more conservative, as the bilinear term is absorbed into a single scalar multiplier rather than the exact, state-dependent characterization~\eqref{eq:multiplier}, and it only applies to the \emph{nominal} system, i.e., without the identification residual $r$ in~\eqref{eq:residual_dynamics}.
Incorporating $r$ as a second uncertainty channel would require an additional scalar parameter, with no guarantee of joint feasibility without further rank constraints, where, to our knowledge, this extension remains open.
In contrast, the LFR-based approach introduces a dedicated uncertainty channel for $r$ from the outset, accommodating both the bilinear coupling and the identification error simultaneously with end-to-end guarantees~\citep{chatzikiriakos2026endtoend}.

Beyond this alternative over-approximation, the two paradigms of \S\ref{sec:control-bilinear-LMI-overapprox} and~\S\ref{sec:control-bilinear-SOS} are complementary, as summarized in Table~\ref{tab:comp_control}.
\ifthenelse{\equal{\version}{arxiv}}{
\begin{table*}[tb]
    \centering
    \caption{Comparison of the two controller-design paradigms for bilinear systems.}
    \label{tab:comp_control}
    \renewcommand{\arraystretch}{1.25}
    \begin{adjustbox}{max width=\columnwidth}
    \begin{tabular}{@{}lll@{}}
        \toprule
        & \textbf{LMI (uncertainty overapprox.)} & \textbf{SOS (polynomial Lyapunov)} \\
        \midrule
        Lyapunov class        & Quadratic only
                              & Polynomial of arbitrary degree \\
        Controller structure  & Linear or rational of degree 1
                              & Rational of arbitrary degree \\
        Conservatism          & Higher (overapproximating bilinearity)
                              & Lower (exploiting bilinearity) \\
        Scalability           & $\mathcal{O}(d_x^3)$ SDP
                              & $\mathcal{O}(d_x^{2\alpha})$ SDP \\
        Implementation        & Standard SDP solvers
                              & SOSTOOLS / YALMIP + standard SDP solvers \\
        Stability guarantee   & Local (ellipsoidal $\mathcal{X}_{\mathrm{RoA}}$)
                              & Global (restricted to validity of residual bound) \\
        \bottomrule
    \end{tabular}
    \end{adjustbox}
\end{table*}
}{
\begin{table*}[tb]
    \centering
    \caption{Comparison of the two controller-design paradigms for bilinear systems.}
    \label{tab:comp_control}
    \renewcommand{\arraystretch}{1.25}
    \begin{tabular}{@{}lll@{}}
        \toprule
        & \textbf{LMI (uncertainty overapprox.)} & \textbf{SOS (polynomial Lyapunov)} \\
        \midrule
        Lyapunov class        & Quadratic only
                              & Polynomial of arbitrary degree \\
        Controller structure  & Linear or rational of degree 1
                              & Rational of arbitrary degree \\
        Conservatism          & Higher (overapproximating bilinearity)
                              & Lower (exploiting bilinearity) \\
        Scalability           & $\mathcal{O}(d_x^3)$ SDP
                              & $\mathcal{O}(d_x^{2\alpha})$ SDP \\
        Implementation        & Standard SDP solvers
                              & SOSTOOLS / YALMIP + standard SDP solvers \\
        Stability guarantee   & Local (ellipsoidal $\mathcal{X}_{\mathrm{RoA}}$)
                              & Global (restricted to validity of residual bound) \\
        \bottomrule
    \end{tabular}
\end{table*}
}
The LFR-based design is preferred when $d_x$ is large or a certified ellipsoidal ROA suffices, while the SOS-based design is preferred when a tighter, less conservative certificate is needed and $d_x \lesssim 10$--$15$ at degree $\alpha=1$ keeps the SDP tractable.
The principal bottleneck of the SOS approach is the SDP size, which grows as $\mathcal{O}(d_x^{2\alpha})$.
This may be mitigated by exploiting chordal sparsity~\citep{zheng2019chordal}, low-rank Gram matrix solvers~\citep{majumdar2020recent}, or first-order SDP solvers~\citep{ahmadi2017improving}, and, if only a regional rather than global certificate is needed, by invoking Putinar's Positivstellensatz~\citep{putinar1993positive} to restrict the Lyapunov decrease condition to a prescribed sublevel set.
Together, these directions make SOS-based controller synthesis an active research area with growing practical reach beyond small-scale systems.

Both paradigms of this section account for the quadratic error bound on the identification residual in~\eqref{eq:residual-quadratic-bound} and thus support end-to-end data-driven controller design with stability guarantees for the underlying bilinear system.

\section{Nonlinear Control based on Koopman Operator Theory and Bilinear Models}\label{sec:koopman}
The previous part of this tutorial established a complete pipeline for learning and controlling \emph{bilinear} dynamical systems of the form~\eqref{eqn:BDS-FO}.
A natural question is whether the scope of this pipeline extends beyond bilinear systems to more general nonlinear systems, typically arising in practical applications.
The Koopman operator framework provides a compelling answer to this question.
It shows that a broad class of \emph{nonlinear} (control-affine) systems can be \emph{lifted} to an approximate bilinear representation in a higher-dimensional observable space, so that the entire control machinery developed in the preceding section is generalizable to nonlinear systems.
In particular, we consider discrete-time nonlinear systems
\begin{equation}\label{eq:control-affine}
    x_{t+1} = F(x_t) + G(x_t) u_t,
\end{equation}
where $u_t \in \mathbb{R}^{d_u}$ is the control input and $G:\mathbb{X}\to\mathbb{R}^{d_x\times d_u}$ is a state-dependent input matrix.
Then, Koopman operator theory allows for the learning and control of (higher-dimensional) bilinear dynamics, representing the nonlinear behavior.
To this end, we first introduce \emph{linear} Koopman theory for discrete-time \emph{autonomous} systems (\S\ref{sec:koopman-autonomous}).
Then, we motivate the transition from linear to bilinear lifted models of \emph{controlled} systems (\S\ref{sec:koopman-controlled}), and explain how the approximation error incurred by a finite dictionary and finite data fits precisely into the quadratic residual bound~\eqref{eq:residual-quadratic-bound} (\S\ref{sec:koopman-error}).
This enables closing the loop to the LMI and SOS controller designs of \S\ref{sec:control-bilinear-LMI-overapprox} and~\S\ref{sec:control-bilinear-SOS} (\S\ref{sec:koopman-control}).
 
\subsection{Koopman Operator for Autonomous Discrete-Time Systems} \label{sec:koopman-autonomous}
Consider an autonomous discrete-time nonlinear system
\begin{equation}\label{eq:nonlinear-autonomous}
    x_{t+1} = F(x_t), \qquad x_t \in \mathbb{X} \subseteq \mathbb{R}^{d_x},
\end{equation}
with a continuous map $F:\mathbb{X}\to\mathbb{X}$.
Rather than tracking the state $x_t$ directly, the \emph{Koopman operator} $\mathcal{K}_0$ acts on \emph{observable functions} $\psi:\mathbb{X}\to\mathbb{R}$ via the pull-back
\begin{equation}\label{eq:koopman-def}
    (\mathcal{K}_0\psi)(x) := \psi(F(x)).
\end{equation}
In other words, $\mathcal{K}_0$ propagates observables one time step forward along trajectories of~\eqref{eq:nonlinear-autonomous}.
While $F$ is nonlinear, $\mathcal{K}_0$ is a \emph{linear} operator. However, the price for this linearity is that $\mathcal{K}_0$ acts on an infinite-dimensional function space, i.e., the Koopman operator is linear but infinite dimensional.
 
\paragraph{Koopman-invariant subspaces and finite representations}
If there exists a finite set of observables $\psi_1,\ldots,\psi_N:\mathbb{X}\to\mathbb{R}$ whose span is \emph{Koopman-invariant}, i.e., $\mathcal{K}_0\psi_i \in \mathrm{span}\{\psi_1,\ldots,\psi_N\}$ for all $i$, then the dynamics on that subspace can be represented exactly by a finite-dimensional \emph{linear} system.
Defining the \emph{lifted state} $\Psi(x) := \begin{bmatrix}\psi_1(x) & \cdots & \psi_N(x)\end{bmatrix}^\top \in \mathbb{R}^N$, the lifted dynamics reads
\begin{equation}\label{eq:koopman-exact}
    \Psi(x_{t+1}) = \mathcal{K}_0\Psi(x_t) = K \Psi(x_t),
\end{equation}
where $K \in \mathbb{R}^{N\times N}$ is the finite matrix representing $\mathcal{K}_0$ on the chosen subspace.
This is the key insight of Koopman theory: nonlinear dynamics becomes \emph{linear} in the lifted coordinates, provided a Koopman-invariant subspace is found.
 
\paragraph{Finite-dimensional approximation via EDMD}
In practice, a Koopman-invariant subspace is rarely available in closed form.
Instead, one chooses a \emph{dictionary} $\mathbb{V} = \mathrm{span}\{\psi_1,\ldots,\psi_N\}$ of $N$ basis observables (e.g., polynomials, radial basis functions, or neural-network features), and computes the best linear approximation of $\mathcal{K}_0$ on $\mathbb{V}$ from data $\{(x_t, x_{t+1})\}_{t=1}^T$.
To this end, a common approximation technique is \emph{extended dynamic mode decomposition} (EDMD,~\citet{williams2015edmd}).
This boils down to the OLS problem
\begin{equation}\label{eq:edmd}
    \hat{A}_0 
    = \underset{K \in \mathbb{R}^{N\times N}}{\arg\min}
    \sum_{t=1}^{T}\bigl\|\Psi(x_{t+1}) - K\Psi(x_t)\bigr\|^2.
\end{equation}
The resulting estimator $\hat{A}_0$ yields a purely data-driven linear surrogate $\Psi_{t+1} \approx \hat{A}_0 \Psi_t$, which can be used for the prediction of the underlying nonlinear system.

Unless $\mathbb{V}$ happens to be Koopman-invariant, the projection step introduces a residual $r_\Psi$.
This \emph{residual error} quantifies how well $\mathbb{V}$ approximates an invariant subspace and consists of both \emph{learning} and \emph{projection} errors.
Crucially, the residual depends on the state $x_t$ through the nonlinear map $F$, and for a sufficiently rich dictionary one can show that $\|r_\Psi(x_t)\|^2 \leq c_x\|\Psi(x_t)\|^2$ for some constant $c_x \geq 0$ that decreases as the dictionary is enriched; see, e.g.,~\citet{nuske2023finite,philipp2024error,strasser2026overview}.
 
\subsection{Koopman for Controlled Systems: Why Bilinear?}\label{sec:koopman-controlled}
We now extend the Koopman framework to control-affine nonlinear systems of the form~\eqref{eq:control-affine}.
Lifting the nonlinear dynamics using the dictionary $\Psi$ yields
\ifthenelse{\equal{\version}{arxiv}}{
\begin{align}\label{eq:lifted-control}
    \Psi(x_{t+1})
    &= \Psi(F(x_t) + G(x_t)u_t), \nn
    \\
    &= \Psi(F(x_t))
    + \sum_{k=1}^{d_u}(u_t)_k \bigl[\Psi\bigl(F(x_t) + (G(x_t))_k\bigr) - \Psi(F(x_t))\bigr] + r(x_t,u_t),
\end{align}
}{
\begin{align}\label{eq:lifted-control}
    &\Psi(x_{t+1})
    = \Psi(F(x_t) + G(x_t)u_t)
    \nonumber\\
    &= \Psi(F(x_t))
    + \sum_{k=1}^{d_u}(u_t)_k \bigl[\Psi\bigl(F(x_t) + (G(x_t))_k\bigr) - \Psi(F(x_t))\bigr]
    \nonumber\\
    &\qquad+ r(x_t,u_t),
\end{align}
}
where $(G(x_t))_k$ denotes the $k$-th column of $G(x_t)$.
Two natural choices of surrogate model arise from how one approximates the input-dependent terms.

\paragraph{Linear EDMDc and its limitations}
The simplest approximation, proposed in~\citet{brunton2016koopman,korda2018mpc}, linearizes the input dependence and approximates
\begin{equation}\label{eq:edmdc}
    \Psi(x_{t+1}) \approx \hat{A}_0 \Psi(x_t) + \hat{B}_u u_t,
\end{equation}
for some constant matrix $\hat{B}_u \in \mathbb{R}^{N\times d_u}$.
This \emph{linear EDMDc} model is computationally attractive and can be identified via OLS.
However, for control-affine systems the input appears through $G(x_t)$, a state-dependent factor, so the effective lifted input map is \emph{not} constant.
In particular, the linear model~\eqref{eq:edmdc} absorbs the state dependence into a constant $\hat{B}_u$, introducing a systematic error that grows with the strength of the nonlinear input coupling.
More fundamentally, it can be shown that linear surrogate models do not, in general, admit finite-dimensional exact representations for control-affine nonlinear systems.
Hence, the associated identification error cannot be made small by increasing the dictionary alone~\citep{bruder2021advantages,brunton2022koopman,strasser2026overview}.
 
\paragraph{Bilinear Koopman models are necessary}
The bilinear structure of the lifted model is not an artifact of the chosen
approximation scheme, but is instead grounded in the mathematical structure of the
Koopman framework itself.
In particular, the Koopman \emph{generator} preserves control-affinity exactly~\citep{surana2016koopman}, i.e., it is affine in the control input $u$.
This Koopman generator acts on $\Psi$ via the Lie derivative corresponding to the continuous-time counterpart of~\eqref{eq:control-affine}.
In particular, the drift term $\mathcal{L}_F \psi$ is input-independent, and each control channel contributes an additive term $u_k \mathcal{L}_{G_k}\psi$ that is linear in $u_k$.
Defining $z_t=\Psi(x_t)$, the resulting lifted control-affine representation is thus a bilinear dynamics.
This observation motivates the bilinear lifted model as the natural finite-dimensional surrogate also in discrete time.

In discrete time, the Koopman operator of the controlled system,
\begin{equation}
    (\mathcal{K}_u \Psi)(x)
    := \Psi\bigl(F(x) + G(x)u\bigr),
\end{equation}
is generally \emph{nonlinear} in $u$ through the composition $\Psi \circ (F(\cdot) + G(\cdot)u)$.
The bilinear structure emerges only as an \emph{approximation}~\citep{peitz2020data,philipp2025error}.
Further,~\citet{goswami2017bilinear,surana2016koopman} show that this approximation becomes \emph{exact} when the dictionary spans a Koopman-invariant subspace, i.e., the bilinear model is the minimal exact finite-dimensional representation.

Approximating the Koopman operator via an EDMD-type regression gives the
\emph{bilinear Koopman surrogate}
\begin{equation}\label{eq:bilinear-koopman}
    \Psi(x_{t+1}) 
    \approx 
    \hat{A}_0 \Psi(x_t) 
    + \hat{B} u_t
    + \sum_{k=1}^{d_u}(u_t)_k \hat{A}_k \Psi(x_t),
\end{equation}
where $\hat{A}_0$, $\hat{A}_1$, \ldots, $\hat{A}_{d_u}$ and $\hat{B}$ approximate variants of the Koopman operator, resulting in a decoupled-experiment OLS procedure for different constant control inputs; compare~\citet{chatzikiriakos2026endtoend,
strasser2026overview}.
This is precisely a bilinear dynamical system of the form~\eqref{eqn:BDS-FO} in the lifted state $z_t := \Psi(x_t) \in \mathbb{R}^N$.
The matrices $\hat{A}_0$, $\hat{A}_1$, \ldots, $\hat{A}_{d_u}$ and $\hat{B}$
play the same role as in the bilinear identification problem and can be estimated from data.
 
The bilinear structure~\eqref{eq:bilinear-koopman} is thus not merely a modeling convenience, but it is the structure inherited from the control-affine form of the dynamics, exact at the level of the Koopman generator and recovered approximately at the level of the Koopman operator.
Thus, bilinear Koopman models are not only more flexible than linear ones but are the minimal structure that captures the control-affine nonlinearity without additional conservatism~\citep{iacob2024koopman}.
\begin{remark}[Linear systems and bilinear lifting]
    For the special case $F(x) = A_0 x$ and $G(x) = B$ (i.e., an LTI system), any dictionary $\Psi$ that includes the identity ($\psi_k(x) = x_k$) as observables recovers an exact linear representation with $\hat{A}_k = 0$ for $k \geq 1$ and $\hat{B} = B$.
    In particular, LTI systems are a degenerate special case of the bilinear Koopman surrogate~\eqref{eq:bilinear-koopman} without bilinear coupling terms.
    Any richer nonlinear control-affine system requires the bilinear terms $(u_t)_k\hat{A}_k z_t$ for a finite-dimensional representation, which underscores why the bilinear setting is the right level of generality for Koopman-based control.
    We emphasize that the state-independent control matrix $\hat{B}$ is crucial for the expressiveness of the approximation, i.e., a linear system could not be represented by~\eqref{eq:bilinear-koopman} without this term.
\end{remark}
 
\subsection{Approximation Error as a Quadratic Residual Bound}\label{sec:koopman-error}
The bilinear surrogate~\eqref{eq:bilinear-koopman} is generally an approximation, not an exact representation, due to the use of a finite dictionary and finite data for learning.
Writing $z_t = \Psi(x_t)$ and collecting the approximation error, the true lifted dynamics satisfy
\begin{equation}\label{eq:koopman-bilinear-residual}
    z_{t+1}
    = \hat{A}_0 z_t + \sum_{k=1}^{d_u}(u_t)_k \hat{A}_k z_t + \hat{B} u_t
    + r_\Psi(z_t, u_t),
\end{equation}
\ifthenelse{\equal{\version}{arxiv}}{
where the \emph{Koopman residual}
\begin{equation}
    r_\Psi(z_t, u_t) 
    := \Psi(F(x_t) + G(x_t)u_t) 
    - \hat{A}_0 z_t - \sum_{k}(u_t)_k\hat{A}_k z_t - \hat{B}u_t
\end{equation}
}{where the \emph{Koopman residual} $r_\Psi := r_\Psi(z_t, u_t)$ is defined as
\begin{equation*}
    r_\Psi 
    {:=} \Psi(F(x_t) + G(x_t)u_t) 
    - \hat{A}_0 z_t - \sum_{k}(u_t)_k\hat{A}_k z_t - \hat{B}u_t
\end{equation*}
}
captures the combined effect of finite-dictionary closure error and learning error.
Equation~\eqref{eq:koopman-bilinear-residual} is of the same form as~\eqref{eq:residual_dynamics}, so the bilinear control framework applies directly once a suitable bound on $r_\Psi$ is available.
 
\paragraph{Quadratic bound on the Koopman residual}
Under standard regularity assumptions on $F$, $G$, and the dictionary $\Psi$ as well as appropriate data collection during learning, the Koopman residual satisfies a \emph{proportional} (quadratic) bound of the form
\begin{equation}\label{eq:koopman-quadratic-bound}
    r_\Psi(z_t, u_t)^\top r_\Psi(z_t, u_t)
    \leq
    \begin{bmatrix}z_t \\ u_t\end{bmatrix}^\top
    Q_\Psi
    \begin{bmatrix}z_t \\ u_t\end{bmatrix},
\end{equation}
for a positive semidefinite matrix $Q_\Psi \succeq 0$; see~\citet{strasser2025kernel,strasser2025koopman,strasser2026safedmd,strasser2026overview} for precise
statements and conditions.
More precisely, $Q_\Psi$ can be bounded from data using results from linear and bilinear system identification.
Here, the identification residual from a specific EDMD variant, i.e., an OLS regression, for the bilinear lifted model satisfies the quadratic bound with $Q_\Psi$ taking the same role as $Q_\Delta$ in~\eqref{eq:residual-quadratic-bound}.
Further, preliminary results of~\citet[Theorem~10]{chatzikiriakos2026endtoend} indicate that finite-sample guarantees for bilinear systems carry over to the lifted system.
 
Three key features make the bound~\eqref{eq:koopman-quadratic-bound} compatible with the controller designs of \S\ref{sec:control-bilinear-LMI-overapprox} and \S\ref{sec:control-bilinear-SOS}.
First, $r_\Psi(0,0) = 0$, so the bound is tight at the origin.
Second, the bound grows quadratically away from the origin, which matches the structure required by both the LMI S-procedure and the SOS optimization using robust control techniques.
Third, $Q_\Psi$ can be evaluated, and the residual error decreases in the infinite-data limit and infinite dictionary $\Psi$ (or if the dictionary is Koopman-invariant).
In the limit of an exact Koopman-invariant dictionary and infinite data, $Q_\Psi \to 0$ and the bilinear surrogate becomes exact, so the controller designed for the lifted system recovers a true stabilizer for the nonlinear system without any residual conservatism. 
The quadratic bound therefore not only enables rigorous guarantees for finite dictionaries, but also makes the framework \emph{asymptotically exact}.
\begin{remark}[Dictionary size and conservatism]
    A richer dictionary (larger $N$) reduces the closure error and, hence, tightens $Q_\Psi$, yielding less conservative stability certificates.
    However, a larger dictionary increases the state dimension $N$ of the lifted bilinear model, which affects the complexity of the SDP associated with the controller design problem.
    The size of the LMI design scales as $\mathcal{O}(N^3)$ and the SOS design as $\mathcal{O}(N^{2\alpha})$.
    Practitioners therefore face a trade-off between model fidelity (larger $N$, smaller $Q_\Psi$) and computational tractability.
    For control applications, the LMI approach of \S\ref{sec:control-bilinear-LMI-overapprox} typically scales better to large dictionaries, while SOS methods are preferable when a tight ROA certificate is needed with a moderate dictionary.
\end{remark}
 
\subsection{Data-Driven Control of Nonlinear Systems via Koopman Bilinear Surrogates}
\label{sec:koopman-control}
Combining the elements above, Figure~\ref{fig:koopman_pipeline} summarizes the end-to-end pipeline for data-driven control of a nonlinear control-affine system~\eqref{eq:control-affine} using the bilinear Koopman surrogate and the controller designs of \S\ref{sec:control-bilinear-systems}.
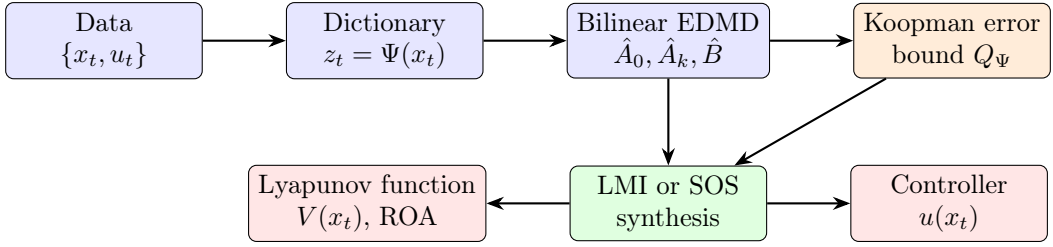
\begin{figure*}[t]
    \centering
    \begin{tikzpicture}[
      node distance=0.55cm and 1.1cm,
      box/.style={draw, rounded corners=3pt, minimum width=2.6cm, minimum height=0.85cm,
                  align=center, font=\small},
      bigbox/.style={draw, dashed, rounded corners=5pt, inner sep=6pt},
      arr/.style={-Stealth, thick}
    ]
\node[box, fill=blue!10]  (data)   {Data\\$\{x_t, u_t\}$};
        \node[box, fill=blue!10,  right=of data]  (lift)   {Dictionary\\$z_t = \Psi(x_t)$};
        \node[box, fill=blue!10,  right=of lift]  (id)     {Bilinear EDMD\\$\hat{A}_0,\hat{A}_k,\hat{B}$};
        \node[box, fill=orange!15,right=of id]    (bound)  {Koopman error\\bound $Q_\Psi$};
\node[box, fill=green!12, below=1.15cm of id]   (design) {LMI or SOS\\synthesis};
        \node[box, fill=red!10,   right=of design]      (ctrl)   {Controller\\$u(x_t)$};
        \node[box, fill=red!10,   left=of design]       (cert)   {Lyapunov function\\$V(x_t)$, ROA};
\draw[arr] (data)  -- (lift);
        \draw[arr] (lift)  -- (id);
        \draw[arr] (id)    -- (bound);
        \draw[arr] (id)    -- (design);
        \draw[arr] (bound) -- (design);
        \draw[arr] (design)-- (ctrl);
        \draw[arr] (design)-- (cert);
    \end{tikzpicture}
    \caption{End-to-end pipeline for data-driven control of nonlinear systems via Koopman bilinear surrogates. Measured data are first lifted through a dictionary to form a bilinear EDMD model. The Koopman approximation error is quantified as a quadratic bound $Q_\Psi$, and the LMI or SOS controller-design methods of \S\ref{sec:control-bilinear-LMI-overapprox} and~\S\ref{sec:control-bilinear-SOS} are applied directly to the lifted bilinear system.} 
    \label{fig:koopman_pipeline}
\end{figure*}
 
\paragraph{Identification of the bilinear lifted model}
Given $T$ trajectory samples $\{(x_t, u_t, x_{t+1})\}_{t=1}^T$ from the nonlinear system~\eqref{eq:control-affine}, one first evaluates the dictionary to obtain the lifted data $\{(z_t, u_t, z_{t+1})\}$ with $z_t := \Psi(x_t)$.
The decoupled-experiment identification strategy of~\citet{chatzikiriakos2026endtoend} then applies directly to the lifted data, yielding OLS estimates $\hat{A}_0,\hat{A}_1,\ldots,\hat{A}_{d_u},\hat{B}$ together with a quadratic error characterization for the lifted bilinear system (cf.~\citet{strasser2026safedmd}).
This provides all the ingredients required by the bilinear controller designs in \S\ref{sec:control-bilinear-systems}.
 
\paragraph{LMI-based stabilization with Koopman surrogates}
The LFR-based controller design of \S\ref{sec:control-bilinear-LMI-overapprox} applies directly to the lifted bilinear system with lifted state dimension $N$ in place of $d_x$ and lifted bilinear matrices $\hat{A}_k$ acting on $z_t = \Psi(x_t)$ rather than $x_t$.
The result is a nonlinear controller with rational structure, i.e., $u_\Psi(x_t) = u(\Psi(x_t))$.
This controller asymptotically stabilizes the bilinear surrogate~\eqref{eq:koopman-bilinear-residual} and, under the quadratic bound $Q_\Psi$, provably stabilizes the true nonlinear system~\eqref{eq:control-affine} for all $x_0$ whose lifted image lies in the certified ROA $\mathcal{X}_{\mathrm{RoA}} \subseteq \mathbb{R}^N$.
This approach is developed in~\citet{strasser2023robust,strasser2025koopman,strasser2026safedmd}, where it is referred to as \emph{SafEDMD}, and provides the first data-driven controller with
rigorous stability guarantees for discrete-time nonlinear systems via Koopman surrogates.
 
\paragraph{SOS-based stabilization with Koopman surrogates}
The LMI approach relies on overapproximating the bilinearity, which can, especially in the context of Koopman lifted bilinear systems, introduce severe conservatism (compare~\citet{strasser2025koopman}).
Instead, the SOS-based design framework of \S\ref{sec:control-bilinear-SOS} applies directly in the lifted coordinates, and the resulting stability certificates can be pulled back to the original state space via a suitable characterization of the Lyapunov function; see~\citet{strasser2025sos,strasser2025performance} for details.
Compared with the LMI approach, the SOS method yields, as for the bilinear setting, larger ROA estimates in the lifted space at the cost of a higher-dimensional SDP (growing as $\mathcal{O}(N^{2\alpha})$ in the dictionary size $N$ and controller degree $\alpha$), so it is best suited to moderate-size dictionaries. 

\paragraph{Closed-loop guarantees for the nonlinear system}
A subtlety that arises in the Koopman setting is the relationship between closed-loop guarantees of the lifted bilinear model and closed-loop guarantees of the original nonlinear system.
The key quantity is the Lyapunov function used for proving the closed-loop guarantees.
In particular, it is crucial to directly construct the Lyapunov function and its decrease inequality in the original state $x_t$.
One possibility to get such a construction is by including the identity
($\psi_k(x) = x_k$ for $k = 1,\ldots,d_x$) as observables in $\Psi$. 
Then, the original state $x_t$ is a linear projection of the lifted state $z_t=\Psi(x_t)$, and closed-loop stability of the original system~\citep{strasser2026safedmd,strasser2025koopman,chatzikiriakos2026endtoend} can be deduced by \emph{robust} stability of the lifted perturbed bilinear system.
More precisely, the certified ROA $\mathcal{X}_{\mathrm{RoA}}$ is directly expressed in the original state space as a sublevel set of the (nonlinear) Lyapunov function $V(x)=\Psi(x)^\top P \Psi(x)$.

\subsection{Numerical Examples}
We conclude this section with a numerical illustration of the Koopman-based control pipeline on a nonlinear inverted pendulum lifted to a bilinear surrogate via the dictionary $\Psi(x) = \begin{bmatrix} x_1 & x_2 & \sin(x_1)\end{bmatrix}^\top$.
Data are collected for $T = 2000$ under sub-Gaussian noise, yielding a lifted bilinear model of dimension $N = 3$ together with the quadratic residual bound~\eqref{eq:koopman-quadratic-bound}.
Applying the SOS-based controller of \S\ref{sec:control-bilinear-SOS} with $\alpha=1$ yields a stabilizing rational controller (compare~\citet[Fig.~3]{chatzikiriakos2026endtoend}).
The advantage of exploiting the bilinear structure via SOS rather than overapproximating it via the LMI is particularly pronounced in the Koopman setting, where the nonlinear lifting causes the true trajectories to occupy only a small subset of the artificial uncertainty set $\mathbf{\Delta}$, so that the LMI overapproximation introduces severe conservatism.
As illustrated in Fig.~\ref{fig:exmp-control-nonlinear}~\citep[Fig.~2]{strasser2025sos}, the SOS-based design yields a ROA with a high coverage of the sampling region.
In contrast, the LMI-based approach requires significantly more data samples to achieve feasibility and still yields a substantially smaller ROA.
\begin{figure}[tb]
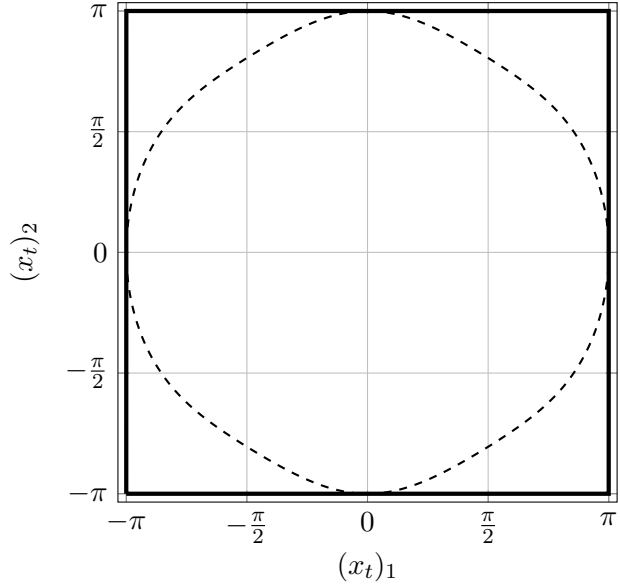

    \centering
     \caption{Sampling region (thick) and guaranteed ROAs of the SOS-based controller design (dashed).}
    \label{fig:exmp-control-nonlinear}
\end{figure}
For further numerical examples on Koopman-based control of nonlinear systems using bilinear surrogates, we refer to the survey article~\citet{strasser2026overview} and the references therein.

\subsection{Notes}
The Koopman framework makes precise the following chain of reasoning:
\begin{enumerate}
    \item 
        \textit{Any} control-affine nonlinear system~\eqref{eq:control-affine} can be lifted to a bilinear system~\eqref{eq:bilinear-koopman} in observable coordinates.
        Moreover, bilinear models are the minimal exact finite-dimensional Koopman representation, and, in particular, linear surrogates are insufficient for representing controlled systems.
    \item 
        The finite-dictionary and finite-data approximation error fits exactly the quadratic bound~\eqref{eq:koopman-quadratic-bound} required by the bilinear controller designs, with the matrix $Q_\Psi$ playing the same role as $Q_\Delta$ in the bilinear identification setting.
    \item 
        As a result, the complete pipeline of Part~\ref{part:learning} (OLS identification with finite-sample error bounds) and the bilinear controller designs of \S\ref{sec:control-bilinear-systems} apply to nonlinear systems, simply replacing the bilinear state $x_t$ by the lifted state $z_t = \Psi(x_t)$.
\end{enumerate}
This generalizes the bilinear theory discussed in this tutorial paper to a broad class of nonlinear systems, establishing bilinear system identification and control as core primitives in data-driven nonlinear control via the Koopman operator.

While Koopman-based control based on \emph{robust} bilinear controller designs are shown to be useful to derive closed-loop guarantees for data-driven control of nonlinear systems, several other Koopman-based controller design approaches exist.
The approach in~\citet{sinha2022data} uses Petersen's lemma for a \emph{nominal} bilinear controller design but, since it does not account for the residual, provides no closed-loop guarantees.
The work in~\citet{moyalan2023data,vaidya2025koopman} instead formulates data-driven optimal and stabilizing control via the \emph{Perron-Frobenius operator}, dual to the Koopman operator and acting on densities rather than observables, yielding a convex (occupation-measure, density-function, or Hamilton-Jacobi based) reformulation solved via linear or SOS programming.
However, it does not explicitly account for the surrogate's learning error and hence lacks the end-to-end guarantees of the approaches in \S\ref{sec:koopman-control}.
A complementary direction combines a Koopman-based LPV surrogate with a data-driven \emph{integral quadratic constraint} (IQC) characterization of the modeling error in the frequency domain, iteratively refining the IQC multiplier and the controller~\citep{eyuboglu2026koopman}.
This is reported to be less conservative on selected examples, but currently only yields asymptotic, rather than finite-sample, guarantees.

Open questions in the realm of Koopman-based control include (i) better tractability of the controller designs discussed in \S\ref{sec:koopman-control} as the lifted dimension $N$ grows; (ii) a systematic comparison with the set-membership approach of~\citet{xie2026koopman} and with control from partial observations~\citep{iacob2025learning,strasser2026inputoutput}; (iii) combining finite-sample residual bounds with \emph{multi-step} (rather than one-step) identification schemes~\citep{eyuboglu2025efficient}, which report improved long-horizon prediction and closed-loop performance over EDMD; and (iv) extending the IQC-based characterization of~\citet{eyuboglu2026koopman} to finite-sample guarantees.
 
\newpage
\part{Discussion \& Conclusion}\label{part:discussion}

The tools developed in this tutorial (non-asymptotic system identification, belief-space and
robust control, and Koopman-lifted bilinear surrogates) connect to a much larger landscape of modern data-driven and learning-based control and reinforcement learning. 
In Part III, we highlight several connections to these broader themes and discuss open directions. 
We emphasize that this is not an exhaustive account of the connections between bilinear control and modern machine learning, but we present it as one path forward among several. We expect the boundary between non-asymptotic control theory and ML-driven control to continue to blur.

\section{End-to-End Learning and Control and Reinforcement Learning}

In this section, we briefly explore approaches that address learning and control of a dynamical system simultaneously. 
One class of end-to-end approaches builds directly on the separate learning and control stages.
Referred to as ``model-based'' by the machine learning community or ``indirect'' by the adaptive control community, these approaches update model estimates online and then recompute control policies based on them.
By carefully designing inputs that both regulate and excite the system, such techniques can be shown to have decaying time-average sub-optimality, a metric referred to as ``regret''.
Regret guarantees have been established mainly in the setting of linear dynamics, including under state~\citep{dean2018regret} and partial observation~\citep{mania2019certainty,lale2020logarithmic}, but also more broadly for Markov jump linear systems~\citep{du2021certainty} and linear time-varying systems~\citep{minasyan2021online}.
Establishing such methods and guarantees remains an open question for bilinear dynamical systems.

Another approach whose analysis for continuous control problems (LQR, LQG, robust control, etc.) has gained attention in recent years, in part due to its empirical success in reinforcement learning (RL), is based on gradient-based techniques known as \emph{policy optimization} methods. This type of approach is distinct from putting together separate learning and control stages (i.e., first learn the dynamics and then design the control); instead, the policy (or the control law, mapping states to actions) is \emph{directly} updated using exact or approximate estimates of the cost gradient (see, e.g., \cite{aagarwal2021policygradient} and references therein). 
Here we focus on the statistical and computational analysis of these methods (which is the theme of this tutorial) applied to control problems with continuous states and actions~\citep{hu2023toward}.
There has been much recent progress on this front for a variety of control problems under linear dynamics. 
For the LQR problem, \citep{fazel2018global} showed that direct policy optimization converges to the optimal policy starting from any stabilizing policy---despite nonconvexity, the loss landscape satisfies the gradient dominance property (a special case of the Polyak-Lojasiewicz property). 
The article \citep{hu2023toward} surveys recently developed theoretical results on the optimization landscape, global convergence, and sample complexity of policy optimization for a variety of (linear) continuous control problems, such as the LQR,  Markov Jump Linear Quadratic control \citep{jansch2022PO}, risk-sensitive and robust control \citep{zhang2021risk}, and LQG \citep{zheng2022escaping};
more recent work includes policy optimization for the output estimation problem \citep{umenberger2022policysearch}.

However, 
the analysis tools for linear dynamics do not readily extend to the nonlinear case, including bilinear systems.
Yet, the practical deep RL algorithms that motivate this line of work (policy gradient, natural policy gradient, and actor-critic methods such as TRPO~\citep{schulman2015trust}, PPO~\citep{schulman2017proximal}, and SAC~\citep{haarnoja2018soft}) are routinely applied to genuinely nonlinear and often partially observed systems, typically without any comparable guarantee.
\citet{agarwal2022reinforcement} give a comprehensive statistical and computational treatment of policy optimization, exploration, and function approximation for general Markov decision processes (MDPs) underlying much of this practice, of which the LQR, LQG, robust and risk-sensitive control results mentioned above are particularly clean and analytically tractable cases studied in control.

Representation learning for \emph{low-rank} and \emph{linear} MDPs poses a similar question to the dictionary learning problem for the Koopman operator in \S\ref{sec:koopman}: when, and how efficiently, a feature map can be learned under which the relevant dynamics become linear.
\citet{agarwal2020flambe} give a statistically and computationally efficient algorithm (FLAMBE) for representation learning in low-rank MDPs with finite-sample guarantees, in a similar spirit to the finite-sample bounds for the bilinear dictionary of Part~\ref{part:learning}.
Further, \citet{uehara2022representation} extend this to oracle-efficient algorithms in both the online and offline settings, and \citet{zhang2022making} show how to make such representations practical at scale via contrastive objectives.
Finally, the Koopman operator has itself been imported directly into RL.
For instance, \citet{rozwood2024koopman} lift the Bellman equation into Koopman coordinates so that the value function evolves linearly in the lifted space, a direct structural echo of the belief-space reformulation of \S\ref{sec:koopman-control}, now applied to the value function rather than to the state.
We discuss further connections between Koopman theory and modern machine learning in \S\ref{sec:connections-modern-ml}.

\ifthenelse{\equal{\version}{arxiv}}{
\paragraph{Open Questions:}
}{
\emph{Open Questions:}}
Developing end-to-end learning and control and guarantees for bilinear dynamical systems requires addressing several open questions.
On the learning side, it is necessary to derive results for closed-loop inputs, rather than the i.i.d. noise injection considered in \S\ref{sec:sysid}.
Closed-loop inputs introduce additional dependence between covariates over time, and without sufficient noise injection, may even fail to guarantee persistence of excitation.
It may also be necessary to consider learning from unstable or marginally stable systems.
On the control side, it is necessary to develop methods that work in the presence of process noise, measurement noise, partial observation, and modeling errors.
Furthermore, it is necessary to derive guarantees on closed-loop behavior like stability and sub-optimality.
One possible starting point is to consider model-based learning of the separation principle controller for LQG-BO (\S\ref{subsec:sp}) and analyze the sub-optimality with respect to the oracle separation principle controller.
Another possible direction is to study online model-based stabilization of unknown bilinear dynamics using the tools from \S\ref{sec:control-bilinear-systems}.

The success of policy optimization methods in practice and their analytical guarantees for linear dynamical systems brings up the question of whether problems in bilinear systems can be addressed similarly.
Even with full observations, loss landscape analysis for the learning and control of a BDS is more complicated than the linear case, and strong properties such as gradient dominance or only-strict-saddles likely do not hold without additional assumptions. 
Perhaps a starting point is to consider learning a separation-principle controller for LQG-BO (\S\ref{subsec:sp}) via policy optimization, given a cost-gradient oracle.
More broadly, this topic is an interesting direction for future research.

\section{Further Connections to Machine Learning}
\label{sec:connections-modern-ml}

{In this section, we highlight three connections that are especially pertinent to the material of \S\ref{sec:sysid}--\S\ref{sec:koopman}: (1) nonparametric and kernel-based learning as an alternative to the parametric identification of Part~\ref{part:learning}; (2) representation learning in the age of foundation models, viewed through the Koopman lens of \S\ref{sec:koopman}; and (3) world models and predictive architectures, whose guiding philosophy, i.e., to predict in a learned representation rather than in the raw state, closely parallels the Koopman approach.

}
 
\subsection{Nonparametric and Kernel-Based Learning: RKHS and Gaussian Processes}
\label{sec:connections-rkhs-gp}

Throughout Part~\ref{part:learning}, we identified the bilinear and Koopman-lifted dynamics via \emph{parametric} least-squares regression: the dictionary $\Psi$ (or the bilinear factors $A_0,\ldots,A_{d_u}$) is fixed ahead of time, and only finitely many coefficients are estimated from data.
An alternative, nonparametric route is to place a prior directly on the unknown dynamics via a reproducing kernel Hilbert space (RKHS) or a Gaussian process (GP), letting the effective dictionary grow with the data rather than being fixed in advance.
This is attractive precisely where the finite-dictionary and finite-data approximation error of \S\ref{sec:koopman-error} is a concern.
More precisely, a well-specified kernel can, in principle, drive the Koopman residual bound $Q_\Psi$ in~\eqref{eq:koopman-quadratic-bound} to zero as data accumulate, without committing to a particular finite basis.
\citet{bonalli2026non} give a representative example outside the Koopman setting, establishing non-asymptotic rates for RKHS-based nonparametric identification of the drift and diffusion of a stochastic differential equation, with rates that tighten as the unknown coefficients become smoother.
Within the Koopman setting specifically, \citet{kostic2022learning} lay the statistical foundations, casting Koopman operator estimation as a regression problem restricted to an RKHS and introducing a notion of risk from which several estimators, together with guarantees on the estimated spectral decomposition, follow.
\citet{bevanda2026nonparametric} extend this approach to control-affine systems, developing an RKHS-based representation for control-Koopman operators that dispenses with an explicit finite dictionary or input parametrization altogether, casting operator estimation as an infinite-dimensional regression problem in a control-affine RKHS.
Their framework yields arbitrarily accurate finite-rank approximations without a priori restricting the hypothesis space to a fixed finite span (directly addressing the finite-dictionary approximation error that motivates the residual bound $Q_\Psi$ in \S\ref{sec:koopman-error}) and uses sketching to keep the resulting estimators scalable, extending the classical (linear) kernel-EDMD line of work~\citep{williams2015kernelEDMD} to the control setting with an explicit approximation-error analysis.
Two further papers connect this line of work directly back to the controller designs of \S\ref{sec:control-bilinear-systems}: \citet{strasser2025kernel} derive kernel-based error bounds for bilinear Koopman surrogate models of exactly the form used throughout \S\ref{sec:koopman}, giving a kernel-based route to the quadratic bound $Q_\Psi$ in~\eqref{eq:koopman-quadratic-bound}, and \citet{bold2025kernel} give a complementary error and stability analysis for kernel-based Koopman approximants for control under flexible, non-i.i.d.\ sampling schemes.

Gaussian processes offer a Bayesian counterpart to the same RKHS-based idea, replacing worst-case (frequentist) error bounds with a calibrated posterior over the unknown dynamics.
This is already a mature tool for nonlinear control.
Here, \citet{berkenkamp2017safe} combine GP regression of unmodeled dynamics with Lyapunov-based stability certificates in a model-based RL loop, and a substantial body of follow-up work builds GP-based model predictive control with probabilistic safety and stability guarantees~\citep{koller2018learning,maiworm2021online}.
Conceptually, a GP posterior variance plays a role analogous to the ellipsoidal or quadratic identification-error bounds of \S\ref{sec:control-bilinear-systems}, but is data-adaptive (shrinking where data is dense) rather than a fixed worst-case set.

\ifthenelse{\equal{\version}{arxiv}}{
\paragraph{Open Question:}
It remains open how to combine a GP or RKHS uncertainty quantification directly with the LMI or SOS controller synthesis of \S\ref{sec:control-bilinear-systems}, replacing the quadratic residual bound $Q_\Delta$ (or $Q_\Psi$) with a state-dependent, data-adaptive counterpart derived from the posterior covariance, while retaining a convex synthesis problem.
}{}

\subsection{Representation Learning, Foundation Models, and Koopman Embeddings}
\label{sec:connections-foundation-models}

Another connection concerns \emph{how} the dictionary or embedding itself is produced.
\S\ref{sec:koopman} treats the dictionary $\Psi$ as either hand-specified or learned via a fixed-basis regression (EDMD).
The broader ML community is increasingly learning such embeddings with large sequence-to-sequence architectures, e.g., in foundation models and vision-language-action (VLA) models for control, and it is natural to ask how these relate to the Koopman lifting used in this tutorial.
A recent example is \citet{li2025mamko}, who replace the constant, offline-fitted Koopman operator with one \emph{generated online} by a Mamba-style selective state-space model.
Here, the operator itself becomes a function of the recent input-output history, 
rather than a fixed matrix estimated once from batch data.
This can be read as a data-conditioned, time-varying relaxation of the bilinear surrogate~\eqref{eq:bilinear-koopman}, in which the Koopman operator's dependence on the input is no longer restricted to being bilinear but is instead parameterized by a recurrent/SSM-based network.
The trade-off is that the resulting model no longer comes with the finite-sample guarantees of \S\ref{sec:sysid} and~\S\ref{sec:koopman-error}.
Related deep-learning-based Koopman parameterizations, which similarly trade the guarantees of a fixed finite dictionary for greater expressiveness, include the deep autoencoder architecture of \citet{lusch2018deep}, which learns Koopman eigenfunctions directly from trajectory data, and its control-affine extensions~\citep{yeung2019learning,han2020deep}, which learn a bilinear Koopman representation end-to-end for use in prediction and control.

More broadly, the idea that a learned embedding linearizes (or bilinearizes) otherwise nonlinear control-affine dynamics recurs across the RL and robotics literature under different names.
For instance, \emph{embed-to-control} learns a locally linear latent dynamics model directly from raw observations~\citep{watter2015embed}, \emph{Koopman Q-learning} exploits a learned Koopman-linear latent space for offline RL~\citep{weissenbacher2022koopman}, and task-oriented Koopman encoders trained contrastively have been used directly for control~\citep{lyu2023taskoriented}.
Each of these can be viewed through the same lens as \S\ref{sec:koopman}, i.e., an embedding $\Psi$ (now learned end-to-end rather than fixed) is chosen so that the induced dynamics in the embedding space are (bi)linear, after which linear or bilinear control machinery applies. Note that this is one thread among many connecting embedding learning to Koopman theory; we have not attempted to survey the rapidly growing literature on foundation models and VLAs for control, and we expect this connection to further develop substantially.

\subsection{World Models and Predictive Architectures}
\label{sec:connections-world-models}

Closely related to the embeddings of \S\ref{sec:connections-foundation-models} and the planning in belief-space of  \S\ref{sec:Blief-space MPC} is the broader notion of a \emph{world model}.
There, a learned model predicts future (latent) states, typically used for planning or as a training signal for a policy.
The joint-embedding predictive architecture (JEPA) proposed by \citet{lecun2022path} formalizes one influential instance of this idea, training an encoder and predictor to forecast future \emph{representations} rather than future raw observations, with concrete instantiations for images and video~\citep{assran2023ijepa,bardes2025vjepa}. The belief-space MPC controller design problem we studied in \S\ref{sec:Blief-space MPC} for controlling linear dynamical systems from bilinear observations fits into the framework of the recent work on adaptive MPC control using world models~\cite{wang2026adajepa}. The appeal of predicting in representation space rather than observation space is structurally similar to the appeal of the Koopman lifting in \S\ref{sec:koopman}.
Both approaches posit that some transformation of the state admits simpler (ideally linear or low-order) predictive dynamics than the raw state itself, and both must contend with a representation-collapse or degeneracy failure mode that has no direct analogue in the finite-dimensional, spectrally well-posed Koopman setting we have considered.
Unlike the Koopman surrogates of \S\ref{sec:koopman}, JEPA-style world models do not currently come with the finite-sample error bounds that make the residual $r_\Psi$ compatible with the robust controller designs of \S\ref{sec:control-bilinear-systems}.
Closing this gap, by, e.g., deriving a Koopman-style residual bound for a learned, JEPA-style predictor, or conversely identifying conditions under which a JEPA-trained representation is provably (approximately) Koopman-invariant, is, to our knowledge, open.

\section{Conclusion}

In this paper, we set out to give a self-contained tutorial on non-asymptotic learning and control of bilinear dynamical systems. 
We reviewed identification of input-output behavior using least-squares estimation with nonlinear features, highlighting the statistical tools used for handling quantities that do not arise in the linear analysis.
On the control side, we discussed how the separation principle, a cornerstone of linear control, can fail to hold in the bilinear setting, motivating a belief-space perspective.
We also introduced Lyapunov-, semidefinite-, and SOS- approaches to state-feedback stabilization under modeling errors.
Finally, we reviewed the generality of these techniques for nonlinear control-affine systems through Koopman operator theory.

We hope to show that bilinear systems retain enough structure to support finite-sample guarantees in the spirit of classical linear theory, while already surfacing many of the difficulties: input-dependent stability, coupled state estimation and control, and nonconvex optimization landscapes that characterize nonlinear systems more broadly.
Open questions reflect the current boundary of what is understood for nonlinear systems, and we expect that closing them will require ideas drawn jointly from control theory, statistics, and machine learning. 
We hope this tutorial offers both the technical foundation and the shared vocabulary needed for researchers across these communities to take up that work, and that bilinear systems continue to serve as a productive testbed for extending non-asymptotic learning and control beyond the linear setting.
 
\section*{ACKNOWLEDGMENTS}
S.D. was partly supported by NSF CCF 2312774, NSF OAC-2311521, NSF IIS-2442137, a gift to the LinkedIn-Cornell Bowers CIS Strategic Partnership, and an AI2050 Early Career Fellowship program at Schmidt Sciences. 
M.F. was supported in part by awards NSF TRIPODS II
2023166, NSF CCF 2212261, NSF CCF 2312775, and by the Moorthy Family professorship at UW. F.A. was supported by the Deutsche Forschungsgemeinschaft (DFG, German Research Foundation) under grant AL 316/15-1 -- 468094890.

\newpage
\bibliographystyle{plainnat}
\bibliography{Bibfiles, refs}

\newpage
\appendix

\section{Proofs of Persistence of Excitation Results}\label{app:persistence excitation general}
In this Appendix, we present the proof of our main result on persistence of excitation as stated by Theorem~\ref{thm:persistence}, and then use Theorem~\ref{thm:persistence} to derive persistence of excitation guatantees for various bilinear systems.
\subsection{Proof of Theorem~\ref{thm:persistence}} \label{app:persistence excitation}

\begin{proof}
To begin, Let $X_{\phi}$ is obtained by stacking $\{\phi_t^\T\}_{t=1}^T$ row-wise. Then, from assumption~\ref{assump:phi_t distribution}, it is straightforward to see that the empirical covariance matrix $\widehat{\Sigma}_T := X_{\phi}^\T X_{\phi} = \sum_{t=1}^T \phi_t \phi_t^\top$ is isotropic in expectation as follows:

\begin{align}
	\E\left[ X_\phi^\top X_\phi\right] = \E \left[\sum_{t = 1}^{T} \phi_t \phi_t^{\top}\right] &= \sum_{t = 1}^{T} \E \left[\phi_t \phi_t^\top\right] = \sum_{t = 1}^{T} I_{d{\phi}} = T I_{d{\phi}}.
\end{align}
With a slight abuse of notation, let $v \in \Scal^{d_{\phi}-1}$, and consider the quantity $v^\top X_{\phi}^\top X_{\phi}v = \sum_{t = 1}^{T} (v^\top\phi_t)^2$, which can be viewed as a summation of the random process $\{(v^\top\phi_t)^2\}_{t = 1}^{T}$. In the following, we will derive a one-sided concentration bound for this random process.

\medskip
\noindent $\bullet$ {\bf Step 1) Blocking:} We begin by using blocking technique to get independent samples as follows,
\begin{align}
	\sum_{t = 1}^{T} (v^\top\phi_t)^2 = \sum_{k = 1}^{\tau} \sum_{\ell = 0}^{T/\tau-1} (v^\T\phi_{\ell \tau  + k})^2,
\end{align}
where we make the simplifying assumption that $T$ can be divided by $\tau$. \emph{This goes without loss of generality, and we assume it for the sake of clarity. It can be easily avoided by noting that $
\sum_{t=1}^{T} \phi_t \phi_t^\top \succeq \sum_{t=1}^{ \tau\lfloor {T}/{\tau} \rfloor } \phi_t \phi_t^\T, 
$ where $\lfloor\cdot \rfloor$ denotes the floor operator. We can then analyze everything with $T_0 =  \tau\lfloor {T}/{\tau} \rfloor$, and note that $ T-\tau <  T_0 \le T$.}

\medskip
\noindent $\bullet$ {\bf Step 2) Bernstein's inequality for non-negative random variables:} From Assumption~\ref{assump:phi_t distribution}, we have $\EE[(v^\top \phi_t)^4] \le \gamma\, \EE[(v^\top \phi_t)^2]^2$, for all $v \in \Scal^{d_\phi-1}$, and all $t \in [1,T]$. Hence, we can use one-sided Bernstein's inequality for non-negative random variables~\cite[Proposition 2.14]{wainwright2019high}, also stated in Proposition~\ref{corr:onesided_bernstein} in Section~\ref{sec:onesided} to obtain a lower bound on the smallest eigenvalue of $\sum_{t=1}^{T} \phi_t \phi_t^\T$. Specifically, we have,
\begin{align}
	&\P \left( \sum_{\ell = 0}^{T/\tau-1} \left((v^\T\phi_{ \ell \tau  + k})^2 - \E[(v^\T\phi_{\ell \tau  + k})^2]\right)  \leq - \frac{T}{\tau}\zeta \right) \leq \exp \left( \frac{-(T/\tau) \zeta^2}{\frac{2}{T/\tau}\sum_{\ell=0}^{T/\tau-1} \E[(v^\T\phi_{\ell \tau  + k})^4] }\right). 
\end{align}
From Assumption~\ref{assump:phi_t distribution}, we know that $\{\phi_{\ell \tau  + k}\}_{\ell=0}^{T/\tau-1}$ are i.i.d. for each $k \in [1,\tau]$. Therefore, we have
\begin{align}
	\P \left( \sum_{\ell = 0}^{T/\tau-1} \left((v^\T\phi_{\ell \tau + k})^2 - \E[(v^\T\phi_{\ell \tau + k})^2]\right)  \leq - \frac{T}{\tau}\zeta~ \E[(v^\T\phi_{\tau})^2]\right) &\leq \exp \left( \frac{-(T/\tau) \zeta^2~\E[(v^\T\phi_{\tau})^2]^2}{2~ \E[(v^\T\phi_{\tau })^4] }\right), \nn \\
    &\leq \exp \left( \frac{-T \zeta^2}{2 \tau \gamma }\right).
\end{align}
This further implies that
\begin{align}
	\P \left( \sum_{\ell = 0}^{T/\tau-1} (v^\T\phi_{\ell \tau + k})^2  \leq \frac{T}{\tau}(1-\zeta) \right) 
    &\leq \exp \left( \frac{-T \zeta^2}{2 \tau\gamma }\right), \label{eqn:onesided_bound_k}
\end{align}
where we use the fact that $\E[(v^\T\phi_{\tau})^2] = \E[v^\T\phi_{\tau} \phi_{\tau}^\T v] = v^\T I_{d_{\phi}} v = 1$. Note that, if $\E[\phi_{\tau} \phi_{\tau}^\T] = \Sigma_{\phi}$ instead of identity, we can do a change of variable by setting $\phi_t' := \Sigma_{\phi}^{-1/2} \phi_t$, and work with $\phi_t'$, and later convert our final results in terms of $\phi_t$. To proceed, union bounding over $\tau$ events of the form \eqref{eqn:onesided_bound_k}, we get the following,
\begin{equation}
\begin{aligned}
	\P \left( \sum_{k = 1}^{\tau}\sum_{\ell = 0}^{T/\tau-1} (v^\T\phi_{\ell \tau + k})^2  \leq T(1-\zeta) \right) 
    &\leq \tau \exp \left( \frac{-T \zeta^2}{2 \tau \gamma}\right),\\
    \implies \quad \P \left( \sum_{t = 1}^{T} (v^\top\phi_t)^2  \leq T - \sqrt{2 \tau T \gamma \log(\tau/\delta)}\right) 
    &\leq \delta \label{eqn:onesided_union_k}
\end{aligned}
\end{equation}
where we get the last inequality by setting,
\begin{align}
	\tau \exp\left(-\frac{T\zeta^2}{2\tau \gamma}\right) &= \delta, \nn \\
	\iff \frac{T\zeta^2}{2\tau \gamma} &= \log(\tau /\delta), \nn \\
	\impliedby \zeta &= \sqrt{\frac{\tau}{T}2 \gamma\log(\tau/\delta)}. 
\end{align}

\noindent $\bullet$ {\bf Step 3) Covering with ${\delta}/{(8d_{\phi})}$-net:} Next, we use a covering argument as follows: Let $\Ncal_\epsilon := \{v_1, v_2, \dots, v_{|\Ncal_\epsilon|} \} \subset \Scal^{d_{\phi}-1}$ be the $\epsilon$-net of $\Scal^{d_{\phi}-1}$ such that for any $v \in \Scal^{d_{\phi}-1}$, there exists $v_i \in \Ncal_\epsilon$ such that $\tn{v - v_i} \leq \epsilon$. From Lemma 5.2 of \cite{vershynin2010introduction}, we have $|\Ncal_\epsilon| \leq (1 + 2/ \epsilon)^{d_{\phi}}$. 

Let us choose $v \in \Scal^{d_{\phi}-1}$ for which $\lambda_{\min} (X_{\phi}^\T X_{\phi}) = v^\T X_{\phi}^\T X_{\phi} v$, and choose $v_i \in \Ncal_\epsilon$ which approximates $v$ as $\tn{v - v_i} \leq \epsilon$. By triangle inequality, we have
\begin{align}
	|v^\T X_{\phi}^\T X_{\phi} v  - v_i^\T X_{\phi}^\T X_{\phi} v_i| &= |v^\T X_{\phi}^\T X_{\phi} ( v - v_i) + (v - v_i)^\T X_{\phi}^\T X_{\phi} v_i|, \nn \\
	& \leq \norm{X_{\phi}^\T X_{\phi}}_\op \tn{v} \tn{v - v_i} + \norm{X_{\phi}^\T X_{\phi}}_\op \tn{v_i} \tn{v - v_i}, \nn \\
	& \leq 2 \epsilon \norm{X_{\phi}^\T X_{\phi}}_\op, \nn \\
    \implies  v^\T X_{\phi}^\T X_{\phi} v &\geq  v_i^\T X_{\phi}^\T X_{\phi} v_i - 2 \epsilon \norm{X_{\phi}^\T X_{\phi}}_\op, \nn \\
     \implies \lambda_{\min}(X_{\phi}^\T X_{\phi} ) &\geq \inf_{v_i \in \Ncal_\epsilon} v_i^\T X_{\phi}^\T X_{\phi} v_i - 2 \epsilon \norm{X_{\phi}^\T X_{\phi}}_\op.
\end{align}
Hence, in order to lower bound $\lambda_{\min}(X_{\phi}^\T X_{\phi} )$, we also need an upper bound on $\norm{X_{\phi}^\T X_{\phi}}_\op$. This can be done as follows: First, we have 
\begin{align}
	\E[\norm{X_{\phi}^\T X_{\phi}}_\op] = \E[\lambda_{\max}(X_{\phi}^\T X_{\phi})] &\leq \E[\tr(X_{\phi}^\T X_{\phi})]= \tr\big(\E[X_{\phi}^\T X_{\phi}]\big), \nn \\
    & = \tr\big(T I_{d_{\phi}}\big) = d_{\phi} T.
\end{align}
Hence, using Markov's inequality, we get
\begin{align}
	\P\left( \norm{X_{\phi}^\T X_{\phi}}_\op > \frac{d_{\phi}T}{\delta}\right) \leq  \frac{\E[ \norm{X_{\phi}^\T X_{\phi}}_\op ]}{d_{\phi}T} \delta \leq \delta.
\end{align}
Let  $\Ecal := \{\norm{X_{\phi}^\T X_{\phi}}_\op  \leq \frac{2d_{\phi}T}{\delta}\}$ denote the event that $\norm{X_{\phi}^\T X_{\phi}}_\op$ is upper bounded by a specified threshold. Then, it is straightforward to see that $\P(\Ecal) \geq 1 - \delta/2$. This further implies,
\begin{align}
	\P&\left(\lambda_{\min}(X_{\phi}^\T X_{\phi}) < T(1/2-\zeta)\right) \nn \\
    &\leq \P\left(\left\{\lambda_{\min}(X_{\phi}^\T X_{\phi}) < T(1/2-\zeta)\right\} \bigcap \Ecal \right) + \P (\Ecal^c), \nn \\
	&\leq \P\left(\left\{\inf_{v_i \in \Ncal_\epsilon} v_i^\T X_{\phi}^\T X_{\phi} v_i - 2 \epsilon \norm{X_{\phi}^\T X_{\phi}}_\op < T(1/2-\zeta)\right\} \bigcap \Ecal \right) + \delta/2, \nn \\
	&\leq \P\left(\inf_{v_i \in \Ncal_\epsilon} v_i^\T X_{\phi}^\T X_{\phi} v_i  < T(1/2-\zeta) +  4 \epsilon\frac{d_{\phi}T}{\delta}\right) + \delta/2, \nn \\
    &= \P\left(\inf_{v_i \in \Ncal_\epsilon} v_i^\T X_{\phi}^\T X_{\phi} v_i  < T(1/2- \zeta +  \frac{4 \epsilon d_{\phi}}{\delta})\right) + \delta/2, \nn \\
	&= \P\left(\inf_{v_i \in \Ncal_\epsilon} v_i^\T X_{\phi}^\T X_{\phi} v_i  < T(1-\zeta )\right) + \delta/2, \label{eqn:before_union_bound}
\end{align}
where we obtain the last inequality by choosing $\epsilon = \frac{\delta}{8d_{\phi}}$. Using \eqref{eqn:onesided_union_k}, with union bounding over all the elements in $\Ncal_\epsilon$, we obtain,
\begin{align}
	 \P\left(\inf_{v_i \in \Ncal_\epsilon} v_i^\T X_{\phi}^\T X_{\phi} v_i  < T(1-\zeta )\right) &\leq |\Ncal_\epsilon| \tau \exp\left(-\frac{T\zeta^2}{2\tau \gamma}\right).
\end{align}
 From Lemma 5.2 of \cite{vershynin2010introduction}, we have $|\Ncal_\epsilon| \leq (1 + 2/ \epsilon)^{d_{\phi}}$. Hence, we have
\begin{align}
    &\P\left(\inf_{v_i \in \Ncal_\epsilon} v_i^\T X_{\phi}^\T X_{\phi} v_i  < T(1-\zeta )\right) \leq \left(1 + \frac{16 d_{\phi}}{\delta}\right)^{d_{\phi}} \tau \exp\left(-\frac{T\zeta^2}{2\tau \gamma }\right), \nn \\
    \implies \quad &\P\left(\lambda_{\min}(X_{\phi}^\T X_{\phi}) < T/2- \sqrt{2\tau \gamma T\left(\log\left(\frac{2\tau }{\delta}\right) +d_{\phi} \log\left(1 + \frac{16 d_{\phi}}{\delta}\right)\right)}\right) \leq \delta,
\end{align}
where we get the last inequality by setting,
\begin{align}
	&\left(1 + \frac{16 d_{\phi}}{\delta}\right)^{d_{\phi}} \tau \exp\left(-\frac{T\zeta^2}{2\tau \gamma }\right) = \delta/2, \nn \\
	 \iff &\exp\left(-\frac{T\zeta^2}{2\tau \gamma }\right) = \delta/(2\tau) (1 + \frac{16 d_{\phi}}{\delta})^{-d_{\phi}}, \nn \\
	 \iff &\frac{T\zeta^2}{2\tau \gamma } = \log(2\tau /\delta) + d_{\phi} \log\left(1 + \frac{16 d_{\phi}}{\delta}\right), \nn \\
	\impliedby &\zeta = \sqrt{\frac{2\tau \gamma }{T}\left(\log\left(\frac{2\tau }{\delta}\right) + d_{\phi} \log\left(1 + \frac{16 d_{\phi}}{\delta}\right)\right)}.
\end{align}
Finally, choosing the trajectory length via
\begin{align}
	T/4 &\geq \sqrt{2\tau \gamma T\left(\log\left(\frac{2\tau }{\delta}\right) +d_{\phi} \log\left(1 + \frac{16 d_{\phi}}{\delta}\right)\right)},\nn \\
	 \iff T/16 &\geq 2\tau \gamma \left(\log\left(\frac{2\tau }{\delta}\right) +d_{\phi} \log\left(1 + \frac{16 d_{\phi}}{\delta}\right)\right), \nn \\
	 \impliedby T &\geq 32\tau \gamma \left(\log\left(\frac{2\tau }{\delta}\right) +d_{\phi} \log\left(1 + \frac{16 d_{\phi}}{\delta}\right)\right),
\end{align}
we obtain the following persistence of excitation result, 
\begin{align}
	\P\left(\lambda_{\min}(X_{\phi}^\T X_{\phi}) \geq T/4 \right) \geq 1 - \delta.  
\end{align}
This completes the proof of Theorem~\ref{thm:persistence}.
\end{proof}

\begin{remark}[Covering with $1/(4 \beta_{\phi,\delta})$-net:]\label{remark:covering_when_bounded_app}
    Note that Theorem~\ref{thm:persistence} holds for a broader class of $\{\phi_t\}_{t=1}^T$ satisfying $(4,2,\gamma )$-hypercontractivity condition. For specific choice of $\{\phi_t\}_{t=1}^T$, we can leverage the additional information to get a tighter sample complexity lower bound. For example, if $\{\phi_t\}_{t=1}^T$ have bounded Euclidean norm, that is, $\tn{\phi_t} \leq \beta_{\phi,\delta}$ with probability at least $1-\delta/2$, for all $t \in [1,T]$. Then, we can use a different covering argument in the proof of Theorem~\ref{thm:persistence}. Specifically, observing that, with probability at least $1-\delta/2$, we have
    \begin{align}
        \norm{X_{\phi}^\T X_{\phi}}_\op  \leq  \sum_{t = 1}^{T} \norm{ \phi_t \phi_t^\T]}_\op &\leq T \max_{t \in [1,T]} \tn{\phi_t}^2 = T \beta_{\phi,\delta}^2, 
    \end{align}
    we repeat the Step 3) in the proof of Theorem~\ref{thm:persistence} with $\Ecal := \{\norm{X_{\phi}^\T X_{\phi}}_\op  \leq T \beta_{\phi,\delta}^2\}$ and using $\epsilon = \frac{1}{4 \beta_{\phi,\delta}^2}$-net covering argument to get the following persistence of excitation result: For all $\delta \in (0,1)$, the event, 
\begin{align}
    \lambda_{\min}\left(X_{\phi}^\T X_{\phi}\right)  \geq T/4
\end{align}
holds with probability at least $1- \delta$, provided that 
\begin{align}
    T \geq 32\tau \gamma \left(\log\left(\frac{2\tau }{\delta}\right) +d_{\phi} \log\left(1 + 8 \beta_{\phi,\delta}^2\right)\right).
\end{align}
\end{remark}

\begin{remark}[Non-Isotropic Covariance:]\label{remark:persistence}
    Note that Theorem~\ref{thm:persistence} requires $\{\phi_t\}_{t \geq 0}$ to be isotropic. 
    If $\phi_t$ is non-isotropic, we still get a similar result. Suppose $\E[\phi_t \phi_t^\T] = \Sigma_{\phi}$ for some $\Sigma_{\phi} \succ 0$. 
    Setting $\phi_t' = \Sigma_{\phi}^{-1/2} \phi_t$, we have $\E[\phi_t' \phi_t'^\T] = I_{d_{\phi}}$, and from Assumption~\ref{assump:phi_t distribution}, we have
 \ifthenelse{\equal{\version}{arxiv}}{
\begin{align*}
    \EE\left[\left(v^\top \Sigma_{\phi}^{1/2}\Sigma_{\phi}^{-1/2}\phi_t\right)^4\right] &\le \gamma \EE\left[\left(v^\top\Sigma_{\phi}^{1/2} \Sigma_{\phi}^{-1/2}\phi_t\right)^2\right]^2, \implies \quad \EE\left[\left(v'^\top \phi_t'\right)^4\right] &\le \gamma \EE\left[\left(v'^\top\phi_t'\right)^2\right]^2,
\end{align*}
}{
\begin{align*}
    \EE\left[\left(v^\top \Sigma_{\phi}^{1/2}\Sigma_{\phi}^{-1/2}\phi_t\right)^4\right] &\le \gamma \EE\left[\left(v^\top\Sigma_{\phi}^{1/2} \Sigma_{\phi}^{-1/2}\phi_t\right)^2\right]^2, \\
    \implies \quad \EE\left[\left(v'^\top \phi_t'\right)^4\right] &\le \gamma \EE\left[\left(v'^\top\phi_t'\right)^2\right]^2,
\end{align*}
}
where we set $v' = \Sigma_{\phi}^{1/2} v$. 
Hence, $\{\phi_t'\}_{t=1}^T$ satisfies $(4,2,\gamma)$-hypercontractivity condition, where the hypercontractivity constant is the same for both $\phi_t$ and $\phi_t'$. 
Therefore, using Theorem \ref{thm:persistence}, the event $\lambda_{\min}\left(\sum_{t=1}^T \phi_t' \phi_t'^\T\right)  \geq T/4$ holds with probability at least $1- \delta$, provided that $T$ satisfies the trajectory condition in \eqref{eqn:trajectory_size_main}. Combining this with $\lambda_{\min}\left(\sum_{t=1}^T \phi_t \phi_t^\T\right) \geq \lambda_{\min}\left(\sum_{t=1}^T \phi_t' \phi_t'^\T\right) \lambda_{\min}\left(\Sigma_{\phi}\right)$, the event, $\lambda_{\min}\left(\sum_{t=1}^T \phi_t \phi_t^\T\right)  \geq \frac{T\lambda_{\min}\left(\Sigma_{\phi}\right)}{4}$ holds with probability at least $1- \delta$, provided that
\begin{align}
    T \gtrsim \tau\,\gamma\left(\log\left(\frac{2\tau}{\delta}\right) {+} d_{\phi} \log\left(1 {+} \frac{16 \tr\big(\Sigma_{\phi}\big) }{\delta}\right)\right).
\end{align}
\end{remark}

In the remaining of the Appendix~\ref{app:persistence excitation general}, we will derive persistence of excitation results for the bilinear systems~\eqref{eqn:LDS-BO}, \eqref{eqn:BDS-FO}, and \eqref{eqn:BDS-PO} by applying Theorem~\ref{thm:persistence} to the covariates $\phi_t$ associated with each one of them.

\newcommand{\lds}{{}}
\newcommand{\bds}{{}}
\newcommand{\pobds}{{}}
\newcommand{\bdstr}{{}}
\subsection{Persistence of Excitation for Learning~\eqref{eqn:LDS-BO}}
In this subsection, we will derive the persistence of excitation guarantees for learning linear dynamical systems with bilinear observations, governed by state-space equation \eqref{eqn:LDS-BO}. In this case, the covariates are constructed as follows:
\begin{equation}
\begin{aligned} \label{eqn:phi_t_LDS_BO_app}
   \phi_t^\lds  = \util_t \otimes \ubar_{t-1}, \quad \text{where,} \quad
  \util_t := \begin{bmatrix} 1 \\ u_t \end{bmatrix}, \quad \text{and} \quad \ubar_{t-1} := \begin{bmatrix} u_{t-1} \\ u_{t-2}  \\ \vdots \\ u_{t-\tau+1} \end{bmatrix}. 
\end{aligned}
\end{equation}

\begin{proposition}[Persistence of Excitation for Learning~\eqref{eqn:LDS-BO}]\label{prop:persistence_LDS_BO_app}
Consider the covariates $\lbrace \phi_t^\lds \rbrace_{t =\tau+1}^T$ defined in \eqref{eqn:phi_t_LDS_BO_app}. Suppose Assumption~\ref{assump:input+noise}\,\emph{(\textbf{a})} holds. Then, for all $\delta \in (0,1)$, the event:
\begin{align}
    \lambda_{\min}\left( \sum_{t=\tau}^{T} \phi_t^\lds \phi_t^{\lds\T} \right) \geq (T-\tau+1)/4,
\end{align}
holds with probability at least $1- \delta$, provided that 
\begin{align}
   T - \tau +1 \gtrsim  \tau \bigg(\log\left(\frac{2\tau}{\delta}\right) &+ d_u(d_u+1)(\tau-1) \log\big({d_u(d_u+1)(\tau-1)}\big) \bigg). \label{eqn:trajectory_size_LDS_BO_app}
\end{align}
\end{proposition}
Proposition~\ref{prop:persistence_LDS_BO_app} follows from the application of Theorem~\ref{thm:persistence} to the covariates $\{\phi_t^\lds\}_{t \geq \tau}$. Therefore, before we begin the proof of Proposition~\ref{prop:persistence_LDS_BO_app}, we show that the covariates $\{\phi_t^\lds\}_{t \geq \tau}$, given by \eqref{eqn:phi_t_LDS_BO_app}, satisfy the conditions required by Theorem~\ref{thm:persistence}. 

\subsubsection{Hypercontractivity of the Covariates in \eqref{eqn:phi_t_LDS_BO_app}} 

From \eqref{eqn:phi_t_LDS_BO_app} and Assumption~\ref{assump:input+noise}(\textbf{a}), it is clear that $\{\phi_t^\lds\}_{t \geq \tau}$ has a block-dependent structure with block-length $\tau$. Moreover, Assumption~\ref{assump:input+noise}(\textbf{a}) also implies that $ \E \left[ \phi_t^\lds\right] = 0$, and
\begin{align}
    \E \left[ \phi_t^\lds \phi_t^{\lds\T}\right] &= \E \left[ \left(\util_t \util_t^T \right) \otimes \left(\ubar_{t-1} \ubar_{t-1}^\T\right)\right] = I_{d_u+1} \otimes I_{d_u(\tau-1)} = I_{d_u(d_u+1)(\tau-1)},
\end{align}
Next, we show that $\{\phi_t^\lds\}_{t \geq \tau}$ is $(4,2,9)$-hypercontractive.
\begin{lemma}[Hypercontractivity of $\phi_t^\lds$ in \eqref{eqn:phi_t_LDS_BO_app}]\label{lemma:hyper_phi_lds_app}
    Let $\phi_t^\lds$ be as defined in \eqref{eqn:phi_t_LDS_BO_app}. Under Assumption~\ref{assump:input+noise}\,\emph{(\textbf{a})}, we have $\EE[(v^\top \phi_t^\lds)^4] \le 9\, \EE[(v^\top \phi_t^\lds)^2]^2$, for all $v \in \R^{d_u(d_u+1)(\tau-1)}$, and all $t \in [\tau+1,T]$. 
\end{lemma}
\begin{proof}
    Let $v \in \R^{d_u(d_u+1)(\tau - 1)}$ be an arbitrary vector, and $V := \matop(v) \in \R^{ d_u(\tau-1) \times (d_u+1)}$ is the corresponding matrix. Then, the second moment of the random variable $v^\top \phi_t^\lds$ is given by,
    \begin{align}
        \E \left[(v^\top \phi_t^\lds)^2\right] = \E \left[\left((\util_t^\T \otimes \ubar_{t-1}^\T) v\right)^2\right]  &= \E \left[\left(\ubar_{t-1}^\T V \util_t\right)^2\right], \nn \\ 
        &= \E \left[\E \left[\left(\ubar_{t-1}^\T V \util_t\right)^2 \bgl \ubar_{t-1}\right] \right], \nn \\
        &= \E \left[ \ubar_{t-1}^\T V V^\T \ubar_{t-1}\right], \nn \\
        &= \tr \left( V V^\T \E[\ubar_{t-1}\ubar_{t-1}^\T]\right), \nn \\
        &= \tr \left( V V^\T\right), \nn \\
        &= \tf{V}^2. 
        \label{eqn:lds_hyper_second_mom_app}
    \end{align}
    Next, we will upper bound the fourth moment of the random variable $v^\top \phi_t^\lds$ as follows:
    \begin{align}
         \E \left[(v^\top \phi_t^\lds)^4\right] = \E \left[\left((\util_t^\T \otimes \ubar_{t-1}^\T) v\right)^4\right]  &= \E \left[\left(\ubar_{t-1}^\T V \util_t\right)^4\right], \nn \\ 
         &= \E \left[\left(\util_t^\T V^\T \ubar_{t-1}\right)^4\right], \nn \\
         &= \E \left[\E \left[\left(\util_t^\T V^\T \ubar_{t-1}\right)^4 \bgl \util_t\right] \right], \nn \\
         &\leqsym{i} 3 \E \left[ \tn{V \util_t}^4\right], \nn \\
         &= 3 \E \left[ (\util_t^\T V^\T V \util_t)^2\right], \nn \\
         & \leqsym{ii} 9 \tf{V}^4, 
        \label{eqn:lds_hyper_fourth_mom_app}
    \end{align}
where we get (i) from the observation that, for any fixed $q = [q_1^\T~~q_2^\T~~\cdots~~q_{\tau-1}^\T]^\T \in \R^{d_u(\tau - 1)}$, using Assumption~\ref{assump:input+noise}(\textbf{a}), we have
\begin{align}
    \E \left[ (q^\T \ubar_{t-1})^4\right] &= \E \left[ \left(\sum_{k=1}^{\tau-1}q_k^\T u_{t-k}\right)^4\right], \nn \\
    &= \sum_{k=1}^{\tau-1} \E \left[ (q_k^\T u_{t-k})^4\right] + 3\sum_{k=1}^{\tau-1}\sum_{\substack{\ell = 1 \\ \ell \neq k}}^{\tau-1}  \E \left[ (q_k^\T u_{t-k})^2\right]  \E \left[ (q_\ell^\T u_{t-\ell})^2\right], \nn \\
    &\leq \sum_{k=1}^{\tau-1} 3\tn{q_k}^4 + 3\sum_{k=1}^{\tau-1}\sum_{\substack{\ell = 1 \\ \ell \neq k}}^{\tau-1} \tn{q_k}^2 \tn{q_\ell}^2, \nn \\
    &= 3 \left(\sum_{k=1}^{\tau-1} \tn{q_k}^2\right)^2, \nn \\
    &= 3 \tn{q}^4.
\end{align}
Moreover, we get (ii) in \eqref{eqn:lds_hyper_fourth_mom_app} from the application of Lemma~\ref{lemma:hyper_to_moment_app} along-with the observation that $\util_t$ is $(4,2,3)$-hypercontractive as follows: for any fixed $b = [b_0~~b_1] \in \R^{d_u+1}$, using Assumption~\ref{assump:input+noise}(\textbf{a}), we have
\begin{align}
    \E \left[ (b^\T \util_t)^4\right] &= \E \left[ \left(b_0 + b_1^\T u_t\right)^4\right], \nn \\
    &= b_0 ^4 + 6 b_0^2 \E \left[(b_1^\T u_t)^2\right] + \E \left[(b_1^\T u
    _t)^4\right], \nn \\
    &\leq b_0^4 + 6 b_0^2\tn{b_1}^2 + 3 \tn{b_1}^4, \nn \\
    &\leq 3 \left(b_0^2 + \tn{b_1}^2\right)^2, \nn \\
    &= 3 \tn{b}^4, \nn \\
    &= 3 \E \left[ (b^\T \util_t)^2\right]^2.
\end{align}
Finally, combining \eqref{eqn:lds_hyper_second_mom_app} and \eqref{eqn:lds_hyper_fourth_mom_app}, we get the statement of Lemma~\ref{lemma:hyper_phi_lds_app} as follows:
\begin{align}
    \E \left[(v^\top \phi_t^\lds)^4\right] \leq 9 \tf{V}^4 = 9\, \E \left[(v^\top \phi_t^\lds)^2\right]^2.
\end{align}
This completes the proof.
\end{proof}

\subsubsection{Finalizing the Proof of Proposition~\ref{prop:persistence_LDS_BO_app}}
We are now ready to apply Theorem~\ref{thm:persistence} to the covariates $\lbrace \phi_t^\lds \rbrace_{t =\tau}^T$ in \eqref{eqn:phi_t_LDS_BO_app}. Specifically, applying Theorem~\ref{thm:persistence} with $\gamma = 9$, $d_\phi = d_u(d_u+1)(\tau-1)$, and using a $1/(4 d_{\phi})$-net covering argument (see Remark~\ref{remark:covering_when_bounded_app}), we get the statement of Proposition~\ref{prop:persistence_LDS_BO_app}.

\subsection{Persistence of Excitation for Learning~\eqref{eqn:BDS-FO}}
In this subsection, we will derive the persistence of excitation guarantees for learning fully observed bilinear dynamical systems, governed by state-space equation \eqref{eqn:BDS-FO}. The covariates for learning \eqref{eqn:BDS-FO} are constructed as follows:
\begin{equation}
\begin{aligned} \label{eqn:phi_t_BDS_app}
   \phi_t^\bds =  \begin{bmatrix} u_{t} \\
    \util_{t} \otimes x_{t}\end{bmatrix} = \begin{bmatrix} u_{t} \\ x_t \\
    u_t \otimes x_{t}\end{bmatrix}. 
\end{aligned}
\end{equation}
Before we state our main result on lower bounding the minimum eigenvalue of $\sum_{t=1}^T \phi_t^\bds \phi_t^{\bds\T}$, we need to define a few quantities that are used in the proofs of our main results related to the dynamical system~\eqref{eqn:BDS-FO}. Recall that, we can unroll the state dynamics in \eqref{eqn:BDS-FO} to write for all $t \ge 0$
\begin{equation}
\begin{aligned}\label{eqn:bilinear sys state app}
        x_{t} & = \sum_{\ell = 0}^{t-1} \left(\prod_{k = 0}^{\ell-1} (u_{t-k-1} \circ A)\right) \left(B u_{t-\ell-1}  +  w_{t-\ell-1}\right), \\ 
        & = \bar{G} \,\begin{bmatrix}
    u_{t-1} \\ 
    \util_{t-1} \otimes u_{t-2} \\ 
    \util_{t-1} \otimes \util_{t-2} \otimes u_{t-3} \\ \vdots \\ 
    \util_{t-1} \otimes \util_{t-2} \otimes \cdots \otimes u_{t-\tau+1} 
    \end{bmatrix} + \bar{F} \,\begin{bmatrix}
    w_{t-1} \\ 
    \util_{t-1} \otimes w_{t-2} \\ 
    \util_{t-1} \otimes \util_{t-2} \otimes w_{t-3} \\ \vdots \\ 
    \util_{t-1} \otimes \util_{t-2} \otimes \cdots \otimes w_{t-\tau+1} 
    \end{bmatrix} \\
    & \quad + \left( \prod_{\ell = 1}^{\tau-1} (u_{t-\ell}\circ A )  \right) x_{t-\tau+1}, 
\end{aligned}
\end{equation}
where the matrices $\bar{G}$ and $\bar{F}$ are given by
\begin{equation}
\begin{aligned}
\bar{G} &:= \begin{bmatrix} \bar{G}_1 & \bar{G}_2 & \cdots & \bar{G}_{\tau-1} \end{bmatrix}, \quad \quad 
\bar{F} := \begin{bmatrix} \bar{F}_1 & \bar{F}_2 & \cdots & \bar{F}_{\tau-1} \end{bmatrix},
\end{aligned}
\end{equation}
where we define $\bar{G}_1 := B $, $\bar{G}_{\ell} := \lbrace A_{i_1} A_{i_2} \cdots A_{i_{\ell-1}} B  \rbrace_{i_1,i_2 \dots, i_{\ell -1}  \in \lbrace 0, \dots, d_u\rbrace} $ for $\ell \in \lbrace 2, \dots, t \rbrace$. 
$\{\bar{G}_\ell\}_{\ell=1}^{t}$ contains $\lbrace B, \lbrace A_{i_1} B \rbrace_{i_1 \in \lbrace 0, \dots, d_u \rbrace},$ $ \dots, \lbrace A_{i_1} A_{i_2}  \cdots A_{i_{t-1}} B \rbrace_{i_1, i_2 \dots, i_{t - 1}  \in \lbrace 0, \dots, d_u\rbrace} \rbrace $. The matrices $\{\bar{F}_\ell\}_{\ell=1}^{t}$ are defined similarly, with $B$ replaced by $I_{d_x}$. Hence, instead of unrolling $x_t$ till $x_{t-\tau+1}$, if we unroll it till $x_0 = 0$, then we can compute the covariance of $x_t$ as follows:
\begin{align}
    \Gamma_x^{(t)} := \E[x_t x_t^\T]  = \sum_{\ell = 1}^{t} \bar{G}_\ell \bar{G}_\ell^\top + \sigma^2 \sum_{\ell = 1}^{t} \bar{F}_\ell \bar{F}_\ell^\top. \label{eqn:BDS_covariance_app}
\end{align}
where we get \eqref{eqn:BDS_covariance_app} from the fact that the nonlinear vectors getting multiplied by $\bar{G}$ and $\bar{F}$ in the expansion of $x_t$ are independent, zero-mean, and isotropic (as shown in the Proof of Lemma 12 in \cite{sattar2025finite_arXiv}). Moreover, under Assumptions~\ref{assump:input+noise} and \ref{assump:stability}, for all $t \in [1,T]$, we have
\begin{align}
    \tn{x_t} \leq \frac{\kappa\sqrt{\sigma^2 + \opn{B}^2}}{1-\rho} \sqrt{d_x + \log(T/\delta)} =: \beta_{x,\delta} \label{eqn:BDS_state_norm_app}
\end{align}
with probability at least $1-\delta$ (see Lemma~\ref{lemma:truncation_bias_phi_bds_app}). With these definitions, we are now ready to state out main result on the persistence of excitations for learning \eqref{eqn:BDS-FO}.

\begin{proposition}[Persistence of Excitation for Learning~\eqref{eqn:BDS-FO}]\label{prop:persistence_BDS_FO_app}
Given a single input-state trajectory $\{(u_t,x_t)\}_{t=0}^T$ of the system~\eqref{eqn:BDS-FO}, construct $\{\phi_t^\bds\}_{t=1}^T$ according to~\eqref{eqn:phi_t_BDS_app}. Suppose Assumptions~\ref{assump:input+noise} and \ref{assump:stability} hold. Suppose $\E[w_t w_t^\T] = I_{d_x}$. Let $\Gamma_x^{(t)}$ be as defined in \eqref{eqn:BDS_covariance_app}, and $\beta_{x,\delta}$ be as defined in \eqref{eqn:BDS_state_norm_app}. Then, for all $\delta \in (0,1)$, the event:
\begin{align}
    \lambda_{\min}\left( \sum_{t=1}^{T} \phi_t^\bds \phi_t^{\bds\T} \right) \geq \frac{T(1 \land \lambda_{\min}(\Gamma_x^{(\tau-1)}))}{8},
\end{align}
holds with probability at least $1- \delta$, provided that 
\begin{align}
   T &\gtrsim \tau\gamma\left (\log\left(\frac{2\tau}{\delta}\right) + (d_u + (d_u+1)d_x) \log\left(d_u + (d_u+1)\beta_{x,\delta}^2\right)\right),  \\
   \tau &\gtrsim 1+ \log \left( \frac{\kappa \left( \sqrt{d_u(d_u+1)}\beta_{x,\delta} + (d_u+1)\beta_{x,\delta}^2\right)}{(1 \land \lambda_{\min}(\Gamma_x^{(\tau-1)}))}\right) \bigg/ (1-\rho).
\end{align}
\end{proposition}
Note that $\{\phi_t^\bds\}_{t=1}^T$ does not have a block-dependence structure because of $x_t$ which gives rise to temporal dependence in $\{\phi_t^\bds\}_{t=1}^T$. Hence, we cannot directly apply Theorem~\ref{thm:persistence} to derive the persistence of excitation guarantees for the covariates $\{\phi_t^\bds\}_{t=1}^T$. Hence, we first analyze a surrogate covariate process $\{\phi_{t,\tau}^\bds\}_{t=1}^T$ (defined below), and use it to derive the persistence of excitation for the original covariate process $\{\phi_t^\bds\}_{t=1}^T$.
\begin{definition}[Truncated state vector $x_{t,\tau}$ and covariate $\phi_{t,\tau}^\bds$~\citep{sattar2022non}]\label{def:truncated vector app}
	Consider the state equation \eqref{eqn:BDS-FO}. Suppose $x_0 = 0$. Given, $t \geq \tau >0$, for each state $x_t$, we define its fictional proxy $x_{t,\tau}$ by resetting $x_{t-\tau+1} = 0$ but preserving the inputs $u_\ell$ and noise $w_\ell$ from $t-\tau+1$ to $t-1$. Alternately, $x_{t,\tau}$ is obtained by driving the system~\eqref{eqn:BDS-FO} with inputs $u_{\ell}'$ and additive noise $w_{\ell}'$ until time $t-1$, where
	\begin{align}
		u_{\ell}' = \begin{cases} {0} \quad \text{if} \quad \ell < t-\tau+1 \\
			u_{\ell} \;\; \text{else}
		\end{cases}, \quad\quad \text{and} \quad \quad \quad
		w_{\ell}' = \begin{cases} {0} \quad \text{if} \quad \ell < t-\tau+1 \\
			w_{\ell} \;\; \text{else}
		\end{cases}.
	\end{align}
 	We call the obtained state $x_{t,\tau}$ as the $\tau$-truncated (or simply truncated) state at time $t$. Similarly, we define the $\tau$-truncated covariance $\phi_{t,\tau}^\bds$ as follows:
    \begin{align}
    \phi_{t,\tau}^\bds =  \begin{bmatrix} u_{t} \\
    \util_{t} \otimes x_{t,\tau}\end{bmatrix} = \begin{bmatrix} u_{t} \\ x_{t,\tau} \\
    u_{t} \otimes x_{t,\tau}\end{bmatrix}. \label{eqn:truncated_covariates_def}
    \end{align}  
\end{definition}
To proceed, we first guarantee persistence of excitation with the truncated covariates  $\{\phi_{t,\tau}^\bds\}_{t=1}^T$, and later combine it with an upper bound on the truncation error to obtain a lower bound on $\lambda_{\min}\big( \sum_{t=1}^{T} \phi_t^{\bds} (\phi_t^{\bds})^\T \big)$.

\subsubsection{Stability Implies Hypercontractivity}
From \eqref{eqn:bilinear sys state app}, note that the truncated states $x_{t,\tau}$ is given by,
\begin{align}
    x_{t,\tau} = \bar{G} \,\begin{bmatrix}
    u_{t-1} \\ 
    \util_{t-1} \otimes u_{t-2} \\ 
    \util_{t-1} \otimes \util_{t-2} \otimes u_{t-3} \\ \vdots \\ 
    \util_{t-1} \otimes \util_{t-2} \otimes \cdots \otimes u_{t-\tau+1} 
    \end{bmatrix} + \bar{F} \,\begin{bmatrix}
    w_{t-1} \\ 
    \util_{t-1} \otimes w_{t-2} \\ 
    \util_{t-1} \otimes \util_{t-2} \otimes w_{t-3} \\ \vdots \\ 
    \util_{t-1} \otimes \util_{t-2} \otimes \cdots \otimes w_{t-\tau+1} 
    \end{bmatrix}. \label{eqn:truncated_states_BDS_app}
\end{align}
Combining this with Assumption~\ref{assump:input+noise} and $\E[w_t w_t^\T] = \sigma^2 I_{d_x}$, we have $\E[x_{t,\tau}] = 0$, and
\begin{align}
    \E[x_{t,\tau} x_{t,\tau}^\T]  = \sum_{\ell = 1}^{\tau - 1} \bar{G}_\ell \bar{G}_\ell^\top + \sigma^2 \sum_{\ell = 1}^{\tau-1} \bar{F}_\ell \bar{F}_\ell^\top = \Gamma_x^{(\tau - 1)}. \label{eqn:cov_truncated_states_BDS_app}
\end{align}
From \eqref{eqn:truncated_states_BDS_app} and \eqref{eqn:cov_truncated_states_BDS_app}, it is clear that the truncated covariates $\{\phi_{t,\tau}^{\bds}\}_{t \geq 1}$ has a block-dependent structure with block-length $\tau$. Moreover, $ \E \left[ \phi_{t,\tau}^{\bds}\right] = 0$, and
\begin{align}
    \E \left[ \phi_{t,\tau}^{\bds}\phi_{t,\tau}^{\bds\T}\right] &= \E \left[ \begin{bmatrix} u_{t} \\ x_{t,\tau} \\
    u_{t} \otimes x_{t,\tau}\end{bmatrix} \begin{bmatrix} u_{t}^\T & x_{t,\tau}^\T &
    u_{t}^\T \otimes x_{t,\tau}^\T \end{bmatrix}\right], \nn \\
    &=  \begin{bmatrix}
        I_{d_u} & 0 & 0 \\
        0 & \Gamma_x^{(\tau-1)} & 0 \\
        0 & 0 & I_{d_u} \otimes \Gamma_x^{(\tau-1)}
    \end{bmatrix} =: \Gamma_{\phi}^{(\tau-1)}, \label{eqn:Gamma_phi_def_app}
\end{align}
which is non-isotropic. However, as noted in Remark \ref{remark:persistence}, we can scale $x_{t,\tau}$ by $(\Gamma_x^{(\tau-1)})^{-1/2}$ to make it isotropic. We are now ready to show that the scaled truncated covariates $\{\phi_{t,\tau}^{'\bds} := (\Gamma_{\phi}^{(\tau-1)})^{-1/2}\phi_{t,\tau}^{\bds}\}_{t \geq 1}$ is $(4,2,\bar\gamma)$-hypercontractive for some $\bar\gamma >0$.

\begin{lemma}[Hypercontractivity of $\phi_{t,\tau}^{'\bds}$]\label{lemma:hyper_phi_bds_app}
    Let $\phi_{t,\tau}^{\bds}$ be as in \eqref{eqn:truncated_covariates_def} and $\phi_{t,\tau}^{'\bds} := (\Gamma_{\phi}^{(\tau-1)})^{-1/2}\phi_{t,\tau}^{\bds}$, where $\Gamma_{\phi}^{(\tau-1)}$ is as defined in \eqref{eqn:Gamma_phi_def_app}. Under Assumptions~\ref{assump:input+noise}, \ref{assump:stability}, and $\E[w_t w_t^\T] = \sigma^2 I_{d_x}$, we have $\EE[(v^\top\phi_{t,\tau}^{'\bds})^4] \le \bar\gamma\, \EE[(v^\top\phi_{t,\tau}^{'\bds})^2]^2$, for all $v \in \R^{d_u+(d_u+1)d_x}$, $t \in [1,T]$, and $\tau > 0$, where $\bar \gamma := \bigg(\sqrt{3} + (3+ \sqrt{3}) \frac{\kappa^2(\sigma+\norm{B}_\op)^2}{(1-\rho)^2\lambda_{\min}\left(\Gamma_x^{(\tau-1)}\right)} \bigg)^2$. 
\end{lemma}
\begin{proof}
Throughout this proof, for notational convenience, we use $\Gamma_x$ instead of $\Gamma_x^{(\tau-1)}$ to denote the covariance of $x_{t,\tau}$. To begin, consider the isotropic covariates as follows: 
\begin{align}
    \phi_{t,\tau}^{'\bds}  = \begin{bmatrix} u_{t} \\ \Gamma_x^{-1/2} x_{t,\tau} \\
    u_{t} \otimes \Gamma_x^{-1/2} x_{t,\tau} \end{bmatrix} =: 
    \begin{bmatrix} u_{t} \\ x_{t,\tau}' \\
    u_{t} \otimes  x_{t,\tau}' \end{bmatrix}. \label{eqn:isotropic_truncated_covariates_def}
\end{align}
 Let $v = [v_1^\T~~v_2^\T~~v_3^\T]^\T \in \R^{d_u+ (d_u+1)d_x}$ be an arbitrary vector. Then, the second moment of the random variable $v^\top \phi_{t,\tau}^{'\bds}$ is given by,
\begin{align}
   \E [(v^\T \phi_{t,\tau}^{'\bds})^2] = \E \left[v^\T \phi_{t,\tau}^{'\bds} (\phi_{t,\tau}^{'\bds})^\T v\right] = \tn{v}^2. \label{eqn:second_moment_bds_app}
\end{align}
 Next, we will upper bound the fourth moment of the random variable $v^\top \phi_{t,\tau}^{'\bds}$ as follows:
\begin{align}
    \E [( v^\top \phi_{t,\tau}^{'\bds})^4]^{1/4}  &= \E\left[ \left(v_1^\T u_t + v_2^\T x_{t,\tau}' + v_3^\T (u_{t} \otimes x_{t,\tau}') \right)^4 \right]^{1/4}, \nn \\
    &= \E\left[ \left(v_1^\T u_t + v_2^\T x_{t,\tau}' + u_{t}^\T V_3 x_{t,\tau}' \right)^4 \right]^{1/4}, \nn \\
    & \leqsym{a}  \E\left[ \left(v_1^\T u_t\right)^4 \right]^{1/4} +  \E\left[ \left(v_2^\T x_{t,\tau}'\right)^4 \right]^{1/4} +  \E\left[ \left(u_{t}^\T V_3 x_{t,\tau}' \right)^4 \right]^{1/4}, \label{eqn:fourth_moment_bds_app_v1}
\end{align}
where we get (a) from using Minkowski's inequality. In the following, we will upper bound each term in \eqref{eqn:fourth_moment_bds_app_v1} separately. The first term can be easily upper bounded as follows:
\begin{align}
    \E\left[ \left(v_1^\T u_t\right)^4 \right]^{1/4} \leq \left(\frac{3d_u}{d_u+2}\right)^{1/4} \tn{v_1} \leq (3)^{1/4} \tn{v_1}. \label{eqn:firt_term_bds_app}
\end{align}
Next, using Minkowski's inequality along-with Assumptions~\ref{assump:input+noise} and \ref{assump:stability}, the second term in \eqref{eqn:fourth_moment_bds_app_v1} can be upper bounded as follows:
\begin{align}
    &\E\left[ \left(v_2^\T x_{t,\tau}'\right)^4 \right]^{1/4} \nn \\
    &=  \E\left[ \left(v_2^\T \Gamma_{x}^{-1/2}\sum_{\ell = 1}^{\tau-1} \left(\prod_{k = 0}^{\ell-1} (u_{t-k} \circ A)\right) B u_{t-\ell} + v_2^\T\Gamma_{x}^{-1/2} \sum_{\ell = 1}^{\tau-1} \left(\prod_{k = 0}^{\ell-1} (u_{t-k} \circ A)\right)  w_{t-\ell}\right)^4 \right]^{1/4}, \nn \\
    &\leqsym{b} \sum_{\ell= 1}^{\tau-1} \E\left[ \left(v_2^\T \Gamma_{x}^{-1/2} \left(\prod_{k = 0}^{\ell-1} (u_{t-k} \circ A)\right) B u_{t-\ell}\right)^4 \right]^{1/4}, \nn \\
    &+ \sum_{\ell= 1}^{\tau-1} \E\left[ \left( v_2^\T\Gamma_{x}^{-1/2} \left(\prod_{k = 0}^{\ell-1} (u_{t-k} \circ A)\right)  w_{t-\ell}\right)^4 \right]^{1/4}, \nn \\
    &\leq  \sum_{\ell= 1}^{\tau - 1} (3)^{1/4} \bignorm{v_2^\T \Gamma_{x}^{-1/2} \left(\prod_{k = 0}^{\ell-1} (u_{t-k} \circ A)\right) B}_{\ell_2} + \sum_{\ell= 1}^{\tau-1} (3)^{1/4} \sigma \bignorm{v_2^\T \Gamma_{\xtil}^{-1/2} \left(\prod_{k = 0}^{\ell-1} (u_{t-k} \circ A)\right)}_{\ell_2}, \nn \\
    & \leq (3)^{1/4}\tn{v_2} \bignorm{\Gamma_{x}^{-1/2}}_{\op}(\sigma+\norm{B}_\op)\sum_{\ell = 1}^{\tau-1}\bignorm{\prod_{k = 0}^{\ell-1} (u_{t-k} \circ A)}_\op, \nn \\
    & \leqsym{c} (3)^{1/4}\tn{v_2} \bignorm{\Gamma_{x}^{-1/2}}_{\op}(\sigma+\norm{B}_\op)\sum_{\ell = 1}^{L-1} \kappa \rho^{\ell-1}, \nn \\
    & \leq (3)^{1/4}\tn{v_2} \bignorm{\Gamma_{x}^{-1/2}}_{\op}\frac{\kappa(\sigma+\norm{B}_\op)}{1-\rho}, \label{eqn:second_term_bds_app}
\end{align}
where we get (b) from using Minkowski's inequality and (c) from Assumption~\ref{assump:stability}. Finally, using the law of total expectation along-with \eqref{eqn:firt_term_bds_app} and \eqref{eqn:second_moment_bds_app}, the third term in \eqref{eqn:fourth_moment_bds_app_v1} can be upper bounded as follows:
\begin{align}
    \E\left[ \left(u_{t}^\T V_3 x_{t,\tau}' \right)^4 \right]^{1/4} \leq (9)^{1/4}\tn{v_3} \bignorm{\Gamma_{x}^{-1/2}}_{\op}\frac{\kappa(\sigma+\norm{B}_\op)}{1-\rho}. \label{eqn:third_term_bds_app}
\end{align}
Hence, plugging in \eqref{eqn:firt_term_bds_app}, \eqref{eqn:second_term_bds_app}, and \eqref{eqn:third_term_bds_app} into \eqref{eqn:fourth_moment_bds_app_v1}, we get the following:
\begin{align}
      \E [( v^\top \phi_{t,\tau}^{'\bds})^4]^{1/4} &\leq (3)^{1/4} \tn{v_1} + (3)^{1/4}\tn{v_2} \bignorm{\Gamma_{x}^{-1/2}}_{\op}\frac{\kappa(\sigma+\norm{B}_\op)}{1-\rho} \nn \\
     &+ (9)^{1/4}\tn{v_3} \bignorm{\Gamma_{x}^{-1/2}}_{\op}\frac{\kappa(\sigma+\norm{B}_\op)}{1-\rho} \nn \\
     &\leqsym{d}  \left(3^{1/2} + (3^{1/2}+ 3)\bignorm{\Gamma_{x}^{-1}}_{\op} \frac{\kappa^2(\sigma+\norm{B}_\op)^2}{(1-\rho)^2} \right)^{1/2} \left(\tn{v_1}^2 + \tn{v_2}^2 + \tn{v_3}^2 \right)^{1/2}, \nn \\
     & = \left(3^{1/2} + (3^{1/2}+ 3)\bignorm{\Gamma_{x}^{-1}}_{\op} \frac{\kappa^2(\sigma+\norm{B}_\op)^2}{(1-\rho)^2} \right)^{1/2} \tn{v}, \label{eqn:fourth_moment_bds_app_v2}
\end{align}
where we get (d) by using Cauchy–Schwarz inequality. Combining this with \eqref{eqn:second_moment_bds_app}, we have
\begin{align}
    \E [( v^\top \phi_{t,\tau}^{'\bds})^4] \leq \left(3^{1/2} + (3^{1/2}+ 3)\bignorm{\Gamma_{x}^{-1}}_{\op} \frac{\kappa^2(\sigma+\norm{B}_\op)^2}{(1-\rho)^2} \right)^{2} \E [( v^\top \phi_{t,\tau}^{'\bds})^2]^2.
\end{align}
This completes the proof.

\end{proof}

\subsubsection{Stability Implies Smaller Truncation Bias}
Next, we show that under Assumptions~\ref{assump:input+noise} and \ref{assump:stability}, if we choose $\tau$ to be large enough, then the distance between the true empirical covariance $\sum_{t=1}^T \phi_t^{\bds} \phi_t^{\bds\T}$ and the surrogate empirical covariance $\sum_{t=1}^T \phi_{t,\tau}^{\bds} \phi_{t,\tau}^{\bds\T}$.

\begin{lemma}[Truncation bias]\label{lemma:truncation_bias_phi_bds_app}
    Consider the same setting of Lemma~\ref{lemma:hyper_phi_bds_app}. Furthermore, suppose $\tau$ satisfies,
    \begin{align}
        \tau \geq 1+ \log \left( \frac{32 \kappa \left( \sqrt{d_u(d_u+1)}\beta_{x,\delta} + (d_u+1)\beta_{x,\delta}^2\right)}{(1 \land \lambda_{\min}(\Gamma_x^{(\tau-1)}))}\right) \bigg/ (1-\rho),
        \label{eqn:tau_size_bds_app}
    \end{align}
     where $\Gamma_{\phi}^{(\tau-1)}$ is as defined in \eqref{eqn:Gamma_phi_def_app}, and we set $\beta_{x,\delta} := \kappa \sqrt{(\sigma^2 + \opn{B}^2)(d_x + \log(T/\delta))}/(1-\rho)$. Then, with probability at least $1-\delta$, we have $\tn{x_t} \leq \beta_{x,\delta}$ for all $t \in [1,T]$, and
     \begin{align}
         \bignorm{\sum_{t=1}^T \phi_t^\bds \phi_t^{\bds\T} - \sum_{t=1}^T \phi_{t,\tau}^\bds \phi_{t,\tau}^{\bds\T}}_\op \leq  \frac{T(1 \land \lambda_{\min}(\Gamma_x^{(\tau-1)}))}{8}.
     \end{align}
\end{lemma}
\begin{proof}
We begin the proof by upper bounding the Euclidean distance between $\phi_t^\bds$ and $\phi_{t,\tau}^\bds$ as follows:     

\begin{align}
    \tn{\phi_t^\bds - \phi_{t,\tau}^\bds} &= \bignorm{\begin{bmatrix} u_{t} \\ x_t \\
    u_{t} \otimes x_{t}\end{bmatrix} -\begin{bmatrix} u_{t} \\ x_{t,\tau} \\
    u_{t} \otimes x_{t,\tau}\end{bmatrix} }_{\ell_2} ,\nn \\
    & = \sqrt{1 + d_u} \tn{x_t - x_{t,\tau}}, \nn \\
    &= \sqrt{1 + d_u} \bignorm{\left( \prod_{\ell = 1}^{\tau-1} (u_{t-\ell}\circ A )  \right) x_{t-\tau+1}}_{\ell_2}, \nn \\
    & \leq \sqrt{1 + d_u} \bignorm{\left( \prod_{\ell = 1}^{\tau-1} (u_{t-\ell}\circ A )  \right)}_\op \tn{x_{t-\tau+1}}, \nn \\
    & \leq \sqrt{1 + d_u}\, \kappa \rho^{\tau-1} \tn{x_{t-\tau+1}}. \label{eqn:phi_t_minus_truncated_bds_app_v1}
\end{align}
Hence, we need to upper bound the Euclidean norm of the states in \eqref{eqn:BDS-FO}. Under Assumptions~\ref{assump:input+noise} and \ref{assump:stability}, for all $t \in [1,T]$, we have
\begin{align}
    \tn{x_t} &= \tn{\sum_{\ell = 1}^{t} \left(\prod_{k = 1}^{\ell-1} (u_{t-k} \circ A)\right) B u_{t-\ell} + \sum_{\ell = 1}^t \left(\prod_{k = 1}^{\ell-1} (u_{t-k} \circ A)\right)  w_{t-\ell}}, \nn \\
    & \leq \sum_{\ell = 1}^{t} \opn{\prod_{k = 1}^{\ell-1} (u_{t-k} \circ A)}  \tn{B u_{t-\ell} + w_{t-\ell}}, \nn \\
    & \leqsym{a} \sum_{\ell = 1}^{t} \kappa\rho^{\ell-1} \sqrt{\sigma^2 + \opn{B}^2}\sqrt{d_x + \log(T/\delta)}, \nn \\
&\leq \frac{\kappa(1-\rho^t)}{1-\rho} \sqrt{\sigma^2 + \opn{B}^2}\sqrt{d_x + \log(T/\delta)}, \nn \\
    &\leq \frac{\kappa\sqrt{\sigma^2 + \opn{B}^2}}{1-\rho} \sqrt{d_x + \log(T/\delta)} =: \beta_{x,\delta} \label{eqn:x_t_bound_bds_app}
\end{align}
with probability at least $1-\delta$, where we obtain (a) from the observation that $B u_{t-\ell} + w_{t-\ell}$ is a $(\sigma^2 + \opn{B}^2)$-subgaussian random vector in $\R^{d_x}$. Hence, plugging \eqref{eqn:x_t_bound_bds_app} into \eqref{eqn:phi_t_minus_truncated_bds_app_v1}, with probability at least $1-\delta$ we have

\begin{align}
    \tn{\phi_t^\bds - \phi_{t,\tau}^\bds} & \leq \sqrt{1 + d_u}\, \kappa \rho^{\tau-1} \frac{\kappa\sqrt{\sigma^2 + \opn{B}^2}}{1-\rho} \sqrt{d_x + \log(T/\delta)}, \nn \\
    &\leq \frac{\kappa^2\rho^{\tau-1} \sqrt{\sigma^2 + \opn{B}^2}}{1-\rho} \sqrt{(d_u+1)(d_x + \log(T/\delta))} \label{eqn:phi_t_minus_truncated_bds_app_v2}
\end{align}
Moreover, using \eqref{eqn:x_t_bound_bds_app}, we can also upper bound the Euclidean norm of the covariates $\phi_t^\bds$ as follows:
\begin{align}
\tn{\phi_t^\bds}^2 &= \tn{u_t}^2 + \tn{x_t}^2 + \tn{u_t \otimes x_t}^2, \nn \\
&= \tn{u_t}^2 + \left(\tn{u_t}^2+1\right)\tn{x_t}^2, \nn \\
& \leq d_u + \frac{\kappa^2(\sigma^2 + \opn{B}^2)}{(1-\rho)^2} (d_u+1)(d_x + \log(T/\delta)). \label{eqn:phi_t_bound_bds_app}
\end{align}
Finally, using \eqref{eqn:phi_t_minus_truncated_bds_app_v2} and \eqref{eqn:phi_t_bound_bds_app}, with probability at least $1-\delta$, we have

\begin{align}
  \bignorm{\sum_{t=1}^T \phi_t^\bds (\phi_t^\bds)^\T - \sum_{t=1}^T \phi_{t,\tau}^\bds (\phi_{t,\tau}^\bds)^\T}_\op   &\leq  \sum_{t=1}^T \bignorm{\phi_t^\bds (\phi_t^\bds)^\T - \phi_{t,\tau}^\bds (\phi_{t,\tau}^\bds)^\T}_\op, \nn \\
  &\leq \sum_{t=1}^T \left(\tn{\phi_{t}^\bds} + \tn{\phi_{t,\tau}^\bds} \right) \tn{\phi_{t}^\bds - \phi_{t,\tau}^\bds}, \nn \\
  &\leq \sum_{t=1}^T 2\left(\sqrt{d_u} + \sqrt{d_u +1}\tn{x_t} \right)\sqrt{d_u+1}\, \kappa \rho^{\tau-1} \tn{x_{t-\tau+1}} \nn \\
&\leq 2 T \kappa\rho^{\tau-1} \bigg( \frac{\kappa\sqrt{\sigma^2 + \opn{B}^2}}{1-\rho} \sqrt{d_u(d_u+1)(d_x + \log(T/\delta))}\nn \\
  & + \frac{\kappa^2(\sigma^2 + \opn{B}^2)}{(1-\rho)^2} (d_u+1)(d_x + \log(T/\delta))\bigg), \nn \\
  &\leq 2 T \kappa^2\rho^{\tau-1} \frac{\sqrt{\sigma^2 + \opn{B}^2}}{1-\rho} \bigg(  \sqrt{d_u(d_u+1)(d_x + \log(T/\delta))}\nn \\
  & + \frac{\kappa\sqrt{\sigma^2 + \opn{B}^2}}{(1-\rho)} (d_u+1)(d_x + \log(T/\delta))\bigg), \nn \\
& \leq \frac{T(1 \land \lambda_{\min}(\Gamma_x^{(\tau-1)}))}{8},
\end{align}
where we get the last inequality by choosing $\tau$ via
\begin{align}
    &2 T \kappa\rho^{\tau-1} \left( \sqrt{d_u(d_u+1)}\beta_{x,\delta} + (d_u+1)\beta_{x,\delta}^2\right) \leq \frac{T(1 \land \lambda_{\min}(\Gamma_x^{(\tau-1)}))}{16}, \nn \\
    \iff & \rho^{\tau-1} \leq \frac{(1 \land \lambda_{\min}(\Gamma_x^{(\tau-1)}))}{32 \kappa \left( \sqrt{d_u(d_u+1)}\beta_{x,\delta} + (d_u+1)\beta_{x,\delta}^2\right)}, \nn \\
    \impliedby & \tau \geq 1+ \log \left( \frac{32 \kappa \left( \sqrt{d_u(d_u+1)}\beta_{x,\delta} + (d_u+1)\beta_{x,\delta}^2\right)}{(1 \land \lambda_{\min}(\Gamma_x^{(\tau-1)}))}\right) \bigg/ (1-\rho),
\end{align}
where  we use $\beta_{x,\delta}$ to denote the high probability Euclidean norm bound on the states $x_t$ in \eqref{eqn:BDS-FO},
\begin{align}
    \beta_{x,\delta} := \frac{\kappa\sqrt{\sigma^2 + \opn{B}^2}}{1-\rho} \sqrt{d_x + \log(T/\delta)}.
\end{align}
This completes the proof.
\end{proof}

\subsubsection{Finalizing the Proof of Proposition~\ref{prop:persistence_BDS_FO_app}}
To prove Proposition~\ref{prop:persistence_BDS_FO_app}, we first apply Theorem~\ref{thm:persistence} to the truncated covariates $\lbrace \phi_{t,\tau}^\bds \rbrace_{t =1}^T$ in \eqref{eqn:truncated_covariates_def}. Specifically, applying Theorem~\ref{thm:persistence} (see Remark~\ref{remark:persistence}) with $\gamma = \bigg(\sqrt{3} + (3+ \sqrt{3}) \frac{\kappa^2(\sigma+\norm{B}_\op)^2}{(1-\rho)^2\lambda_{\min}\left(\Gamma_x^{(\tau-1)}\right)} \bigg)^2$, $d_\phi = d_u+(d_u+1)d_x$, and using a $1/(4 \beta_{\phi,\delta})$-net covering argument, with probability at least $1-\delta$, we have,
\begin{align}
    \lambda_{\min}\left( \sum_{t=1}^{T} \phi_{t,\tau}^\bds \phi_{t,\tau}^{\bds\T} \right) \geq \frac{T(1 \land \lambda_{\min}(\Gamma_x^{(\tau-1)}))}{4}, \label{eqn:persistence_of_truncated_app}
\end{align}
provided that 
\begin{align}
    T &\geq  32\tau\gamma\left (\log\left(\frac{2\tau}{\delta}\right) + (d_u + (d_u+1)d_x) \log\left(1+8\left(d_u + (d_u+1)\beta_{x,\delta}^2\right)\right) \right). \label{eqn:trajectory_size_bds_app}
\end{align}
Combining this with Lemma~\ref{lemma:truncation_bias_phi_bds_app} and Weyl's inequality, with probability at least $1-\delta$, we have
\begin{align}
    \lambda_{\min}\left( \sum_{t=1}^{T} \phi_{t}^\bds \phi_{t}^{\bds\T} \right) &\geq \lambda_{\min}\left( \sum_{t=1}^{T} \phi_{t,\tau}^\bds \phi_{t,\tau}^{\bds\T} \right)  - \bignorm{\sum_{t=1}^T \phi_t^\bds \phi_t^{\bds\T} - \sum_{t=1}^T \phi_{t,\tau}^\bds \phi_{t,\tau}^{\bds\T}}_\op, \nn \\
    & \geq \frac{T(1 \land \lambda_{\min}(\Gamma_x^{(\tau-1)}))}{8}
\end{align}
provided that $T$ satisfies \eqref{eqn:trajectory_size_bds_app}, and $\tau$ satisfies \eqref{eqn:tau_size_bds_app}. This completes the proof of Proposition~\ref{prop:persistence_BDS_FO_app}.

\subsection{Persistence of Excitation for Learning~\eqref{eqn:BDS-PO}}
In this subsection, we will derive the persistence of excitation guarantees for learning partially observed bilinear dynamical systems, governed by state-space equation \eqref{eqn:BDS-PO}. The proofs of the results in this subsection can be found in \cite{sattar2025finite_arXiv}. Note that, in the case of \eqref{eqn:BDS-PO}, the covariates take the following form:
\begin{equation}
\begin{aligned} \label{eqn:phi_t_BDS_PO_app}
   \phi_t^\pobds  = \begin{bmatrix} u_{t} \\u_{t-1} \\  \util_{t-1} \otimes u_{t-2} \\ \util_{t-1} \otimes \util_{t-2} \otimes u_{t-3} \\ \vdots \\  \util_{t-1} \otimes \util_{t-2} \otimes \cdots \otimes u_{t-\tau+1}  \end{bmatrix}, \quad \text{where,} \quad
  \util_t := \begin{bmatrix} 1 \\ u_t \end{bmatrix}. 
\end{aligned}
\end{equation}

\begin{proposition}[Persistence of Excitation for Learning~\eqref{eqn:BDS-PO}]\label{prop:persistence_BDS_PO_app}
Consider the covariates $\lbrace \phi_t^{\pobds}\rbrace_{t =\tau+1}^T$ defined in \eqref{eqn:phi_t_BDS_PO_app}. Suppose Assumption~\ref{assump:input+noise}\,\emph{(\textbf{a})} holds. Then, for all $\delta \in (0,1)$, the event:
\begin{align}
    \lambda_{\min}\left( \sum_{t=\tau}^{T} \phi_t^{\pobds} (\phi_t^{\pobds})^\T \right) \geq (T-\tau+1)/4,
\end{align}
holds with probability at least $1- \delta$, provided that 
\begin{align}
   T - \tau+1 \gtrsim  \tau(\tau-1) 3^{\tau} \bigg(\log\left(\frac{2\tau}{\delta}\right) + (d_u+1)^{\tau-1} \log\big({(d_u+1)^{\tau-1}}\big) \bigg). \label{eqn:trajectory_size_BDS_PO_app}
\end{align}
\end{proposition}

Before we begin the proof of Proposition~\ref{prop:persistence_BDS_PO_app}, we show that the covariates $\{\phi_t^{\pobds}\}_{t \geq \tau}$, given by \eqref{eqn:phi_t_BDS_PO_app}, satisfy the conditions required by Theorem~\ref{thm:persistence}. 

\subsubsection{Hypercontractivity of the Covariates in \eqref{eqn:phi_t_BDS_PO_app}}
 
From \eqref{eqn:phi_t_BDS_PO_app} and Assumption~\ref{assump:input+noise}(\textbf{a}), it is clear that $\{\phi_t^{\pobds}\}_{t \geq \tau}$ has a block-dependent structure with block-length $\tau$. Moreover, Assumption~\ref{assump:input+noise}(\textbf{a}) also implies that $ \E \left[ \phi_t^{\pobds}\right] = 0$, and
\begin{align}
    \E \left[ \phi_t^{\pobds}(\phi_t^{\pobds})^\T\right]  \eqsym{a} I_{(d_u+1)^{\tau-1} + d_u-1},
\end{align}
where we get (a) from the proof of Lemma 12 in \cite{sattar2025finite_arXiv}. Next, we show that the covariates $\{\phi_t^{\pobds}\}_{t \geq \tau}$ is $(4,2,(\tau-1) 3^\tau)$-hypercontractive.
\begin{lemma}[Hypercontractivity of $\phi_t^{\pobds}$]\label{lemma:hyper_phi_bds_po_app}
    Let $\phi_t^{\pobds}$ be as defined in \eqref{eqn:phi_t_BDS_PO_app}. Under Assumption~\ref{assump:input+noise}\,\emph{(\textbf{a})}, we have $\EE[(v^\top \phi_t^{\pobds})^4] \le (\tau-1) 3^\tau\, \EE[(v^\top \phi_t^{\pobds})^2]^2$, for all $v \in \R^{(d_u+1)^{\tau-1} + d_u -1}$, and all $t \in [\tau+1,T]$. 
\end{lemma}
The proof of Lemma~\ref{lemma:hyper_phi_bds_po_app} is given in the Appendix B.2 in \cite{sattar2025finite_arXiv}. Note that, the hypercontractivity parameter $\gamma = (\tau-1) 3^\tau $ does not depend on the dimension $d_u$ of the input $u_t$. The dependence on $\tau$ seems to be unavoidable. For example, suppose $\{u_i\}_{i=1}^\tau \distas \Ncal(0, \sigma^2)$. Then we have, $\E\left[ \left(\prod_{i=1}^\tau u_i\right)^4 \right] = 3^\tau \sigma^{4\tau} = 3^\tau \,\E\left[ \left(\prod_{i=1}^\tau u_i\right)^2 \right]^2$. Hence, the hypercontractivity parameter of $\phi_t^\pobds$ scales with $3^\tau$ even in the case of standard normal distribution. 

\subsubsection{Finalizing the Proof of Proposition~\ref{prop:persistence_BDS_PO_app}}
We are now ready to apply Theorem~\ref{thm:persistence} to the covariates $\lbrace \phi_t^{\pobds}\rbrace_{t =\tau}^T$. Specifically, applying Theorem~\ref{thm:persistence} with $\gamma = (\tau - 1) 3^\tau$, $d_\phi = (d_u+1)^{\tau-1} + d_u -1 \leq 2(d_u+1)^{\tau-1}$, and using a $1/(4 d_{\phi})$-net covering argument (see Remark~\ref{remark:covering_when_bounded_app}), we get the statement of Proposition~\ref{prop:persistence_BDS_PO_app}. 
\renewcommand{\Ab}{A}
\renewcommand{\vN}{N}
\section{Proofs of Approximation/Truncation Bias Results}\label{app:error bias}
In this Appendix, we present the proof of our main result on upper bounding the estimation error term due to approximation/truncation bias, which is denoted by $E_2$ in \eqref{eqn:estimation_error}. Specifically, we will upper bound $E_2$ by controlling the term $\opn{\sum_{t=1}^T\phi_t z_t^\T}$ for the dynamical systems we studied in Section~\ref{sec:sysid}. Recall that $z_t = 0$ in the case of \eqref{eqn:BDS-FO}.

\subsection{Upper Bounding Approximation Error for \eqref{eqn:BDS-PO}}
In this subsection, we will upper bound the term $\opn{\sum_{t=\tau}^T\phi_t z_t^\T}$ for learning partially observed bilinear dynamical systems, governed by state-space equation \eqref{eqn:BDS-PO}. The proofs of the results in this subsection are adapted from \cite{sattar2025finite_arXiv}.  
Recall that, in the case of \eqref{eqn:BDS-PO}, the truncation bias $z_t$ takes the following form:
\begin{equation}
\begin{aligned} \label{eqn:z_t_BDS_PO_app}
   z_t  = C  \left( \prod_{\ell = 1}^{\tau-1} (u_{t-\ell}\circ A )  \right) x_{t-\tau+1}
\end{aligned}
\end{equation}

\begin{proposition}[Truncation bias for \eqref{eqn:BDS-PO}]\label{prop:truncation_bias_bds_po_app}
    Consider the system~\eqref{eqn:BDS-PO}. Let $\phi_t$ and $z_t$ be as in \eqref{eqn:phi_t_BDS_PO_app} and \eqref{eqn:z_t_BDS_PO_app}, respectively. Suppose Assumptions~\ref{assump:input+noise} and \ref{assump:stability} hold. Let $\delta \in (0,1)$ and $T \ge \tau-1$. Then, the event:
        \begin{align*}
        \left\Vert \sum_{t=\tau}^T \phi_t z_t^\top \right\Vert_\op \le  \frac{c_1 \kappa \rho^{\tau-1} }{1-\rho}\sqrt{(T-\tau+1)(d_u+1)^{\tau-1}\left( \log\left(\frac{e}{\delta}\right) + d_y + (d_u+1)^{\tau-1} \right)}      
    \end{align*}
    holds with probability at least $1-\delta$, with $c_1 = \Vert \vC  \Vert_\op \sqrt{2\log(9)(\Vert \vB \Vert_\op^2 + \sigma^2)}$.
\end{proposition}

\begin{proof}
First, using the variational form of the operator norm, we see that:
\begin{align*}
   \left \Vert \sum_{t=\tau}^T \phi_t  z_t^\top   \right\Vert_\op  = \sup_{\theta \in \cS^{d_{\phi} - 1}, \lambda \in \cS^{d_y - 1}}   \sum_{t=\tau}^T (\theta^\top\phi_t) (\lambda^\top z_t). 
\end{align*}
We use a $1/4$-net argument to bound the supremum. Let $\Mcal$ (resp. $\Ncal$) be $1/4$-nets with minimal cardinality of the sphere $\cS^{d_{\phi} - 1}$ (resp. $\cS^{d_y - 1}$). Thus, using Lemma \ref{lem:net argument}   we obtain: for all $r > 0$: 
\begin{align}\label{eqn:eq net argument}
    \PP\left(    \left \Vert \sum_{t=\tau}^T \phi_t  z_t^\top   \right\Vert_\op  > 2 r \right) \le 9^{d_{\phi}+ d_y}\max_{\theta \in \Mcal, \lambda \in \Ncal} \PP\left(  \sum_{t=\tau}^T (\theta^\top\phi_t) (\lambda^\top z_t) > r \right).
\end{align}
It remains to bound $\sum_{t=\tau}^T (\theta^\top\phi_t) (\lambda^\top z_t)$ with high probability uniformly over the unit spheres. Let $\theta \in \cS^{d_{\phi} - 1}$, and $\lambda \in \cS^{d_y - 1}$. We have: 
    \begin{align*}
        & \sum_{t=\tau}^T (\theta^\top\phi_t) (\lambda^\top z_t) \\
        & \qquad =  \sum_{t=\tau}^T (\theta^\top\phi_t) (\lambda^\top  C) \left( \prod_{\ell = 1}^{\tau-1} (u_{t-\ell}\circ \Ab )  \right)  \vx_{t-\tau+1}, \\
        & \qquad = \sum_{t=1}^{T-\tau+1} (\theta^\top\phi_{t+\tau-1}) (\lambda^\top  C) \left( \prod_{\ell = 1}^{\tau-1} (u_{t +\tau -1 - \ell}\circ \Ab )  \right)  \vx_{t}, \\
         & \qquad \overset{(a)}{=} \sum_{t=1}^{T-\tau+1} (\theta^\top\phi_{t+\tau-1}) (\lambda^\top  C) \left( \prod_{\ell = 1}^{\tau-1} (u_{t +\tau -1 - \ell}\circ \Ab )  \right)  \sum_{s = 0}^{t-1} \left(\prod_{k = 1}^{t-s-1} (\vu_{t-k} \circ \vA)\right) (\vB \vu_{s} + \vw_{s}), \\
         & \qquad = \sum_{t=1}^{T-\tau+1} \sum_{s = 0}^{t-1}  (\theta^\top\phi_{t+\tau-1}) (\lambda^\top  C) \left( \prod_{\ell = 1}^{\tau-1} (u_{t +\tau-1- \ell}\circ \Ab )  \right)  \left(\prod_{k = 1}^{t-s-1} (\vu_{t-k} \circ \vA)\right) (\vB \vu_{s} + \vw_{s}), 
         \\ & \qquad \overset{(b)}{=} \sum_{t=0}^{T-\tau} \sum_{s = 0}^{t}  \underbrace{(\theta^\top\phi_{t+\tau}) (\lambda^\top  C) \left( \prod_{\ell = 1}^{\tau-1} (u_{t +\tau- \ell}\circ \Ab )  \right)}_{:= M(\vu_{t+1}, \dots, \vu_{t+\tau})}  \underbrace{\left(\prod_{k = 1}^{t-s} (\vu_{t+1-k} \circ \vA)\right) }_{:=N(\vu_{s+1}, \dots, \vu_{t})} (\vB \vu_{s} + \vw_{s}), 
    \end{align*}
    where we used in $(a)$, the dynamics \eqref{eqn:BDS-FO} to express $\vx_t$ in terms of $(\vu_{s}, \vw_s)_{s < t}$ (e.g., see \eqref{eqn:bilinear sys state}). We also introduce in $(b)$ the quantities $M$ and $N$ which depends on inputs. Next, we perform next a change of indices to obtain: 
   \begin{align}\label{eq:martingale bias}
        \sum_{t=\tau}^T (\theta^\top\phi_t) (\lambda^\top z_t)  
        & = \sum_{s=0}^{T-\tau} \underbrace{\left( \sum_{t = s}^{T-\tau} \vM(\vu_{t+1}, \dots, \vu_{t+\tau}) \vN (\vu_{s+1}, \dots, \vu_t) \right)}_{:= f_s(\vu_{s+1}, \dots, \vu_{T})} (\vB \vu_{s} + \vw_{s}) \nonumber \\
        & \overset{(c)}{=} \sum_{s=0}^{T-\tau} f_s(\vu_{s+1}, \dots, \vu_{T}) (\vB \vu_{s} + \vw_{s}),
    \end{align}
    where we introduce in $(c)$ the functions $f_s(\cdot)$. Now we clearly see that $\sum_{t=\tau}^T (\theta^\top\phi_t) (\lambda^\top z_t)$, written in the form \eqref{eq:martingale bias} is a martingale difference. Specifically, the right side of \eqref{eq:martingale bias} is a sum of martingale differences when the indices are revealed in reverse order. Before we use this fact, let us note that using Lemma \ref{lem:upper bound stability} we can show that the terms involving the functions $f_s(\cdot)$ are well bounded. More specifically, we have 
    
    \begin{align*}
        \Vert f(\vu_{s+1}, \dots, \vu_{T}) \Vert_{\ell_2} & =  \left\Vert \sum_{t = s}^{T-\tau} \vM(\vu_{t+1}, \dots, \vu_{t+\tau}) \vN (\vu_{s+1}, \dots, \vu_t)  \right\Vert_{\ell_2} \\
        & \le  \sum_{t = s}^{T-\tau} \left\Vert \phi_{t+\tau} \right\Vert_{\ell_2} \left \Vert \vC\right\Vert_\op \left\Vert \prod_{\ell = 1}^{\tau-1} (u_{t + \tau - \ell}\circ \Ab )  \prod_{k = 1}^{t-s} (u_{ t+1-k}\circ \Ab ) \right\Vert_{\op}  \\
        & \le \sup_{t \le T}\Vert \phi_t\Vert_{\ell_2}\Vert \vC \Vert_\op \kappa  \rho^{\tau-1} \sum_{i = 0}^{T-\tau-s}  \rho^{i} \\
        &  \le  \frac{\sqrt{2(d_u+1)^{\tau-1}}\Vert \vC \Vert_\op \kappa\rho^{\tau-1}}{1-\rho}
    \end{align*}
    where we bound $\Vert \phi_t \Vert_{\ell_2} \le \sqrt{d_{\phi}} \leq \sqrt{2(d_u+1)^{\tau-1}}$. Now, we also remark that $\vB \vu_s + \vw_s$ is zero-mean and $(\Vert \vB\Vert_\op^2 + \sigma^2)$-sub-Gaussian. Hence, using Freedman's inequality (see Lemma \ref{lem:freedman}), we obtain: for all $\delta \in (0,1)$, the event:
    
    \begin{align*}
        & \left\vert \sum_{s=0}^{T-\tau}f_s(\vu_{s+1}, \dots, \vu_{T}) (\vB \vu_{s} + \vw_{s}) \right\vert \\
        & \qquad \qquad \qquad \qquad >  \frac{\Vert \vC  \Vert_\op \kappa  \rho^{\tau-1} \sqrt{\Vert \vB \Vert_\op^2 + \sigma^2}  \sqrt{2 (T-\tau+1) (d_u+1)^{\tau-1}\log(2 \cdot9^{d_{\phi} + d_y }/\delta)} }{1-\rho} 
    \end{align*}
    with probability at most $\delta/9^{d_{\phi} + d_y }$. The conclusion follows immediately by recalling the inequality \eqref{eqn:eq net argument}. This completes the proof.
\end{proof}

\subsection{Upper Bounding Approximation Error for \eqref{eqn:LDS-BO}}
In this subsection, we will upper bound the term $\opn{\sum_{t=\tau}^T\phi_t z_t^\T}$ for learning bilinearly observed linear dynamical systems, governed by state-space equation \eqref{eqn:LDS-BO}. Recall that, in the case of \eqref{eqn:LDS-BO}, the truncation bias $z_t$ takes the following form:
\begin{equation}
\begin{aligned} \label{eqn:z_t_LDS_BO_app}
   z_t  = (u_{t} \circ C)  A^{\tau-1} x_{t-\tau+1}
\end{aligned}
\end{equation}

\begin{proposition}[Truncation bias for \eqref{eqn:LDS-BO}]\label{prop:truncation_bias_lds_bo_app}
    Consider the system~\eqref{eqn:LDS-BO}. Let $\phi_t$ and $z_t$ be as in \eqref{eqn:phi_t_LDS_BO_app} and \eqref{eqn:z_t_LDS_BO_app}, respectively. Suppose Assumptions \ref{assump:stability_linear} and \ref{assump:input+noise} hold. Let $\delta \in (0,1)$ and $T \ge \tau-1$. Then, the event:
        \begin{align*}
        \left\Vert \sum_{t=\tau}^T \phi_t z_t^\top \right\Vert_\op \le  \frac{c_2 \kappa \rho^{\tau-1} }{1-\rho}\sqrt{(T-\tau+1)d_u(d_u+1)(\tau-1)\left( \log\left(\frac{e}{\delta}\right) + d_y + d_u(d_u+1)(\tau-1) \right)}      
    \end{align*}
    holds with probability at least $1-\delta$, with $c_2 = \beta_C \sqrt{2\log(9)(\Vert \vB \Vert_\op^2 + \sigma^2)}$, where we define $\beta_C := \sup_{t \leq T}\opn{u_t \circ C}$. 
\end{proposition}
\begin{proof}
Using a similar argument as the proof of Proposition~\ref{prop:truncation_bias_bds_po_app}, for all $r > 0$: 
\begin{align}\label{eqn:eq net argument}
    \PP\left(    \left \Vert \sum_{t=\tau}^T \phi_t  z_t^\top   \right\Vert_\op  > 2 r \right) \le 9^{d_{\phi}+ d_y}\max_{\theta \in \Mcal, \lambda \in \Ncal} \PP\left(  \sum_{t=\tau}^T (\theta^\top\phi_t) (\lambda^\top z_t) > r \right).
\end{align}
It remains to bound $\sum_{t=\tau}^T (\theta^\top\phi_t) (\lambda^\top z_t)$ with high probability uniformly over the unit spheres. Let $\theta \in \cS^{d_{\phi} - 1}$, and $\lambda \in \cS^{d_y - 1}$. We have: 
    \begin{align*}
        & \sum_{t=\tau}^T (\theta^\top\phi_t) (\lambda^\top z_t) \\
        & \qquad =  \sum_{t=\tau}^T (\theta^\top\phi_t) (\lambda^\top  (u_{t} \circ C))  A^{\tau-1}  \vx_{t-\tau+1}, \\
        & \qquad = \sum_{t=1}^{T-\tau+1} (\theta^\top\phi_{t+\tau-1}) (\lambda^\top  (u_{t+\tau-1} \circ C)) A^{\tau-1}  \vx_{t}, \\
         & \qquad = \sum_{t=1}^{T-\tau+1} \sum_{s = 0}^{t-1}  (\theta^\top\phi_{t+\tau-1}) (\lambda^\top  (u_{t+\tau-1} \circ C)) A^{\tau-1} A^{t-s-1} (\vB \vu_{s} + \vw_{s}), 
         \\ & \qquad = \sum_{t=0}^{T-\tau} \sum_{s = 0}^{t}  \underbrace{(\theta^\top\phi_{t+\tau}) (\lambda^\top  (u_{t+\tau} \circ C)) A^{t+\tau-s-1}}_{:= M(\vu_{t+1}, \dots, \vu_{t+\tau})}   (\vB \vu_{s} + \vw_{s}), \nn \\
         & \qquad = \sum_{s=0}^{T-\tau} \underbrace{\left( \sum_{t = s}^{T-\tau} \vM(\vu_{t+1}, \dots, \vu_{t+\tau}) \right)}_{:= f_s(\vu_{s+1}, \dots, \vu_{T})} (\vB \vu_{s} + \vw_{s}) \\
        & \qquad = \sum_{s=0}^{T-\tau} f_s(\vu_{s+1}, \dots, \vu_{T}) (\vB \vu_{s} + \vw_{s}),
    \end{align*}
   The rest of the proof follows similar to the proof of Proposition~\ref{prop:truncation_bias_bds_po_app}, with the upper bound on $\Vert f(\vu_{s+1}, \dots, \vu_{T}) \Vert_{\ell_2}$ derived as follows:
    \begin{align*}
        \Vert f(\vu_{s+1}, \dots, \vu_{T}) \Vert_{\ell_2} 
        & \le  \sum_{t = s}^{T-\tau} \left\Vert \phi_{t+\tau} \right\Vert_{\ell_2} \left \Vert u_{t+\tau} \circ C \right\Vert_\op \opn{A^{t+\tau-s-1}}  \\
        & \le \sup_{t \le T}\Vert \phi_t\Vert_{\ell_2} \sup_{t \le T}\Vert u_t \circ \vC \Vert_\op \kappa  \rho^{\tau-1} \sum_{i = 0}^{T-\tau-s}  \rho^{i} \\
        &  \le  \frac{\sqrt{d_u(d_u+1)(\tau-1)}\beta_C \kappa\rho^{\tau-1}}{1-\rho},
    \end{align*}
    where $\beta_C := \sup_{t \le T}\Vert u_t \circ \vC \Vert_\op$.
\end{proof} \section{Proofs of Self-Normalized Martingale Bounds}\label{app:error noise}
In this Appendix, we present the proof of our main results on upper bounding the estimation error term due to noise processes, which is denoted by $E_1$ in \eqref{eqn:estimation_error}. Specifically, we will upper bound $E_1$ by controlling the self-normalized martingale term $\opn{ \left( \sum_{t=1}^T  \eta_t \phi_t ^\top \right) \left(\sum_{t=1}^{T} \phi_{t} \phi_{t}^\top \right)^{-1/2}}$ for the dynamical systems we studied in Section~\ref{sec:sysid}.

\subsection{Upper Bounding Self-Normalized Noise Term for \eqref{eqn:BDS-PO}}
In this subsection, we will upper bound the term $\opn{ \left( \sum_{t=\tau}^T  \eta_t \phi_t ^\top \right) \left(\sum_{t=\tau}^{T} \phi_{t} \phi_{t}^\top \right)^{-1/2}}$ for learning partially observed bilinear dynamical systems, governed by state-space equation \eqref{eqn:BDS-PO}. Recall that, in the case of \eqref{eqn:BDS-PO}, the noise process $\eta_t$ takes the following form:
\begin{equation}
\begin{aligned} \label{eqn:eta_t_BDS_PO_app}
   \eta_t  &= F_t \omega_t  = v_t + \sum_{j = 1}^{\tau-1} (F_{t})_j w_{t-j}, \quad \quad \quad  (F_{t})_j := C \prod_{i=1}^{j-1}(u_{t-i}\circ A), 
\end{aligned}
\end{equation}
where $F_t  = \begin{bmatrix}
     I_{d_y} &
    C   &
    C \,  (u_{t-1}\circ A )  &
    \cdots &
    C \prod_{\ell = 1}^{\tau-2} (u_{t-\ell}\circ A ) 
\end{bmatrix}$ and $\omega_t = \begin{bmatrix} v_t^\T & w_{t-1}^\T & \cdots & 
w_{t-\tau+1}^\T \end{bmatrix}^\T$.

\begin{proposition}[Self-Normalized Term for \eqref{eqn:BDS-PO}]\label{prop:self_normalized_bds_po_app}
    Consider the system~\eqref{eqn:BDS-PO}. Let $\phi_t$ and $\eta_t$ be as in \eqref{eqn:phi_t_BDS_PO_app} and \eqref{eqn:eta_t_BDS_PO_app}, respectively. Suppose Assumptions~\ref{assump:input+noise} and \ref{assump:stability} hold. Let $\delta \in (0,1)$ and $T \ge \tau-1$. Then, the event:
\begin{align}
    \opn{ \left( \sum_{t=\tau}^T  \eta_t \phi_t ^\top \right) \left(\sum_{t=\tau}^{T} \phi_{t} \phi_{t}^\top\right)^{-1/2}} 
    &\leq  c_3 \sigma  \sqrt{(d_u+1)^{\tau-1} \log(2) + 2d_y\log(5) + 2\log\left(\frac{\tau}{\delta}\right)}, 
\end{align}
    holds with probability at least $1-\delta$, with $c_3 = 2\sqrt{2}\left(1 + {\kappa \opn{C}}/{(1-\rho)}\right)$.
\end{proposition}
\begin{proof}
To begin, we define,
\begin{align}
V_T := \sum_{t=\tau}^{T} \phi_{t} \phi_{t}^\top , \quad \text{and} \quad \bar V_T := \sum_{t=\tau}^{T} \phi_{t} \phi_{t}^\top + \lambda I_{d_\phi}. \label{eqn:Vbar_V_def_BDS_PO_app}
\end{align}
Combining this with \eqref{eqn:eta_t_BDS_PO_app}, we can write the self-normalized noise term in the regularized case as follows:
\begin{align}
    \opn{ \left( \sum_{t=\tau}^T  \eta_t \phi_t ^\top \right) \bar V_T^{-1/2}} &= \opn{ \left( \sum_{t=\tau}^T  \left( v_t + \sum_{j = 1}^{\tau-1} (F_{t})_j w_{t-j}\right) \phi_t ^\top \right) \bar V_T^{-1/2}}, \nn \\
    & \leq \opn{ \left( \sum_{t=\tau}^T  v_t \phi_t ^\top \right) \bar V_T^{-1/2}} + \sum_{j=1}^{\tau-1}  \opn{ \left( \sum_{t=\tau}^T  (F_{t})_j w_{t-j} \phi_t ^\top \right) \bar V_T^{-1/2}}. \label{eqn:split_noise_BDS_PO_app}
\end{align}
In the following, we will upper bound each term in \eqref{eqn:split_noise_BDS_PO_app} separately. To proceed, note that conditioned on the input sequence $\Ucal_T:=\sigma(u_0,\ldots,u_T)$, the matrices $V_T, \bar V_T, \{(F_{t})_j\}$, and the covariates $\phi_t$ are deterministic for all $t \in [\tau,T]$ and $j \in[1,\tau-1]$. Moreover, from Assumption~\ref{assump:stability} we have
\begin{align}
    \opn{(F_{t})_j}  = \opn{C \prod_{i=1}^{j-1}(u_{t-i}\circ A)} \leq \opn{C} \kappa \rho^{j-1}.
\end{align}
Hence conditioned on the input sequence $\Ucal_T:=\sigma(u_0,\ldots,u_T)$, the noise process $(F_{t})_j w_{t-j}$ is $\sigma^2 \opn{C}^2 \kappa^2 \rho^{2(j-1)}$-sub-Gaussian random vector, whereas, $v_t$ is $\sigma^2$-sub-Gaussian random vector. For any $j \in [0, \tau-1]$, let $\Fcal_t^{(j)}:=\sigma(u_0,\ldots,u_T, v_0,\ldots,v_t, w_0, \ldots w_{t-j})$ denote the filtration generated by the entire input sequence, and noise processes when $t \geq \tau$. Furthermore, let $\Fcal_{\tau-1}^{(j)}:=\sigma(u_0,\ldots,u_T)$. Note that, $v_t$ is $\Fcal_t^{(0)}$-measurable, and conditioned on $\cF_{t-1}^{(0)}$, it is $\sigma^2$-sub-Gaussian. Similarly, $(F_{t})_j w_{t-j}$ is $\Fcal_t^{(j)}$-measurable, and conditioned on $\cF_{t-1}^{(j)}$, it is $\sigma^2 \opn{C}^2 \kappa^2 \rho^{2(j-1)}$-sub-Gaussian. Hence applying Proposition~\ref{prop:self-normalized vector martingale} to each term in \eqref{eqn:split_noise_BDS_PO_app}, with probability at least $1-\delta/\tau$, we have
\begin{equation}
\begin{aligned}\label{eqn:apply_self_mart_prop_BDS_PO_v1}
    \opn{ \left( \sum_{t=\tau}^T  v_t \phi_t ^\top \right) \left(V_T + \bar V_T\right)^{-1/2}}^2 &\leq 4\sigma^2\log \left(\frac{\det(V_T+\bar V_T)}{\det (\bar V_T)}\right) + 8\sigma^2 \log\left(\frac{\tau 5^{d_y}}{\delta}\right), \\
    \opn{ \left( \sum_{t=\tau}^T  (F_{t})_j w_{t-j} \phi_t ^\top \right) \left(V_T + \bar V_T\right)^{-1/2}}^2 &\leq 4\sigma^2 c_j^2\log\left(\frac{\det(V_T+\bar V_T)}{\det (\bar V_T)}\right) + 8\sigma^2 c_j^2\log\left(\frac{\tau 5^{d_y}}{\delta}\right),
\end{aligned}
\end{equation}
for all $j \in [1,\tau-1]$, where we set $c_j := \opn{C} \kappa \rho^{j-1}$. To proceed, since $V_T \preceq \bar V_T$, we have
\begin{equation}
\begin{aligned}
&V_T+\bar V_T\preceq2\bar V_T, \qquad\qquad\qquad \bar V_T^{-1} \preceq 2(V_T+\bar V_T)^{-1}, \\
\text{and} \quad &\log \left(\frac{\det(V_T+\bar V_T)}{\det \bar V_T}\right) = 
\log \left(\det\left(I_{d_\phi}+\bar V_T^{-1/2}V_T\bar V_T^{-1/2}\right) \right)
\le d_{\phi}\log2.
\end{aligned}\label{eqn:comparison_BDS_PO}
\end{equation}
Hence, on the event in \eqref{eqn:apply_self_mart_prop_BDS_PO_v1}, with probability at least $1-\delta/\tau$, we have
\begin{equation}
\begin{aligned}\label{eqn:apply_self_mart_prop_BDS_PO_v1}
    \opn{ \left( \sum_{t=\tau}^T  (F_{t})_j w_{t-j} \phi_t ^\top \right) \bar V_T^{-1/2}}^2 & \leq 2\opn{ \left( \sum_{t=\tau}^T  (F_{t})_j w_{t-j} \phi_t ^\top \right) \left(V_T + \bar V_T\right)^{-1/2}}^2, \\
    &\leq 8\sigma^2 c_j^2 \left( d_{\phi} \log(2) + 2\log\left(\frac{\tau 5^{d_y}}{\delta}\right) \right),
\end{aligned}
\end{equation}
for all $j \in [1,\tau-1]$. Similar bound holds for $ \opn{ \left( \sum_{t=\tau}^T  v_t \phi_t ^\top \right) \bar V_T^{-1/2}}^2$ with $c_j$ replaced by $1$. Setting $c_0 = 1$, and $c_j = \opn{C} \kappa \rho^{j-1}$, we union bound over $\tau$ events~\eqref{eqn:apply_self_mart_prop_BDS_PO_v1} to get the following:
\begin{align}
    \opn{ \left( \sum_{t=\tau}^T  \eta_t \phi_t ^\top \right) \bar V_T^{-1/2}} 
    & \leq \opn{ \left( \sum_{t=\tau}^T  v_t \phi_t ^\top \right) \bar V_T^{-1/2}} + \sum_{j=1}^{\tau-1}  \opn{ \left( \sum_{t=\tau}^T  (F_{t})_j w_{t-j} \phi_t ^\top \right) \bar V_T^{-1/2}}, \nn \\
    &\leq \sum_{j=0}^{\tau-1} 2\sqrt{2}\sigma c_j \sqrt{d_{\phi} \log(2) + 2\log\left(\frac{\tau 5^{d_y}}{\delta}\right)}, \nn \\
    &\leq 2\sqrt{2} \sigma \left(1 + \opn{C} \kappa\sum_{j=1}^{\tau-1} \rho^{j-1} \right) \sqrt{d_{\phi} \log(2) + 2\log\left(\frac{\tau 5^{d_y}}{\delta}\right) }, \nn \\
    &\leq 2\sqrt{2} \sigma \left(1 + \frac{\opn{C} \kappa}{1-\rho} \right) \sqrt{d_{\phi} \log(2) + 2d_y\log(5) + 2\log\left(\frac{\tau}{\delta}\right) },
    \label{eqn:after_union_bound_BDS_Po}
\end{align}
with probability at least $1-\delta$. In the un-regularized case (i.e., $\lambda = 0$ in \eqref{eqn:Vbar_V_def_BDS_PO_app}), given the persistence of excitation holds, i.e., $V_T\succ0$, we have $\bar V_T=V_T$ and $V_T+V_T=2V_T$.
The two comparisons in \eqref{eqn:comparison_BDS_PO} are then equalities. Thus the same proof directly gives
\begin{align}
    \opn{ \left( \sum_{t=\tau}^T  \eta_t \phi_t ^\top \right) V_T^{-1/2}} 
    &\leq 2\sqrt{2} \sigma \left(1 + \frac{\opn{C} \kappa}{1-\rho} \right) \sqrt{d_{\phi} \log(2) + 2d_y\log(5) + 2\log\left(\frac{\tau}{\delta}\right) },
    \label{eqn:final_unregularized_bound_BDS_Po}
\end{align}
with probability at least $1-\delta$. This completes the proof.
\end{proof}

\subsection{Upper Bounding Self-Normalized Noise Term for \eqref{eqn:LDS-BO}}
In this subsection, we will upper bound the term $\opn{ \left( \sum_{t=\tau}^T  \eta_t \phi_t ^\top \right) \left(\sum_{t=\tau}^{T} \phi_{t} \phi_{t}^\top \right)^{-1/2}}$ for learning linear dynamical systems with bilinear observations, governed by state-space equation \eqref{eqn:LDS-BO}. Recall that, in the case of \eqref{eqn:LDS-BO}, the noise process $\eta_t$ takes the following form:
\begin{equation}
\begin{aligned} \label{eqn:eta_t_LDS_BO_app}
   \eta_t  &= H_t \omega_t  = v_t + \sum_{j = 1}^{\tau-1} (H_{t})_j w_{t-j}, \quad \quad \quad  (H_{t})_j := (u_t \circ C) A^{j-1}, 
\end{aligned}
\end{equation}
where $H_t  = \begin{bmatrix}
     I_{d_y} &
    (u_t \circ C)   &
    (u_t \circ C) A  &
    \cdots &
    (u_t \circ C) A^{\tau-2} 
\end{bmatrix}$ and $\omega_t = \begin{bmatrix} v_t^\T & w_{t-1}^\T & \cdots & 
w_{t-\tau+1}^\T \end{bmatrix}^\T$.

\begin{proposition}[Self-Normalized Term for \eqref{eqn:LDS-BO}]\label{prop:self_normalized_lds_bo_app}
    Consider the system~\eqref{eqn:LDS-BO}. Let $\phi_t$ and $\eta_t$ be as in \eqref{eqn:phi_t_LDS_BO_app} and \eqref{eqn:eta_t_LDS_BO_app}, respectively. Suppose Assumptions~\ref{assump:input+noise} and \ref{assump:stability_linear} hold. Let $\delta \in (0,1)$ and $T \ge \tau-1$. Then, the event:
\begin{align}
    \opn{ \left( \sum_{t=\tau}^T  \eta_t \phi_t ^\top \right) \left(\sum_{t=\tau}^{T} \phi_{t} \phi_{t}^\top\right)^{-1/2}} 
    &\leq c_4\sigma  \sqrt{d_u(d_u+1)(\tau-1) \log(2) + 2d_y\log(5) + 2\log\left(\frac{\tau}{\delta}\right)}, 
\end{align}
    holds with probability at least $1-\delta$, with $c_4 = 2\sqrt{2}\left(1 + {\kappa \beta_C}/{(1-\rho)}\right)$, where $\beta_C := \sup_{t \le T}\Vert u_t \circ \vC \Vert_\op$.
\end{proposition}
Proposition \ref{prop:self_normalized_lds_bo_app} follows from the same proof techniques as we used to prove Proposition~\ref{prop:self_normalized_bds_po_app}, with $(H_{t})_j w_{t-j} = (u_t \circ C) A^{j-1} w_{t-j}$ being conditionally sub-Gaussian with variance proxy $\sigma^2\beta_C^2 \kappa^2 \rho^{2(j-1)}$. We therefore, omit the proof of Proposition \ref{prop:self_normalized_lds_bo_app} here. Once can directly get the statement of Proposition~\ref{prop:self_normalized_lds_bo_app} by following the proof of Proposition~\ref{prop:persistence_BDS_PO_app} with $(F_{t})_j$ replaced by $(H_t)_j$.

\subsection{Upper Bounding Self-Normalized Noise Term for \eqref{eqn:BDS-FO}}
In this subsection, we will upper bound the term $\opn{ \left( \sum_{t=1}^T  \eta_t \phi_t ^\top \right) \left(\sum_{t=1}^{T} \phi_{t} \phi_{t}^\top \right)^{-1/2}}$ for learning fully observed bilinear dynamical systems, governed by state-space equation \eqref{eqn:BDS-FO}. Recall that, in the case of \eqref{eqn:BDS-FO}, the noise process $\eta_t$  and the covariate $\phi_t$ has the following form:
\begin{equation}
\begin{aligned} \label{eqn:eta_t_BDS_FO_app}
   \eta_t  & = w_t , \quad \quad \quad  \phi_t =  \begin{bmatrix} u_{t} \\
    \util_{t} \otimes x_{t}\end{bmatrix} = \begin{bmatrix} u_{t} \\ x_t \\
    u_t \otimes x_{t}\end{bmatrix}.
\end{aligned}
\end{equation}
Before we state our main result, recall that the covariance of $\phi_t$ in the case of \eqref{eqn:BDS-FO} can be written as:
\begin{align}
    \Gamma_{\phi}^{(t)} = \E \left[ \phi_{t}^{\bds}\phi_{t}^{\bds\T}\right] &= \E \left[ \begin{bmatrix} u_{t} \\ x_{t} \\
    u_{t} \otimes x_{t}\end{bmatrix} \begin{bmatrix} u_{t}^\T & x_{t}^\T &
    u_{t}^\T \otimes x_{t}^\T \end{bmatrix}\right] =  \begin{bmatrix}
        I_{d_u} & 0 & 0 \\
        0 & \Gamma_x^{(t)} & 0 \\
        0 & 0 & I_{d_u} \otimes \Gamma_x^{(t)}
    \end{bmatrix}, \label{eqn:Gamma_phi_def_last}
\end{align}
where $\Gamma_x^{(t)} = \E[x_t x_t^\T]  = \sum_{\ell = 1}^{t} \breve{G}_\ell \breve{G}_\ell^\top + \sigma^2 \sum_{\ell = 1}^{t} \breve{F}_\ell \breve{F}_\ell^\top$. We are now ready to state our main result on upper bounding the self-normalized error for estimating \eqref{eqn:BDS-FO}.
\begin{proposition}[Self-Normalized Term for \eqref{eqn:BDS-PO}]\label{prop:self_normalized_bds_fo_app}
    Consider the system~\eqref{eqn:BDS-FO}. Let $\phi_t$ and $\eta_t$ be as in \eqref{eqn:eta_t_BDS_FO_app}. Suppose Assumptions~\ref{assump:input+noise}, \ref{assump:stability} hold and $\E[w_t w_t^\T]  =\sigma^2 I_{d_x}$. Let $\delta \in (0,1)$ and $T \ge 1$. Then, the event:
\begin{align}
    \opn{ \left( \sum_{t=\tau}^T  w_t \phi_t ^\top \right) \left(\sum_{t=\tau}^{T} \phi_{t} \phi_{t}^\top\right)^{-1/2}} 
    &\leq   2\sqrt{2} \sigma \bigg( d_u(d_u+1) d_x \log \left(25 \frac{d_u(d_u+1) d_x}{\delta}\right) \nn \\
    & \quad+ \log \left( \det\left({\Gamma_\phi^{(T)}}/{\lambda_{\min}\left(\Gamma_{\phi}^{(\tau-1)}\right)}\right)\right) \nn\\
    &\quad + 2d_x \log(5) + 2\log(3/\delta) \bigg)^{1/2}, 
\end{align}
    holds with probability at least $1-\delta$. 
\end{proposition}
\begin{proof}
To begin, we define,
\begin{align}
V_T := \sum_{t=1}^{T} \phi_{t} \phi_{t}^\top , \quad \text{and} \quad \bar V_T := \sum_{t=\tau}^{T} \phi_{t} \phi_{t}^\top + \lambda I_{d_\phi}. \label{eqn:Vbar_V_def_BDS_PO_app}
\end{align}
Let $\Fcal_t:=\sigma(u_0,\ldots,u_T, w_t, \ldots w_{t})$ denote the filtration generated by the entire input sequence, and noise process when $t \geq 0$. Furthermore, let $\Fcal_{-1}:=\sigma(u_0,\ldots,u_T)$. Then, $x_t $ is $\Fcal_{t-1}$-measurable, and as a result $\phi_t $ is also $\Fcal_{t-1}$-measurable. Moreover, $w_t$ is $\Fcal_t$-measurable, and conditioned on $\cF_{t-1}$, it is $\sigma^2$-sub-Gaussian. Hence applying Proposition, with probability at least $1-\delta/3$, we have
\begin{equation}
\begin{aligned}\label{eqn:apply_self_mart_prop_BDS_FO_v1}
    \opn{ \left( \sum_{t=1}^T  w_t \phi_t ^\top \right) \bar V_T^{-1/2}}^2 &\leq 4\sigma^2\log \left(\frac{\det(\bar V_T)}{\det (\lambda I_{d_\phi})}\right) + 8\sigma^2 \log\left(\frac{3 \cdot 5^{d_x}}{\delta}\right).
\end{aligned}
\end{equation}
Next, under the event that persistency of excitation holds, we have $V_T \succeq T \lambda_{\min}\left(\Gamma_{\phi}^{(\tau-1)}\right)/8$. Then, setting $\lambda = T \lambda_{\min}\left(\Gamma_{\phi}^{(\tau-1)}\right)/8$ guarantees that
\begin{align}
    V_T^{-1} \preceq 2 (V_T + \lambda I_{d_\phi} )^{-1}.
\end{align}
Hence, under the event that persistency of excitation holds, with probability at least $1-\delta/3$, we have

\begin{equation}
\begin{aligned}\label{eqn:apply_self_mart_prop_BDS_FO_v2}
     \opn{ \left( \sum_{t=1}^T  w_t \phi_t ^\top \right) V_T^{-1/2}}^2 & \leq 2\opn{ \left( \sum_{t=1}^T  w_t \phi_t ^\top \right) \bar V_T^{-1/2}}^2, \\
    &\leq 8\sigma^2\log \left(\frac{\det(\bar V_T)}{\det (\lambda I_{d_\phi})}\right) + 16\sigma^2 \log\left(\frac{3 \cdot 5^{d_x}}{\delta}\right).
\end{aligned}
\end{equation}
It remains to upper bound the log-determinant term. For this purpose, we use the same arguments as used in the proof of Theorem 5.1 in \cite{ziemann2023tutorial}. First, we use Lemma 5.2 in \cite{ziemann2023tutorial} (matrix Markov's inequality) to obtain:
\begin{align}
    V_T \preceq \frac{3d_{\phi}}{\delta} \sum_{t=1}^T \E[\phi_t \phi_t^\T] = \frac{3d_u(d_u+1) d_x}{\delta} \sum_{t=1}^T \Gamma_\phi^{(t)} \preceq \frac{3d_u(d_u+1) d_x}{\delta} T \Gamma_\phi^{(T)},
\end{align}
with probability at least $1-\delta/3$. To proceed, following the proof of Theorem 5.1 in \cite{ziemann2023tutorial}, when persistence of excitation holds, we have
\begin{align}
    \frac{\det(V_T + \lambda I_{d_\phi})}{\det (\lambda I_{d_\phi})} &\leq \det \left( I_{d_u(d_u+1) d_x} + 24 \frac{d_u(d_u+1) d_x}{\delta \lambda_{\min}\left(\Gamma_{\phi}^{(\tau-1)}\right)}  \Gamma_\phi^{(T)} \right), \nn \\
    &\leq \left(25 \frac{d_u(d_u+1) d_x}{\delta}\right)^{d_u(d_u+1) d_x} \det\left(\Gamma_\phi^{(T)}/\lambda_{\min}\left(\Gamma_{\phi}^{(\tau-1)}\right)\right).
\end{align}
Plugging this back into \eqref{eqn:apply_self_mart_prop_BDS_FO_v2}, and union bounding, we get:
\begin{equation}
\begin{aligned}\label{eqn:apply_self_mart_prop_BDS_FO_final}
     &\opn{ \left( \sum_{t=1}^T  w_t \phi_t ^\top \right) V_T^{-1/2}}^2 \\ 
    &\qquad\leq 8\sigma^2\log \left(\frac{\det(\bar V_T)}{\det (\lambda I_{d_\phi})}\right) + 16\sigma^2 \log\left(\frac{3 \cdot 5^{d_x}}{\delta}\right), \\
    &\qquad\leq 8 \sigma^2 \left( d_u(d_u+1) d_x \log \left(25 \frac{d_u(d_u+1) d_x}{\delta}\right)  + \log \left( \det\left(\frac{\Gamma_\phi^{(T)}}{\lambda_{\min}\left(\Gamma_{\phi}^{(\tau-1)}\right)}\right)\right)\right) \\
    &\qquad + 16 \sigma^2 \left( d_x \log(5) + \log(3/\delta) \right),
\end{aligned}
\end{equation}
with probability at least $1-\delta$. Taking the square root on both sides gives us the statement of the Proposition~\ref{prop:self_normalized_bds_fo_app}. This completes the proof.

\end{proof} \section{Proofs of Main Results in Section~\ref{sec:sysid}}\label{app sysid proofs}

\subsection{Proof of Theorem~\ref{thm:main_LDS_BO}}\label{app:LDS-BO}
Theorem~\ref{thm:main_LDS_BO} follows from combining Propositions~\ref{prop:persistence_LDS_BO_app}, \ref{prop:truncation_bias_lds_bo_app}, and \ref{prop:self_normalized_lds_bo_app}. Specifically, define the events:
\begin{align*}
    \Ecal_1 &:= \Bigg\{ \lambda_{\min}\left( \sum_{t=\tau}^{T} \phi_t^\lds \phi_t^{\lds\T} \right) \geq (T-\tau+1)/4 \Bigg\}, \\
    \Ecal_2 &:=  \Bigg\{ \left\Vert \sum_{t=\tau}^T \phi_t z_t^\top \right\Vert_\op \\
    &\le  \frac{c_2 \kappa \rho^{\tau-1} }{1-\rho}\sqrt{(T-\tau+1)d_u(d_u+1)(\tau-1)\left( \log\left(\frac{e}{\delta}\right) + d_y + d_u(d_u+1)(\tau-1) \right)} \Bigg\}, \\
    \Ecal_3 &:= \Bigg \{ \opn{ \left( \sum_{t=\tau}^T  \eta_t \phi_t ^\top \right) \left(\sum_{t=\tau}^{T} \phi_{t} \phi_{t}^\top\right)^{-1/2}} \\ 
    &\leq c_4\sigma  \sqrt{d_u(d_u+1)(\tau-1) \log(2) + 2d_y\log(5) + 2\log\left(\frac{\tau}{\delta}\right)}, \Bigg\}
\end{align*}
Recalling the estimation error decompositions \eqref{eqn:estimation_error} and \eqref{eqn:estimation_error_decomposition}, we see that when the event $\cE_1 \cap \cE_2 \cap \cE_3$ holds, then the upper bound on the estimation error presented in Theorem \ref{thm:main_LDS_BO} follows. Now, we remark that by union bound, we have 
\begin{align*}
    \PP(\cE_1 \cap \cE_2 \cap \cE_3) = 1 - \PP(\cE_1^c \cup \cE_2^c \cup \cE_3^c) \ge 1 -  \PP(\cE_1^c) - \PP ( \cE_2^c) - \PP(\cE_3^c). 
\end{align*}
Thus, using Propositions~\ref{prop:persistence_LDS_BO_app}, \ref{prop:truncation_bias_lds_bo_app}, and \ref{prop:self_normalized_lds_bo_app}, we obtain  $ \PP(\cE_1 \cap \cE_2 \cap \cE_3) \le 3 \delta$, provided the conditions in Theorem \ref{thm:main_LDS_BO} hold. This concludes the proof.

\subsection{Proof of Theorem~\ref{thm:main_BDS_FO}}\label{app:BDS-FO}
Theorem~\ref{thm:main_BDS_FO} follows from combining Propositions~\ref{prop:persistence_BDS_FO_app} and \ref{prop:self_normalized_bds_fo_app}.

\subsection{Proof of Theorem~\ref{thm:main_BDS_PO}}\label{app:BDS-PO}
Theorem~\ref{thm:main_BDS_PO} follows from combining Proposition~\ref{prop:persistence_BDS_PO_app}, \ref{prop:truncation_bias_bds_po_app}, and \ref{prop:self_normalized_bds_po_app}.
 \section{Supporting Lemmas}\label{app:supporting}
In this Appendix, we provide a list of auxiliary lemmas that will be useful to derive our main results.

\subsection{Stability of Bilinear Dynamical Systems} \label{app:stability of BDS}

In this Appendix, we present some results that concerns the stability of bilinear system in the sense of Definition \ref{def:stability}. In Lemma \ref{lem:stability notion}, we present a condition under which a bilinear system satisfies uniform stability. In Lemma \ref{lem:upper bound stability}, we present a bound on $\Vert F_t \Vert_\op$ which follows under our stability assumption.

\begin{lemma}[Stability in BDS]\label{lem:stability notion}
    Let $\cU$ be a bounded and non-empty subset of $\RR^{d_u}$ such that $\rho(\cU \circ A) < 1$, then for all $ \rho \in (\rho(\cU \circ A), 1)$, there exists $\kappa \ge 1$ such that the  system is $(\cU , \kappa, \rho)$-uniformly stable.
\end{lemma}

\begin{proof}
    Let $\rho \in (\rho(\cU \circ \vA), 1)$. We recall that the joint spectral radius of $\cU \circ \vA$ is defined as follows: 
    \begin{align}
        \rho(\cU \circ \vA ) = \lim_{k \to \infty} \sup_{\vM_1, \dots, \vM_k \in \cU \circ \vA} \Vert \vM_1 \vM_2 \cdots \vM_k \Vert_\op^{1/k}. 
    \end{align}
    Since by assumption the limit exists, we have by definition that, 
    \begin{align*}
         \forall \; \epsilon > 0, \; \exists \; k_0 \ge 1: \;  \forall \; k \ge k_0, \quad  \left\vert \sup_{\vM_1, \dots, \vM_k \in \cU \circ \vA }\Vert \vM_1 \vM_2 \cdots \vM_k \Vert_\op^{1/k} - \rho(\cU \circ \vA) \right\vert < \epsilon \rho(\cU \circ \vA).  
ˇ    \end{align*}
    Thus choosing $\epsilon > 0$, such that $\rho \ge  (1+\epsilon) \rho(\cU \circ \vA)$, we can find $k_0$ such that, \begin{align*}
        \forall \; k \ge k_0, \qquad \sup_{\vM_1, \dots, \vM_k \in \cU \circ \vA }\Vert \vM_1 \vM_2 \cdots \vM_k \Vert_\op <  \rho^k.
    \end{align*}
    Now, we can further define 
    \begin{align}
        \kappa = \max\left\lbrace 1, \sup_{1 \le k \le k_0}\sup_{ \vM_1, \dots, \vM_k \in \cU \circ \vA }\frac{\Vert \vM_1 \vM_2 \cdots \vM_k \Vert_\op^{1/k}}{\rho}  \right\rbrace.
    \end{align}
    We note that $\kappa$ is well defined because $\cU \circ \vA$ is a bounded set of matrices since $\cU$ is bounded. Indeed, for all $\vM \in \cU \circ \vA$, we have $\Vert \vM\Vert_\op \le \max\lbrace 1, \sup_{\vu \in \cU }\Vert \vu \Vert_{\ell_\infty}\rbrace (\Vert \vA_0 \Vert_\op  + \dots + \Vert \vA_{d_u} \Vert_\op )$. This concludes the proof.
\end{proof}

\begin{lemma}\label{lem:upper bound stability}
Suppose Assumption \ref{assump:stability} holds and 
let $(\vu_t)_{t \ge 0}$ be a sequence of inputs taking values in $\sqrt{d_u} \cdot\cS^{d_u-1}$, and let $F_t =  \begin{bmatrix}
     I_{d_y} &
    C   &
    C \,  (u_{t-1}\circ A )  &
\cdots &
    C \prod_{\ell = 1}^{\tau-2} (u_{t-\ell}\circ A ) 
 \end{bmatrix}$. We have 
\begin{align}\label{eq:stability first bound}
    \forall \; t \ge 0, \qquad \left\Vert \prod_{\ell = 0}^t (\vu_\ell \circ \vA )\right\Vert_\op \le \kappa \rho^{t+1}.
\end{align}
Consequently, we have \begin{align}\label{eq:stability second bound}
    \left\Vert F_t  \right\Vert_\op \le   1 + \frac{\kappa \Vert \vC \Vert_\op  }{1- \rho}
\end{align}
\end{lemma}
\begin{proof}
The first inequality \eqref{eq:stability first bound} holds immediately thanks to stability. The second inequality \eqref{eq:stability second bound} is an immediate consequence of \eqref{eq:stability first bound}. Indeed, we have 
 
\begin{align*}
        \Vert \vF_t \Vert_\op \le 1 + \Vert \vC\Vert_\op \sum_{\ell = 1}^{\tau-1} \left\Vert \prod_{i = 1}^{\ell -1} (\vu_{t- i} \circ A) \right\Vert_\op \le 1 + \Vert \vC\Vert_\op \sum_{\ell = 1}^{\tau-1} \kappa \rho^{\ell - 1} \le 1 + \frac{\kappa \Vert \vC \Vert_\op}{1- \rho}.  
\end{align*} 
This concluded the proof.
\end{proof}

\subsection{Properties of Hypercontractive Random Vectors} \label{app:moment_bounds}
\begin{lemma}[Bounded moment (weighted)]\label{lemma:hyper_to_moment_app}
     Let $Q, R \in \RR^{d_u \times d}$, and $\gamma >0$. Let $u$ be an isotropic and $(4,2,\gamma)$-hypercontractive random vector taking values in $\RR^{d_u}$. Then, we have
    \begin{align}
        \EE[(u^\top Q Q^\top u)(u^\top R R^\top u)] \leq \gamma \norm{Q}_F^2 \norm{R}_F^2 
    \end{align}
\end{lemma}

\begin{proof}
    Note that since $u$ is isotropic and $(4,2, \gamma)$-hypercontractive, therefore, for any $Q = [q_1~~q_2~~ \cdots~~q_d] \in \RR^{d_u \times d}$ and $R = [r_1~~r_2~~ \cdots~~r_d] \in \RR^{d_u \times d}$, we have 
\begin{align*}
    \EE[(u^\top Q Q^\top u)(u^\top R R^\top u)] & = \EE\left[\left( \sum_{i=1}^d (u^\top q_i)^2 \right) \left( \sum_{i=1}^d (u^\top r_i)^2 \right)  \right], \\
    & =  \sum_{1 \le i,j \le d} \EE\left[(u^\top q_i)^2(u^\top r_j)^2  \right], \\
    & \leqsym{a} \sum_{1 \le i,j \le d}\sqrt{\EE\left[ (u^\top q_i)^4   \right]} \sqrt{\EE\left[  (u^\top r_j)^4  \right]}, \\
    & \leqsym{b} \sum_{1 \le i,j \le d} \sqrt{\gamma} \EE\left[ (u^\top q_i)^2   \right] \sqrt{\gamma}\EE\left[  (u^\top r_j)^2  \right], \\
    & \le \gamma \sum_{1 \le i,j \le d}  \norm{q_i}_{\ell_2}^2  \norm{r_j}_{\ell_2}^2,  \\
    & \le \gamma  \norm{Q}_F^2 \norm{R}_F^2 ,  
\end{align*}
where we obtain (a) by using Cauchy-Schwarz inequality, and (b) from the $(4,2, \gamma)$-hypercontractivity assumption. This completes the proof. 
\end{proof}

\subsection{Proof of Lemma \ref{lem:one sided net} (Covering Argument)}\label{app:covering}

\begin{proof}
    Let $\mathcal{N}$ be a maximal $\epsilon$-separated subset of
    $\mathcal{S}^{d-1}$. Then $\mathcal{N}$ is an $\epsilon$-net, and the
    standard volumetric estimate yields
    \[
        |\mathcal{N}|
        \leq
        \left(1+\frac{2}{\epsilon}\right)^d.
    \]
    Fix a realization of $W$ such that $\lambda_{\min}(W)<\nu$. Since $W$
    is symmetric, there exists $u\in\mathcal{S}^{d-1}$ such that
    \[
        u^\top W u=\lambda_{\min}(W)<\nu.
    \]
    By the $\epsilon$-net property, we may choose $x\in\mathcal{N}$ such
    that $\|x-u\|_2\leq\epsilon$. Using
    \[
        x^\top W x-u^\top W u
        =(x-u)^\top W x+u^\top W(x-u),
    \]
    together with $\|x\|_2=\|u\|_2=1$, we obtain
    \begin{align*}
        \left|x^\top W x-u^\top W u\right|
        &\leq
        \|x-u\|_2\|Wx\|_2
        +\|u\|_2\|W(x-u)\|_2 \\
        &\leq
        2\epsilon\|W\|_{\op}.
    \end{align*}
    It follows that
    \[
        x^\top W x
        \leq
        u^\top W u+2\epsilon\|W\|_{\op}
        <
        \nu+2\epsilon\|W\|_{\op}.
    \]
    Consequently,
    \[
        \left\{\lambda_{\min}(W)<\nu\right\}
        \subseteq
        \bigcup_{x\in\mathcal{N}}
        \left\{
            x^\top W x
            <
            \nu+2\epsilon\|W\|_{\op}
        \right\}.
    \]
    Applying the union bound and the cardinality estimate for
    $\mathcal{N}$ gives
    \begin{align*}
        \PP\left(\lambda_{\min}(W)<\nu\right)
        &\leq
        \sum_{x\in\mathcal{N}}
        \PP\left(
            x^\top W x
            <
            \nu+2\epsilon\|W\|_{\op}
        \right) \\
        &\leq
        |\mathcal{N}|
        \max_{x\in\mathcal{N}}
        \PP\left(
            x^\top W x
            <
            \nu+2\epsilon\|W\|_{\op}
        \right) \\
        &\leq
        \left(1+\frac{2}{\epsilon}\right)^d
        \max_{x\in\mathcal{N}}
        \PP\left(
            x^\top W x
            <
            \nu+2\epsilon\|W\|_{\op}
        \right),
    \end{align*}
    as claimed.
\end{proof}
\subsection{Miscellaneous Lemmas \& Concentration Tools}

In this Appendix we present a set of lemmas and concentration inequalities that we persistently make use of in our proofs. First, we provide a version of Freedman's inequality which can also be deduced from Azuma-Hoeffding's inequality.  
Finally, we formalize the trick of a net arguments in the following lemma which is a classical argument that can be found in \cite{vershynin2010introduction}: 
\begin{lemma}[$\epsilon$-net argument]\label{lem:net argument}
    Let $W$ be a  $d_1 \times d_2$ random matrix and $\varepsilon \in (0,1/2)$. Let $\Mcal$  (resp. $\Ncal$) be an $\varepsilon$-net of $(\cS^{d_1 - 1}, \Vert \cdot \Vert_{\ell_2})$ (resp. $(\cS^{d_2 - 1}, \Vert \cdot \Vert_{\ell_2})$). For all $\rho>0$, it holds that  
    \begin{align*}
        \PP\left( \Vert W \Vert_\op  > \frac{\rho}{1-2\varepsilon}\right) \le \left(1 + \frac{2}{\varepsilon}\right)^{d_1+d_2}\max_{x \in \Mcal, y \in \Ncal} \PP( x^\T W y > \rho).
    \end{align*}
\end{lemma}
We omit the proofs of these Lemmas as they are standard results within the literature. 

\end{document}